\documentclass[a4paper,onecolumn,11pt,unpublished]{quantumarticle}
\pdfoutput=1
\usepackage[utf8]{inputenc}
\usepackage[english]{babel}
\usepackage[T1]{fontenc}
\usepackage{amsmath, mathtools}
\usepackage{hyperref}
\usepackage{mleftright}
\usepackage{amsfonts}

\usepackage{tikz}
\usepackage{bm}
\usetikzlibrary{arrows.meta}
\usetikzlibrary{decorations.markings,quantikz}
\usetikzlibrary{fit,shapes.arrows}
\usepackage{subcaption}

\usepackage{tikz}
\usepackage{lipsum}

\usepackage{amssymb}

\usepackage{physics}
\usepackage{multicol}

\usepackage{graphics, setspace}

\usepackage{mleftright}

\usepackage[inkscapelatex=false]{svg}

\newcommand{\sgeq}{\succcurlyeq}

\DeclareMathOperator{\m}{\mathcal{M}}
\DeclareMathOperator{\K}{\mathcal{K}}
\newcommand{\id}{\mathbbm{1}}
\newcommand{\pg}{P_{\textnormal{guess}}} 
\newcommand{\pgs}{P^{*}_{\textnormal{guess}}}
\newcommand{\hmin}{H_{\textnormal{min}}}
\DeclareMathOperator{\p}{\mathcal{P}}
\usepackage{amsthm}
\newcommand{\vecr}{\vec{r}}
\newcommand{\ve}{\varepsilon}
\newcommand{\conv}{\textnormal{conv}}

\newtheorem{theorem}{Theorem}

\newtheorem{corollary}{Corollary}[theorem]
\newtheorem{observation}{Observation}[theorem]
\usepackage{chngcntr}

\usepackage[numbers,sort&compress]{natbib}

\usepackage{changes}

\begin{document}

\title{Quantum randomness beyond projective measurements}

\author{Fionnuala Curran}
\affiliation{ICFO-Institut de Ci\`encies Fot\`oniques, The Barcelona Institute of Science and Technology,
Av. Carl Friedrich Gauss 3, 08860 Castelldefels (Barcelona), Spain}
\orcid{0000-0001-9036-8113}
\email{fionnuala.curran@icfo.eu}
\maketitle
\begin{abstract}
The unpredictability of quantum physics gives rise to intrinsic randomness. In an adversarial scenario, any additional degrees of freedom must be attributed to an eavesdropper with correlations to the measurement set-up. The \emph{true} randomness is then quantified by the probability that she correctly guesses the measurement outcomes, optimised over all possible strategies.
Extremal measurements are appealing here, since they do not allow information to leak to such an eavesdropper. Beyond projective measurements, however, a simple question remains open: how much intrinsic randomness can be generated by a given extremal measurement? In a step towards solving it, we characterise the randomness generated by any unbiased extremal rank-one measurement acting on any state, solving the problem explicitly in dimension two. Four-outcome qubit measurements of this type are tomographic, so these results hold for fully \emph{source}-device-dependent randomness too. The tetrahedral symmetric informationally complete (SIC) measurement, we find, has the \emph{least} intrinsic randomness within this class. We also present the \emph{skewed} SIC family of measurements, and use them to partially solve an open problem: we prove that $2 \log d$ bits of randomness, the maximal amount, can be generated device-dependently (or source-device-independently) in any dimension in which there exists a SIC measurement.  
\end{abstract}

\section{Introduction}
Quantum measurements generate an intrinsic form of randomness with no counterpart in classical physics. That the outcomes of a well-chosen measurement on a quantum system are unpredictable, even to an eavesdropper with perfect knowledge of the pre-measurement set-up, is exploited in the design of quantum random number generators \cite{Ma_2016, Herrero_Collantes_2017, Mannalatha_2023}. Some measurements are more random than others, however. Of particular interest are the \emph{extremal} measurements, which cannot be decomposed into a probabilistic mixture of other measurements \cite{D_Ariano_2005}, and are therefore secure against quantum adversaries \cite{Dai_2023, Senno_2023}. A question of practical, as well as theoretical, importance then arises: how much randomness can be generated using an extremal measurement, and which measurements are the most random?

When an extremal measurement acts on a generic (mixed) quantum state, we can quantify the resulting randomness by the probability with which a correlated eavesdropper correctly guesses the outcomes \cite{Law_2014}. The eavesdropper may have quantum correlations, such that she holds a purification of the state, or, equivalently, classical correlations, such that she knows which state from a probabilistic mixture is \emph{actually} prepared in a given measurement round. The guessing probability relates directly to the conditional min-entropy \cite{Konig_2009}, which lower bounds the number of (almost) perfectly random bits that can be extracted from the measurement outcomes\footnote{We use logarithms of base two throughout.}  \cite{renner2006thesis,Konig_2011}. Recent works have studied the maximal randomness that can be produced from any state using projective \cite{Meng_2023} or general \cite{anco2024securerandomnessquantumstate} measurements, as well as the maximal intrinsic randomness of several noisy (non-extremal) measurements \cite{curran2025maximalintrinsicrandomnessnoisy}.  
At the limits of randomness generation, it has been shown \cite{anco2024securerandomnessquantumstate} that no more than $2\log d$ bits of conditional min-entropy can be generated in dimension $d$, and that this bound can be approached arbitrarily closely.

We begin by introducing the class of \emph{skewed} SIC measurements, which are derived from symmetric informationally complete (SIC) measurements. These measurements are biased, in that they produce a non-uniform probability distribution from the maximally mixed `junk' state. While this bias makes them poor candidates for quantum state tomography, it allows them to produce a uniform distribution from a pure state instead. We thereby prove that one can generate maximal randomness from any pure or isotropically noisy state in any dimension in which there exists a SIC measurement, partially answering an open question in \cite{anco2024securerandomnessquantumstate}. Since the skewed SIC measurements are informationally complete, this result holds in a source-device-independent setting as well.
These measurements can be seen as higher-order analogues of a two-dimensional class that was previously shown to generate maximal device-independent randomness \cite{Woodhead_2020}, an area in which skewed SIC measurements may also find application. 

Turning to the more familiar \emph{unbiased} extremal measurements, we study the randomness they generate from any given state. This, we find, is proportional to the maximal quantum fidelity of the chosen state with any state that lives in the convex hull of the measurement projectors. Given a measurement with $n$ outcomes, by choosing the worst possible state, we derive a lower bound of $\log \frac{n}{d}$ bits of conditional min-entropy, which is relevant for source-device-independent randomness generation \cite{Avesani_2022}. Notably, in the extreme case where $n=d^2$, we can always reach at least $\log d$ bits of randomness, which is the maximum that can be achieved using a projective measurement. In pursuit of maximal randomness, we also introduce the one-parameter `scissors' family of unbiased measurements in dimensions two, three and four, which all come arbitrarily close to achieving $2 \log d$ bits of randomness. 

We then explicitly solve the \emph{maximal} intrinsic randomness of all unbiased extremal measurements in dimension two, as well as the randomness they generate from any state. On the way, we prove that all four-outcome measurements of this type are group covariant, and present a characterisation in terms of two angles, a result which may be of independent interest. Leveraging the tomographic nature of these measurements, one could adapt our results to randomness generation in a fully source-device-independent setting. A final result shows that the task of randomness generation can itself give an insight into the landscape of quantum measurements: among all four-outcome unbiased extremal qubit measurements, the SIC measurement, which is optimal for quantum state tomography \cite{Scott_2006}, is singled out as having the \emph{least} intrinsic randomness.

This paper is laid out as follows. In Section \ref{sec: setup}, we set up our problem and give the necessary preliminaries. In Section \ref{sec: biased}, we present the family of skewed SIC measurements. In Section \ref{sec: ub any dim}, we study unbiased extremal rank-one measurements in any dimension, while in Section \ref{sec: ub qubit}, we give our results on unbiased extremal qubit measurements. Finally, in Section \ref{sec: scissors}, we derive the scissors family of measurements in dimensions two, three and four. We present our conclusions and outlook in Section \ref{sec: conc}.

\section{Set-up}\label{sec: setup}
We work throughout in a Hilbert space of finite dimension $d$. A quantum system is represented by a state $\rho$, a positive semidefinite matrix of unit trace, while a measurement of $n$ outcomes is represented by a positive operator-valued measure (POVM) $\m=\{M_j\}_j$, where the positive semidefinite elements $\{M_j\}$ represent the outcomes. The probability of obtaining the outcome $j$ is given by the Born rule, $p_j = \tr \mleft( \rho M_j \mright)$. To ensure that the probabilities are non-negative and normalised for any state, the POVM elements must also satisfy $\sum_{j=1}^{n}M_j= \id$. We call a POVM \emph{unbiased} if all of its elements have equal trace, and biased otherwise\footnote{Note that the `unbiased' quality of a POVM has nothing to do with unbiased bases.}. Projective measurements are POVMs whose elements satisfy $M_i M_j = M_i \delta_{i, \, j}$ for all $i, \, j \in \{1, \, \hdots, \, n\}$; \emph{any} POVM can be realised, however, by applying a joint projective measurement to the target quantum state composed with some auxiliary state \cite{Peres_1990}.

On the boundary of the convex set of POVMs there live the \emph{extremal} measurements, which cannot be represented as a probabilistic decomposition of other measurements \cite{D_Ariano_2005}. These POVMs, of which the projective measurements are a special case, are constrained to have $n \leq d^2$ outcomes, and can be seen as the POVM analogue of pure quantum states. Unlike pure states, the elements of an extremal measurement may be mixed, but we will only concern ourselves with rank-one extremal POVMs in this work. Another distinguished class is that of \emph{informationally complete} (IC) measurements \cite{Prugovecki_1977, Busch_1991}, as one can uniquely reconstruct any state $\rho$ based solely on the probability distribution it generates from such a measurement. IC measurements span the entire space of $d \times d$ Hermitian matrices, so they must have at least $d^2$ outcomes. Those with $d^2$ elements are called \emph{minimal} informationally complete (MIC). Unbiased MICs, which can be seen as analogous to orthonormal bases in the larger operator space, are relevant for quantum tomography \cite{Caves_2002}, and can be used to generate certified randomness in any dimension \cite{farkas2024} (see \cite{debrota2021} for an overview of MICs). 

Bridging these two concepts, we know that any extremal POVM with $d^2$ outcomes is necessarily an MIC \cite{D_Ariano_2005}. The most famous MICs are undoubtedly the symmetric informationally complete (SIC) measurements \cite{Zauner1999, Renes_2004}, which are unbiased, with elements $M_j = \frac{1}{d} \Pi_j$ satisfying
\begin{equation}\label{eqn: SIC cond}
    \tr \mleft( \Pi_i \Pi_j \mright) = \frac{d \delta_{i, \, j} +1}{d+1} \qquad \textnormal{for all} \;\; i, \, j \in \{1, \, \hdots, \, d^2\}\,.
\end{equation}
SIC POVMs are conjectured, though not proven, to exist in every dimension. Solutions---some exact and some numerical---have been found in every dimension from two to 151, as well as in a handful of higher dimensions \cite{Fuchs_2017}. These POVMs show up in diverse areas, including quantum foundations (they are central to the project of quantum Bayesianism, or QBism \cite{Fuchs_2013}), quantum state tomography \cite{Scott_2006}, quantum key distribution \cite{Renes_2005,englert2008efficientrobustquantumkey,Bouchard_2018} and quantum randomness \cite{Acin_2016, Tavakoli_2021}. 

Quantum randomness is studied within different frameworks, depending on the assumptions one makes about the state and measuring devices. At one extreme sits device-independent randomness generation \cite{colbeck2011, Pironio_2010}, wherein randomness is certified by the measurement statistics alone, and no assumptions are made about the workings of the devices. At the other, we find trusted or device-\emph{dependent} randomness generation \cite{Law_2014}, wherein both devices are fully characterised. There also exist a range of intermediate, \emph{semi}-device-independent settings with fewer assumptions, including that of source-independent randomness generation, where the state is completely untrusted \cite{Cao_2016, Marangon_2017, Avesani_2018, Avesani_2022}. This work is based primarily on the device-dependent framework, since we assume that both the state $\rho$ and the POVM $\m$ are fully characterised, though our results on MICs apply equally to the source-device-independent setting.

It is known \cite{Dai_2023, Senno_2023} that the randomness generated by an extremal measurement does not depend on how that measurement is implemented. Given an extremal measurement $\m$ and a mixed state $\rho$, then, the optimal guessing probability of an eavesdropper is a maximisation over all the possible probabilistic decompositions $\{p_j, \, \rho_j\}$ of $\rho$, 
\begin{equation}\label{eqn: pguess}
    \pg \mleft( \rho, \, \m \mright) = \max_{ \{p_j, \, \rho_j\} } \; \sum_{j=1}^{n} p_j \tr \mleft( \rho_j M_j\mright) \,.
\end{equation}
The intrinsic randomness is then quantified by the conditional min-entropy \cite{Konig_2009}, as
\begin{equation}\label{eqn: hmin}
    \hmin \mleft( \rho, \, \m \mright) = - \log \pg \mleft( \rho, \, \m \mright)\,.
\end{equation}
We can find the \emph{maximal} intrinsic randomness \cite{anco2024securerandomnessquantumstate, curran2025maximalintrinsicrandomnessnoisy} of a state $\rho$ or an extremal measurement $\m$ by minimising \eqref{eqn: pguess} over all $\m$ or $\rho$,
\begin{equation}\label{eqn: max rand two}
    \pgs \mleft( \rho \mright) = \min_{\m } \, \pg \mleft( \rho, \, \m \mright)\,, \qquad \pgs \mleft( \m \mright) = \min_{\rho } \, \pg \mleft( \rho, \, \m \mright)\,.
\end{equation}
Note that the state that achieves the minimisation in $\pgs \mleft( \m \mright)$ is necessarily pure, so there is no decomposition involved here. The authors of \cite{anco2024securerandomnessquantumstate} 
derived the following lower bound for any $\rho$,
\begin{equation}\label{eqn: rand bound state}
    \pgs \mleft( \rho \mright) \geq \frac{1}{d^2} \left( \tr \sqrt{\rho} \right)^2\,.
\end{equation}
They showed that this bound can be achieved in dimension two, and constructed measurements that achieve a guessing probability arbitrarily close to \eqref{eqn: rand bound state} in every dimension. Note that, setting $\rho$ to be pure in \eqref{eqn: rand bound state}, we have 
\begin{equation}\label{eqn: rand bound pure state}
    \pgs \mleft( \rho \mright) \geq \frac{1}{d^2}  \,,
\end{equation}
such that, from \eqref{eqn: hmin}, 
\begin{equation}
    \hmin \mleft( \rho \mright) \leq 2 \log d
\end{equation}
is an upper bound for device-dependent randomness generation in any dimension, matching its device-independent counterpart \cite{Acin_2016}.

\section{Skewed SIC measurements}\label{sec: biased}
Here, we present a family of measurements, the \emph{skewed} SIC POVMs, that can generate maximal randomness from a pure or isotropically noisy state in any dimension in which there exists a SIC POVM. Let $\m= \{M_j\}_j$ be any symmetric informationally complete (SIC) measurement in dimension $d$, whose elements $M_j=\frac{1}{d}\Pi_j$ satisfy \eqref{eqn: SIC cond}.
Singling out the element $\Pi_1$ and following a prescription in \cite{Caves_2002}, in Appendix \ref{app: max rand meas} we derive a one-parameter family of biased MIC POVMs $\mathcal{N}=\{N_j\}_j$ with elements
\begin{equation}\label{eqn: max rand N main}
    N_{j} = \begin{cases}
        \gamma \Pi_{1}\,, \qquad & j =1\,,
        \\[5pt]
      \frac{1- \gamma}{d-1} \left(  \frac{\left(1- \Gamma \right)^2}{d+1} \Pi_1 + \Gamma \left( 1 - \Gamma \right) \left( \Pi_1 \Pi_j + \Pi_j \Pi_1 \right) + \Gamma^{2} \Pi_j \right)\,, \qquad & j \in \{2, \, \hdots, \, d^2\}\,,
    \end{cases}
\end{equation}
where
\begin{equation}
    \Gamma = \sqrt{ \frac{d-1}{d \left( 1 - \gamma  \right)} }\,, \qquad 0 < \gamma < 1\,.
\end{equation}
If we performed this measurement on the state $\rho = \Pi_1$, we would obtain the following probability distribution,
\begin{equation}
    p_j = \tr \mleft( \rho \,
    N_j\mright) =  
    \begin{cases}
     \gamma\,, \qquad & j=1\,,
     \\[5pt]
     \frac{1- \gamma}{d^2-1}\,, \qquad & j \in \{2, \, \hdots, \, d^2  \}\,.
    \end{cases}
\end{equation}
Choosing $\gamma=\frac{1}{d^2}$ gives a uniform distribution, so we can achieve the bound \eqref{eqn: rand bound pure state}. Note, too, that the probability of the obtaining outcome $1$ tends to one as $\gamma \rightarrow 1$ and to zero as $\gamma \rightarrow 0$. In dimension two, the POVM approaches the trine measurement in the latter case (see Figure \ref{fig: skewed} for a representation of all three cases for $d=2$). 
\begin{figure}
    \centering
    \begin{subfigure}{0.3\textwidth}
    \includegraphics[width=\textwidth]{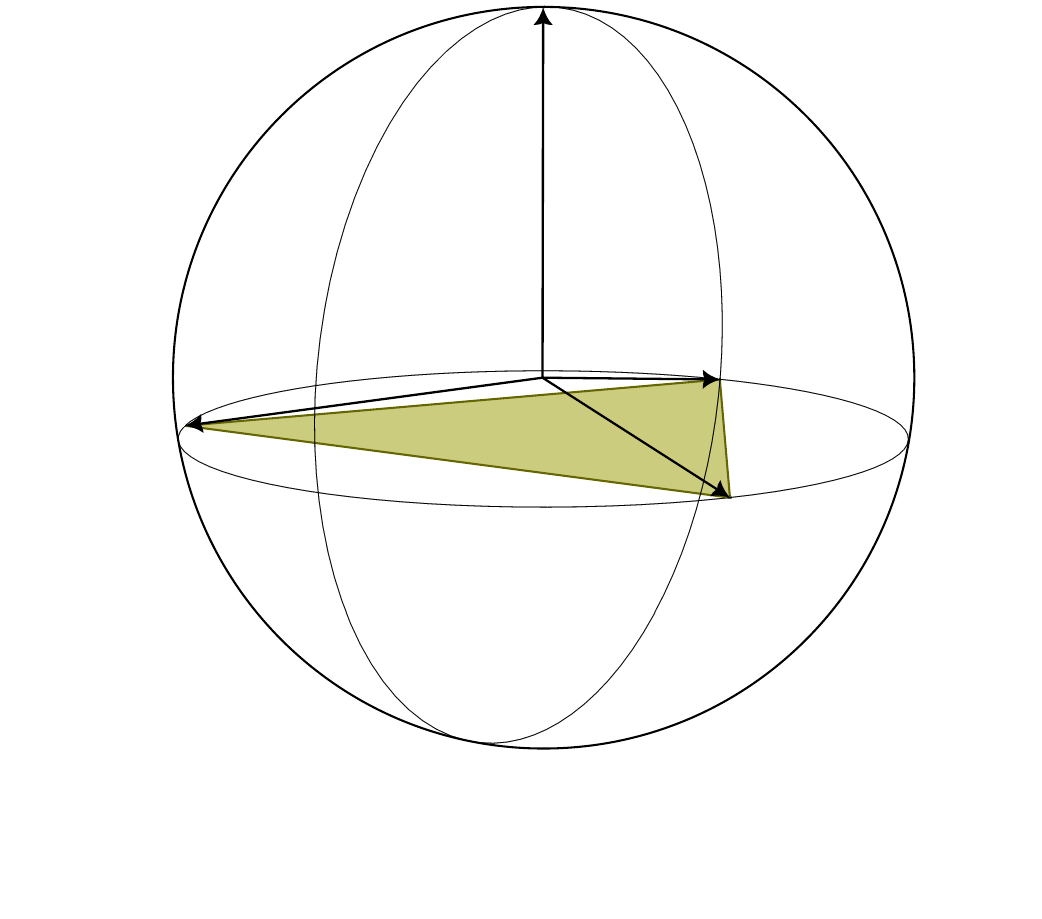}
    \caption{Maximally random, $\gamma = \frac{1}{4}$.}
    \label{fig: 4rand}
\end{subfigure}
\hfill
\begin{subfigure}{0.3\textwidth}
\includegraphics[width=\textwidth]{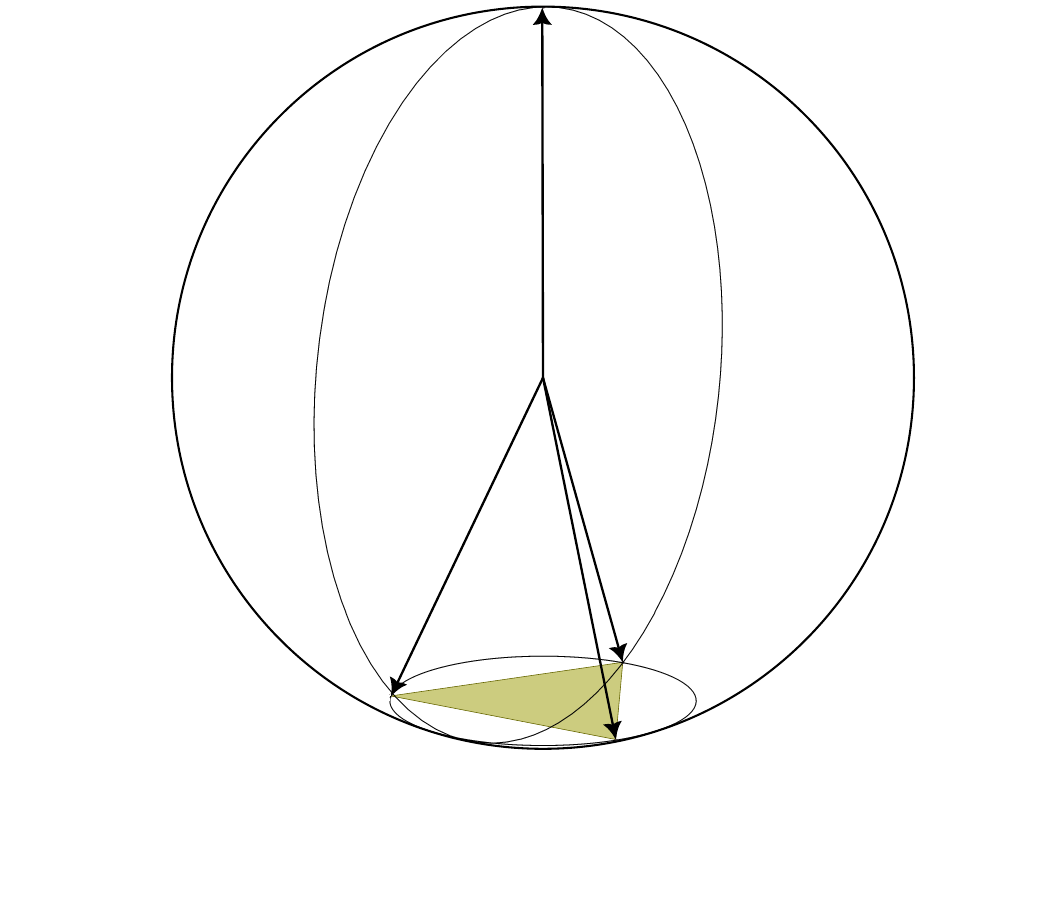}
    \caption{`Almost certain', $\gamma \rightarrow 1$.}
    \label{fig: 4proj}
\end{subfigure}
\hfill
\begin{subfigure}{0.3\textwidth}
\includegraphics[width=\textwidth]{Pics/4proj.pdf}
    \caption{`Almost trine', $\gamma \rightarrow 0$.}
    \label{fig: 4trine}
\end{subfigure}
    \caption{Skewed SIC POVMs in dimension two.}
    \label{fig: skewed}
\end{figure}

Now let's consider an isotropically noisy quantum state $\rho_\ve$ given by
\begin{equation}
    \rho_\ve  = \left( 1 - \varepsilon \right) \Pi_1 + \varepsilon\frac{\id}{d} \,, \qquad 0 < \ve \leq 1\,.
\end{equation}
Choosing
\begin{equation}\label{eqn: gamma par}
    \gamma =  \sqrt{\frac{d}{d - \varepsilon \left( d-1 \right)} }  \,\frac{\tr \sqrt{\rho_\ve}}{d^2}\, 
\end{equation}
gives us
\begin{equation}
    \tr \mleft( \sqrt{\rho_\ve} \, N_j\mright) = \frac{\tr \sqrt{\rho_\ve}}{d^2} \qquad \textnormal{for all} \;\; j \in \{1, \, \hdots, \, d^2\}\,.
\end{equation}
From \cite[Theorem B.4]{anco2024securerandomnessquantumstate}, this POVM lets us saturate the bound \eqref{eqn: rand bound state} and achieve maximal intrinsic randomness from $\rho_\ve$. In dimension two, for example, \eqref{eqn: gamma par} becomes 
\begin{equation}
    \gamma = \frac{1}{4} \left( 1 + \sqrt{\frac{\ve}{2-\ve}} \right)\,,
\end{equation}
and, choosing $\Pi_1= \ketbra{0}{0}$ and the standard representation of the tetrahedron SIC, we have
\begin{equation}
    N_1 = \begin{pmatrix}
        \gamma & 0
        \\
        0 & 0
    \end{pmatrix}\,, \qquad N_{j} = \frac{1}{3} 
    \begin{pmatrix}
    1 - \gamma & \sqrt{1- \gamma} \, \eta^{2 \left( j-2\right)} 
    \\
    \sqrt{1- \gamma} \, \eta^{ j-2} & 1
    \end{pmatrix} \qquad \textnormal{for all} \;\; j \in \{2, \, 3, \, 4\}\,,
\end{equation}
where $\eta= e^{\frac{2 \pi i}{3}}$.
This POVM appears in \cite[Example 4.7]{anco2024securerandomnessquantumstate}, where it was shown to generate maximal device-dependent randomness from any qubit state, while a similar construction was given in \cite{Woodhead_2020} which generates maximal device-\emph{independent} randomness from a partially entangled two-qubit state.

\section{Unbiased extremal measurements in any dimension}\label{sec: ub any dim}
Here, we consider unbiased extremal rank-one POVMs with $n$ outcomes, which have the form $\m= \{M_j\}_j$, with $M_j = \frac{d}{n} \ketbra{\psi_j}{\psi_j}$ and $d \leq n \leq d^2$. For a given POVM, we can construct an $n$-dimensional simplex $\p$ using the projectors $\{\ketbra{\psi_j}{\psi_j}\}$ as vertices,
\begin{equation}
    \p = \conv \mleft( \{\ketbra{\psi_j}{\psi_j}\}\mright) \,,
\end{equation}
where $\conv\mleft( \cdot \mright)$ denotes the convex hull of its arguments.
\begin{theorem}\label{thm: fidelity rand}
Let $\m = \{M_j\}_j$ be an unbiased extremal rank-one POVM with $n$ outcomes in dimension $d$, such that $M_j= \frac{d}{n} \ketbra{\psi_j}{\psi_j}$, let $\rho$ be a quantum state and let $\p$ be the convex hull of the set $\{\ketbra{\psi_j}{\psi_j}\}$. Then
\begin{equation}\label{eqn: fidelity rand main}
    \pg \mleft( \rho, \, \m \mright) = \frac{d}{n} \max_{\sigma \in \p} \; F \mleft( \rho, \, \sigma \mright)\,,
\end{equation}
where $F \mleft( \rho, \, \sigma\mright)$ is the quantum fidelity 
\begin{equation}
    F \mleft( \rho, \, \sigma\mright) = \left( \tr \sqrt{\sqrt{\rho} \, \sigma \sqrt{\rho}} \right)^2\,.
\end{equation}
\end{theorem}
The proof is given in Appendix \ref{ssec: fidelity proof}. Note that if we restrict to rank-one projective measurements, this theorem reduces to \cite[Theorem 1 (iii)]{Coles_2012}. If we choose $\rho \in \mathcal{P}$ in \eqref{eqn: fidelity rand main}, we get $\pg \mleft( \rho, \, \m \mright) = \frac{d}{n}$, as the fidelity of $\rho$ with itself is one, while if we choose $\rho \notin \mathcal{P}$, we find $\pg \mleft( \rho, \, \m \mright) < \frac{d}{n}$, as the fidelity must be strictly less than one. We can deduce several corollaries.
\begin{corollary}\label{cor: hmin}
    Let $\m$ be an unbiased extremal rank-one POVM with $n$ outcomes in dimension $d$. The conditional min-entropy generated by measuring any state $\rho$ with $\m$ is lower bounded by
    \begin{equation}\label{eqn: hmin lower bound}
        \hmin \mleft( \rho, \, \m\mright) \geq \log \frac{n}{d}\,.
    \end{equation}
\end{corollary}
\begin{proof}
    The result follows immediately from $\pg \mleft(\rho, \, \m \mright) \leq \frac{d}{n}$, using \eqref{eqn: hmin}.
\end{proof}
Corollary \ref{cor: hmin} generalises a lower bound from \cite{Avesani_2022} for the qubit trine measurement to all unbiased extremal rank-one measurements in any dimension, which includes all SIC measurements. Since it is state-independent, \eqref{eqn: hmin lower bound} also holds in a source-device-independent setting. 
\begin{corollary}\label{corr: 1/d^2 rand}
   Let $\m = \{M_j\}_j$ be an unbiased extremal POVM in dimension $d$. Then 
   \begin{equation}
       \pg \mleft( \rho, \, \m \mright) > \frac{1}{d^2} 
   \end{equation}
   for any state $\rho$.
\end{corollary}
\begin{proof}
    For any measurement with $n$ outcomes, the guessing probability is lower bounded by $\pg \mleft( \m \mright) \geq \frac{1}{n}$, so the bound \eqref{eqn: rand bound pure state} cannot be saturated by any POVM with $n < d^2$ outcomes. Our only candidates in the class of unbiased extremal measurements are those  with $n=d^2$ outcomes, which are necessarily rank one \cite[Corollary 6]{D_Ariano_2005}. In this case, the bound can be reached only if the probability distribution is uniform, i.e. if $\tr \mleft( \rho \, M_j \mright)= \frac{1}{d^2}$ for all $j \in \{1, \, \hdots, \, d^2\}$, or else the eavesdropper would choose a more probable outcome. However, since $\m$ is an unbiased MIC, the maximally mixed state $\rho = \frac{\id}{d}$ is the unique state that gives the uniform distribution, and since $\frac{\id}{d^2} \in \mathcal{P}$, from Theorem \ref{thm: fidelity rand} it must give $\pg \mleft( \rho, \, \m \mright)= \frac{1}{d}$, so the bound \eqref{eqn: rand bound pure state} is strict for all unbiased extremal POVMs.
\end{proof}

The following corollary holds in the special case where $n=d^2$, and is proven in Appendix \ref{ssec: fidelity proof}.
\begin{corollary}\label{corr: notd^2 rand}
    Let $\m = \{M_j\}_j$ be an unbiased extremal POVM with $d^2$ outcomes in dimension $d$, such that $M_j= \frac{1}{d} \ketbra{\psi_j}{\psi_j}$, and let $\mathcal{P}$ be the convex hull of the set $\{\ketbra{\psi_j}{\psi_j}\}$. When $\rho \notin \mathcal{P}$, the optimal decomposition $\{p_j , \, \rho_j\}$ for an eavesdropper contains at most $d^2-1$ nonzero elements.
\end{corollary}
The results of this section can be translated to a state discrimination scenario wherein a receiver tries to discriminate $n$ subnormalised states $\{\nu_j\}$ that are linearly independent on the operator space and satisfy  
\begin{equation}
    \tr \mleft( \nu^{-1} \, \nu_j \mright) = \frac{d}{n} \qquad \text{for all} \;\; j \in \{1, \, \hdots, \, n\}\,,
\end{equation}
where $\sum_{j=1}^{n}\nu_j=\nu$ and $d= \text{rank}\mleft(\nu \mright)$ (see Appendix \ref{app: SD}).

\section{Unbiased extremal qubit measurements}\label{sec: ub qubit}

\subsection{Description}
The set of extremal qubit measurements was characterised in \cite{D_Ariano_2005}. Here, we further characterise the subset of \emph{unbiased} extremal qubit POVMs, which may have two, three or four outcomes. We will use Bloch vector notation for the (possibly mixed) state $\rho$ and the normalised POVM elements $\{\ketbra{\psi_j}{\psi_j}\}$, 
\begin{equation}
  \rho = \frac{1}{2} \left( \id + \vecr \cdot \vec{\sigma}\right)\,, \qquad  \ketbra{\psi_j}{\psi_j}= \frac{1}{2} \left(  \id + \vecr_j \cdot \vec{\sigma}\right)\,, \qquad j \in \{1, \, \hdots, \, n\}\,,
\end{equation}
where $\vec{\sigma}= \left( X, \, Y, \, Z\right)$ is a vector of Pauli matrices and $n \in \{2, \, 3, \, 4\}$. The only extremal qubit measurements with two outcomes are the rank-one projective measurements, which are automatically unbiased, while the only one with three outcomes is the trine measurement, which is composed of three measurement vectors lying equidistant in a plane (see Figure \ref{fig: trine}).  
Unbiased extremal qubit POVMs with four outcomes have the form $\m=\{M_j\}_j$ with $M_j = \frac{1}{2}\ketbra{\psi_j}{\psi_j}$.
Up to a global rotation, their POVM elements can be described by the Bloch vectors
\begin{equation}\label{eqn: qubit UB MIC}
    \begin{aligned}
        \vecr_1 &= \big( -\cos \delta, \, -\sin \delta, \, 0 \big)\,, 
        \\
        \vecr_2 &= \big( -\cos \delta, \, \sin \delta, \, 0 \big)\,,
        \\
        \vecr_3 &= \big( \cos \delta, \, -\sin \delta \sin \alpha, \,  \sin \delta \cos \alpha\big)\,, 
        \\
        \vecr_4 &= \big( \cos \delta, \, \sin \delta \sin \alpha, \, - \sin \delta \cos \alpha \big)\,,
    \end{aligned}
\end{equation}
with $0 < \delta < \frac{\pi}{2}$ and $ - \frac{\pi}{2}< \alpha < \frac{\pi}{2}$. The measurement vectors $\{\vecr_j\}$ form a disphenoid, a tetrahedron with four congruent triangular faces (see Figure \ref{fig: MIC} for an example).
If we set $\alpha=0$ and $\delta = \arccos \frac{1}{\sqrt{3}}$, we recover the tetrahedron SIC POVM. Full details for this derivation are given in Appendix \ref{sssec: two angles}. 

\begin{figure}
    \centering
   \begin{subfigure}{0.45\textwidth}
    \includegraphics[width=\textwidth]{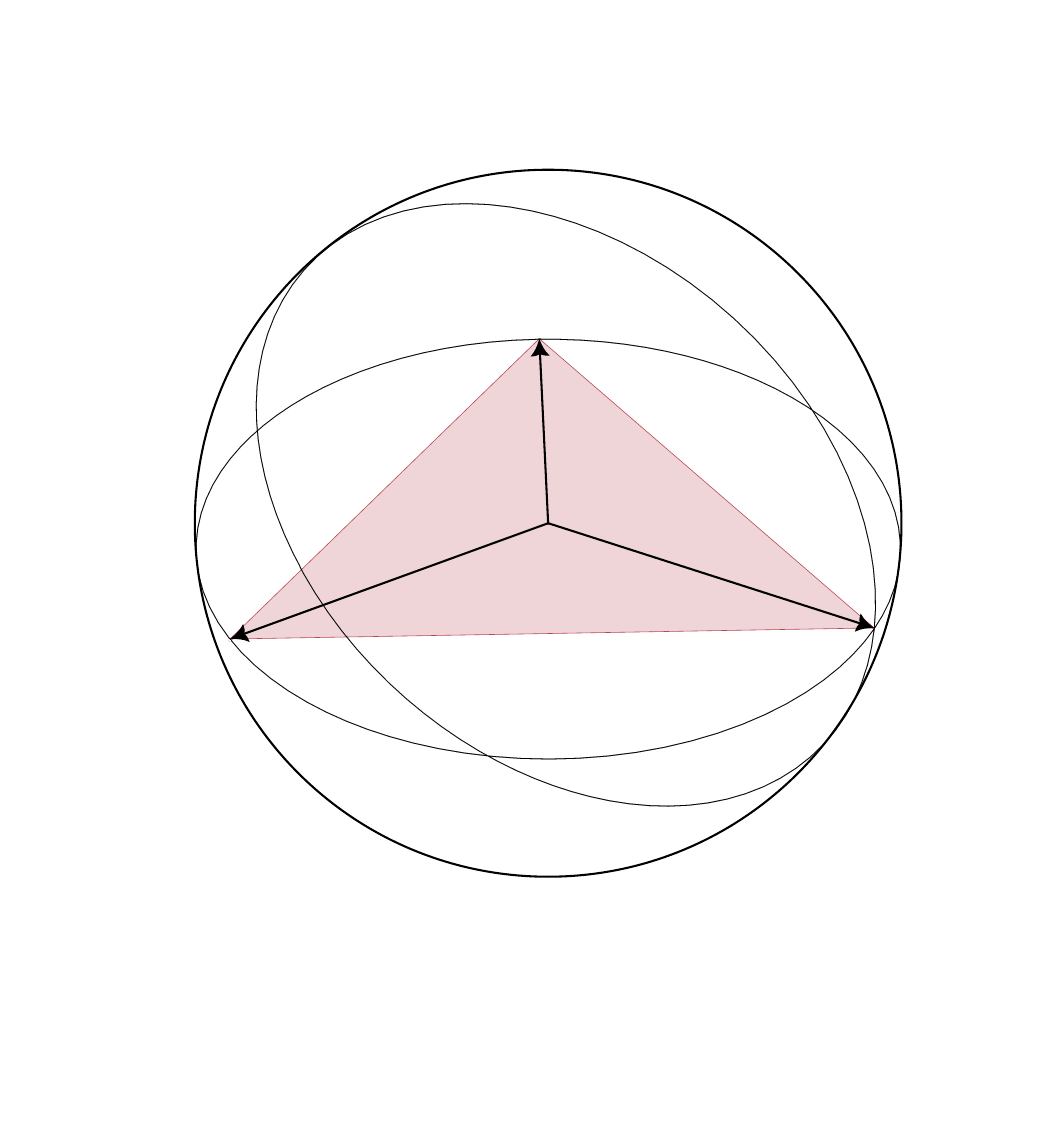}
    \caption{A trine measurement.}
    \label{fig: trine}
\end{subfigure}
\hfill
\begin{subfigure}{0.45\textwidth}
   \includegraphics[width=\textwidth]{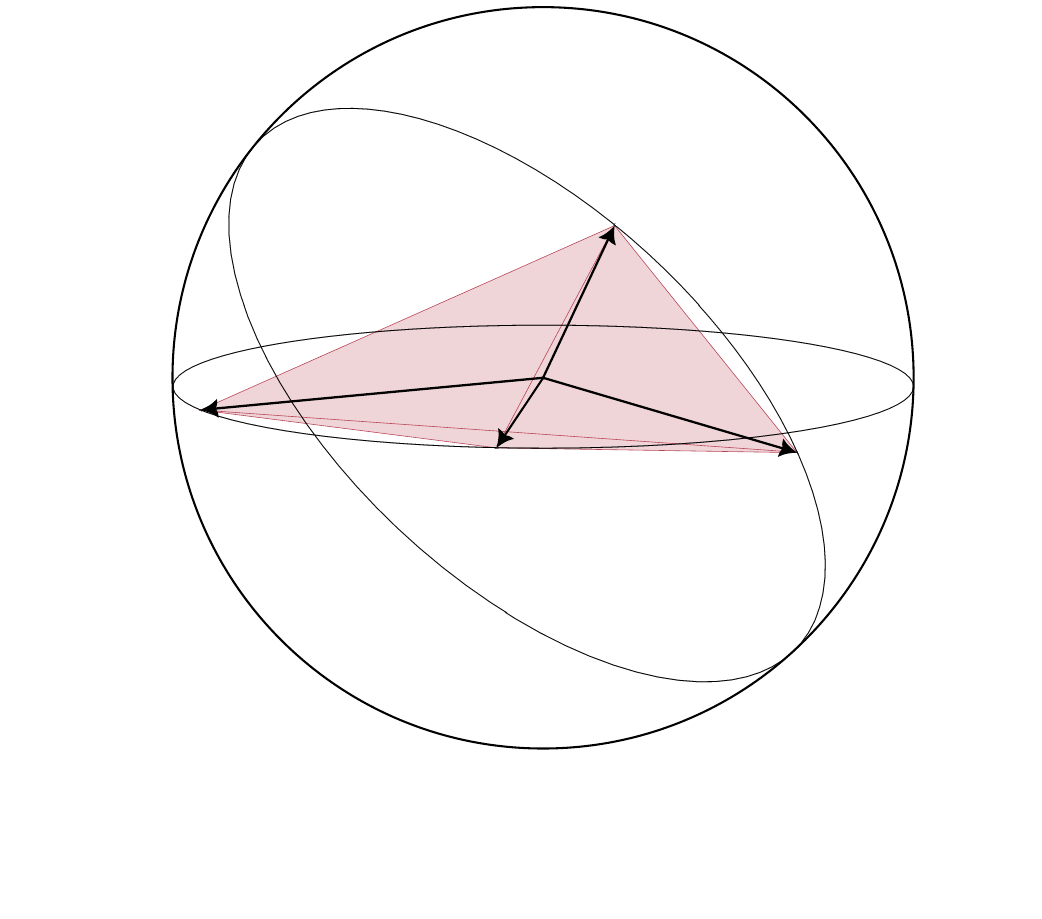}
   \caption{A general unbiased MIC.}
    \label{fig: MIC}
\end{subfigure}
    \caption{Unbiased extremal measurements in dimension two.}
    \label{fig: qubit POVMs}
\end{figure}

\begin{theorem}\label{thm: group cov}
  All unbiased qubit MICs are group covariant. 
\end{theorem}
\begin{proof}
We show in Appendix \ref{sssec: group cov} that the measurement vectors represented in \eqref{eqn: qubit UB MIC} can be generated by acting on a suitable fiducial vector with the following basis of unitary matrices, 
    \begin{equation}
 \left\{ \id, \, \; X, \, \; f \mleft( \alpha\mright) Y - g \mleft( \alpha\mright) Z, \, \; g \mleft( \alpha\mright) Y + f \mleft( \alpha \mright) Z   \right\}\,,
\end{equation}
where 
\begin{equation}
    f \mleft( \alpha \mright) = \frac{1}{\sqrt{2}} \left( \cos \frac{\alpha}{2} + \sin \frac{\alpha}{2} \right)\,, \qquad g \mleft( \alpha \mright) = \frac{1}{\sqrt{2}} \left( \cos \frac{\alpha}{2} - \sin \frac{\alpha}{2}\right)\,.
\end{equation}
\end{proof}

\subsection{Maximal randomness}
One can generate maximal randomness from a rank-one projective measurement by choosing any state that's unbiased to the measurement basis. For the trine measurement, there are only two optimal states, both perpendicular to the plane spanned by the three measurement vectors. In either case, we achieve a uniform probability distribution, such that 
\begin{equation}
    \pgs \mleft( \m \mright) = \frac{1}{n}\,, \qquad n = 2 \;\, \textnormal{or} \; \, 3\,.
\end{equation}
The situation is more complicated when $n=4$, as we know from Corollary \ref{corr: 1/d^2 rand} that $\pgs \mleft(\m \mright) > \frac{1}{4}$. Let's begin by labelling the four faces of the polytope $\mathcal{P}=\conv\mleft( \{\vecr_j\}\mright)$ formed by the vectors of the unbiased qubit MIC by $\left( 123\right)$, $\left( 124\right)$, $\left( 134\right)$ and $\left( 234\right)$. We show in Appendix \ref{ssec: qubit MIC max} that the unique states that generate maximal intrinsic randomness from such a POVM are the four pure states that are normal to the faces of $\mathcal{P}$, which we denote by $\vec{n}_{\left( 123\right)}$, $\vec{n}_{\left( 124\right)}$, $\vec{n}_{\left( 134\right)}$ and $\vec{n}_{\left( 234\right)}$. Choosing the state $\vec{n}_{\left( 123\right)}$, for example, the best approach for the eavesdropper is to deterministically guess any of the outcomes 1, 2 or 3. Note that
\begin{equation}\label{eqn: l main text}
    \vec{n}_{\left( 123\right)} \cdot \vecr_1 = \vec{n}_{\left( 123\right)} \cdot \vecr_2 = \vec{n}_{\left( 123\right)} \cdot \vecr_3 = l\,,
\end{equation}
where $l$ is the perpendicular distance from the centre of the Bloch ball to any of the faces of $\mathcal{P}$. Equivalent results hold if we choose any of the other normal vectors instead of $\vec{n}_{\left( 123\right)}$. The intrinsic randomness is then given by 
\begin{equation}\label{eqn: max qubit four}
    \pgs \mleft( \m \mright) = \frac{1}{4} \left( 1+ l\right)\,.
\end{equation}
Using the parameterization of unbiased qubit MIC POVMs given in \eqref{eqn: qubit UB MIC}, we can write $l$ as
\begin{equation}\label{eqn: l param main}
    l= \frac{\sin \delta \cos \delta \cos \alpha}{ \sqrt{4 \cos^2 \delta + \sin^2 \delta \cos^2 \alpha }}\,,
\end{equation}
such that 
\begin{equation}
    \pgs \mleft( \m \mright) = \frac{1}{4} \left( 1 + \frac{\sin \delta \cos \delta \cos \alpha}{ \sqrt{4 \cos^2 \delta + \sin^2 \delta \cos^2 \alpha }}\right)\,.
\end{equation}

\subsection{Randomness for any state}
Starting with $n=2$, let's consider a measurement vector $\vec{m}$ such that $\vecr_1=\vec{m}$ and $\vecr_2=- \vec{m}$. Without loss of generality, we can represent any mixed state $\vecr$ as
\begin{equation}\label{eqn: qubit par proj}
    \vecr = m \vec{m} + p \vec{n}\,, \qquad m^2 + p^2 <1\,,
\end{equation}
where $\vec{n}$ is a unit vector such that $\vec{m} \cdot \vec{n}=0$.
We show in Appendix \ref{ssec: 2 outcomes any} that
the optimal decomposition for the eavesdropper is into the states $\{\vec{q}_j\}$, where 
\begin{equation}
    \vec{q}_1 = \sqrt{1-p^2} \, \vec{m} + p \vec{n}\,, \qquad \vec{q}_2 = -\sqrt{1-p^2} \, \vec{m} + p \vec{n}\,.
\end{equation}
The optimal guessing probability is then 
\begin{equation}\label{eqn: qubit result}
    \pg \mleft( \vecr, \, \m \mright) = \frac{1}{2}\left( 1 + \sqrt{1-p^2} \right)\,. 
\end{equation}
When $m=0$ in \eqref{eqn: qubit par proj}, we return a result of \cite{Fiorentino_2007, Meng_2023}.

To solve the intrinsic randomness generated by measuring any mixed state $\vecr$ with the trine POVM, we begin by constructing a three-dimensional shape $\Delta$ in the Bloch ball. Let's introduce a unit vector $\vec{n}$ that's perpendicular to the plane in which the three measurement vectors $\{\vecr_j\}$ live. Without loss of generality, we will restrict ourselves to states for which
\begin{equation}
    p = \vecr \cdot \vec{n} \geq 0\,.
\end{equation}
From Theorem \ref{thm: fidelity rand}, we know that $\pg \mleft( \vecr, \, \m \mright)= \frac{1}{3}$ when $p=0$.
We can contract and translate the vectors $\{\vecr_j\}$ so they lie a distance $p$ along the vector $\vec{n}$, and still touch the inside of the Bloch sphere. The new vectors are
\begin{equation}
    \vec{s}_j = p \vec{n} + \sqrt{1-p^2} \vec{r}_j \qquad \textnormal{for all} \;\; j \in \{1, \, 2, \, 3\}\,.
\end{equation}
We can then define the shape $\Delta$ as the set
\begin{equation}
    \Delta = \conv \mleft( \{\vec{s}_j\}   \mright) \quad \text{for all } \quad 0 \leq p \leq 1\,. 
\end{equation}
We show in Appendix \ref{ssec: 3 outcomes any} that for any state inside $\Delta$, the optimal decomposition of the eavesdropper uses the three states $\{\vec{s}_j\}$ at the appropriate height $p$, such that 
\begin{equation}\label{eqn: trine result}
    \pg \mleft( \vecr, \, \m \mright) = \frac{1}{3} \left( 1 + \sqrt{1-p^2}\right) \qquad \textnormal{for all} \;\; \vecr \in \Delta\,.
\end{equation}
Note the similarity of \eqref{eqn: trine result} to the solution for projective measurements in \eqref{eqn: qubit result}.
We defer the solution for states $\vecr \notin \Delta$ to Appendix \ref{ssec: 3 outcomes any}, noting only that the eavesdropper's optimal solution is into two, rather than three, states. 
\begin{figure}
    \centering
   \begin{subfigure}{0.45\textwidth}
    \includegraphics[width=\textwidth]{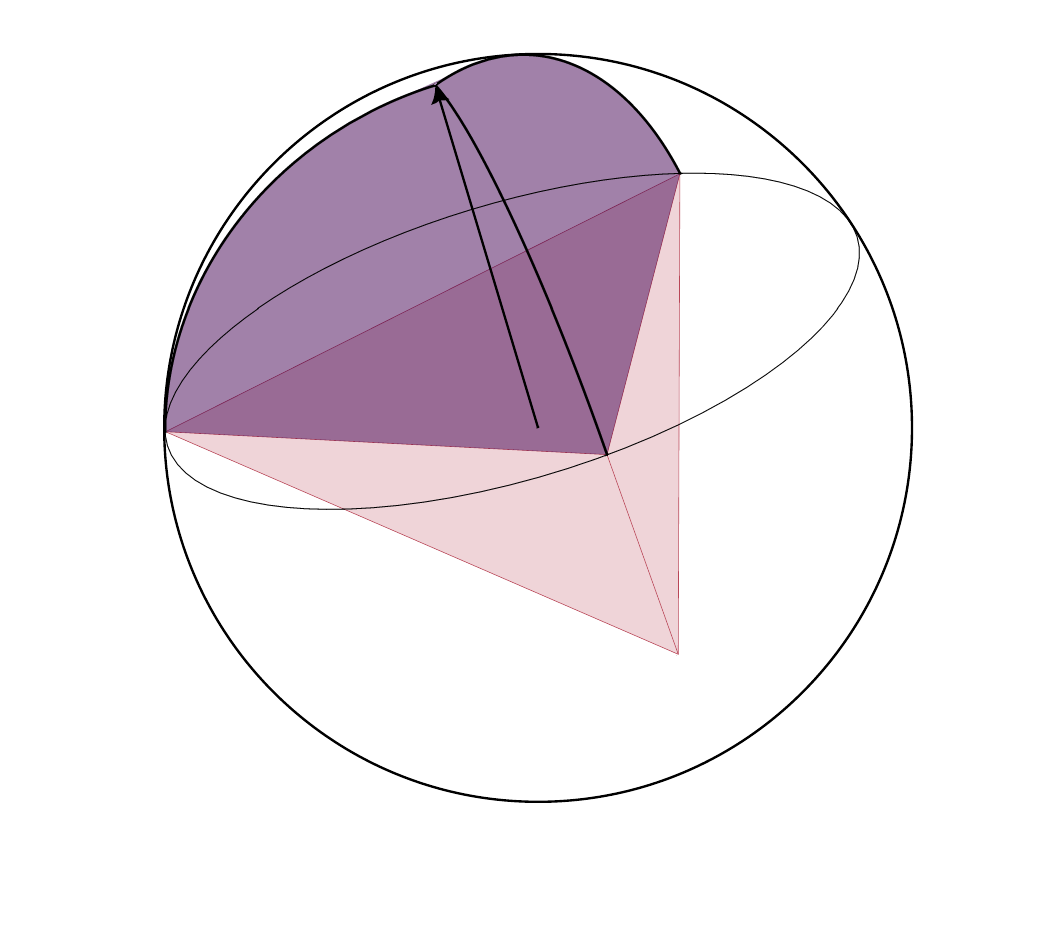}
    \caption{The shape $\Delta_{\left(123 \right)}$.}
    \label{fig: shape}
\end{subfigure}
\hfill
\begin{subfigure}{0.45\textwidth}
    \includegraphics[width=\textwidth]{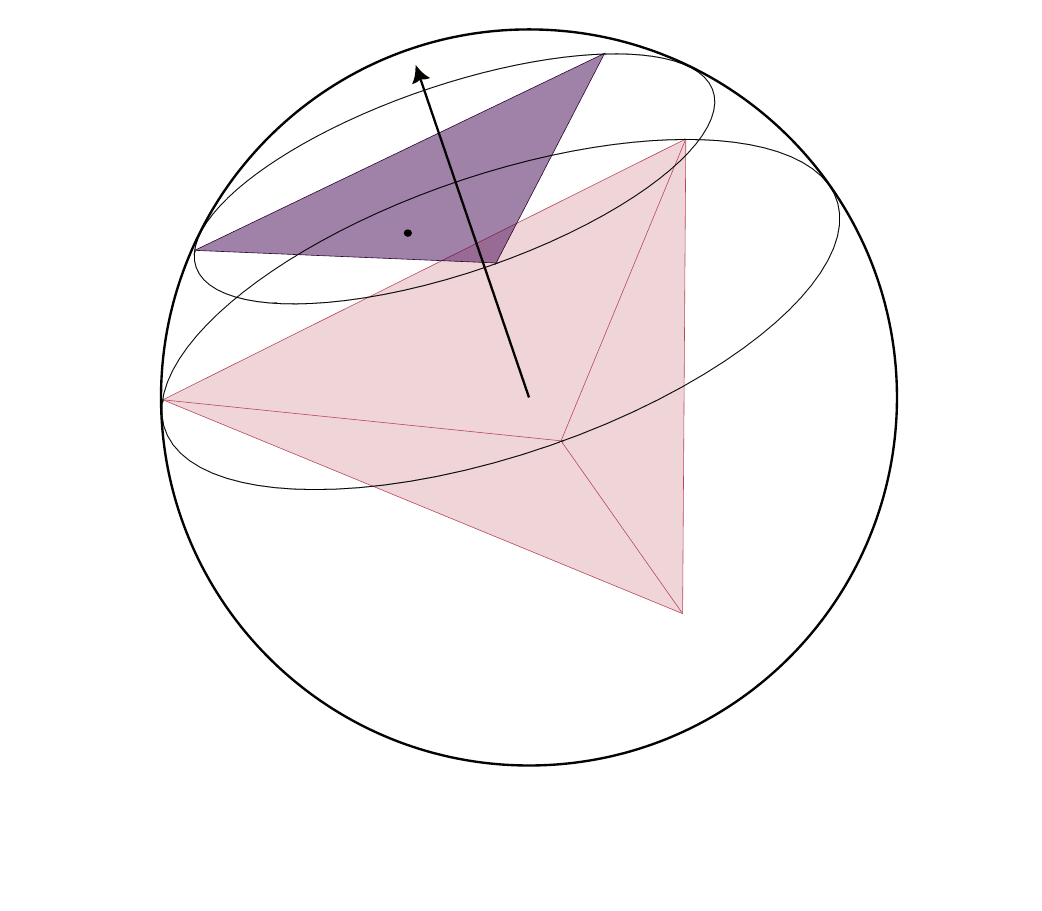}
    \caption{A state inside $\Delta_{\left(123 \right)}$.}
    \label{fig: insideT}
\end{subfigure}
    \caption{Regions of the Bloch ball for an unbiased qubit MIC.}
    \label{fig: dilations}
\end{figure}

We take a similar approach to solve the randomness generated from an unbiased qubit MIC. We solve it here in terms of the POVM vectors $\{\vecr_j\}$, without referring to the angle parameterization in \eqref{eqn: qubit UB MIC}. The relevant proofs are relegated to Appendix \ref{ssec: 4 outcomes any}. We know from Theorem \ref{thm: fidelity rand} that 
\begin{equation}
    \pg \mleft( \vecr, \,  \m \mright) = \frac{1}{2} \qquad \textnormal{for all} \;\; \vecr \in \mathcal{P}\,,
\end{equation}
so we need only concern ourselves with states outside the polytope. From Corollary \ref{corr: notd^2 rand}, we also know that the optimal decomposition for the eavesdropper will have at most three elements.
Without loss of generality, we will work in the region of the Bloch sphere defined by
\begin{equation}\label{eqn: MIC region main}
  p =  \vecr \cdot \vec{n}_{\left(123 \right)} \geq \vecr \cdot \vec{n}\,, \qquad \vec{n} \in \{ \vec{n}\}\,,
\end{equation}
where $\{\vec{n}\}= \{\vec{n}_{\left(123 \right)}, \,  \vec{n}_{\left(124 \right)} , \, \vec{n}_{\left(134 \right)}, \, \vec{n}_{\left(234 \right)}\}$ is the set of unit vectors normal to the faces of $\mathcal{P}$. 
We prove in Appendix \ref{ssec: 4 outcomes any} that every state $\vecr$ in our region \eqref{eqn: MIC region main} must satisfy $p \geq l$, where $l$ is given by \eqref{eqn: l main text}. Let's translate the face $\left( 123 \right)$ such that it lies at a distance $l \leq p \leq 1$ along the axis $\vec{n}_{\left( 123\right)}$, and contract it so all of the translated vectors touch the surface of the Bloch sphere. The result is the set of vertices $\{\vec{s}_j\}$, with 
\begin{equation}
    \vec{s}_j = \left( p -kl \right)\vec{n}_{\left( 123\right)} + k \vecr_j\,, \qquad k = \sqrt{\frac{1-p^2}{1-l^2} } \qquad \textnormal{for all} \;\; j \in \{1, \, 2, \, 3\}\,.
\end{equation}
We can construct a shape, which we call $\Delta_{\left( 123 \right)}$, defined by 
\begin{equation}\label{eqn: shape}
    \Delta_{\left( 123\right)} = \conv \mleft( \{\vec{s}_j\}   \mright) \qquad \text{for all } \;\; l \leq p \leq 1\,. 
\end{equation}
\begin{theorem}\label{thm: inside shape}
    For any unbiased qubit MIC with elements $\{\vecr_j\}$ and for any state $\vecr \in \Delta_{\left( 123 \right)}$, the guessing probability is 
     \begin{equation}
     \pg \left( \vecr, \, \m \right) =  \frac{1}{4} \left( 1 + pl + \sqrt{1-p^2} \sqrt{1-l^2} \right)\,,
 \end{equation}
 where $p = \vecr \cdot \vec{n}_{\left( 123 \right)}$.
\end{theorem}
This form differs from \eqref{eqn: qubit result} and \eqref{eqn: trine result}, since $l$ cannot be zero if $\m$ is extremal. To express the result in terms of the explicit parameterization of $\m$ in \eqref{eqn: qubit UB MIC}, we can simply sub in \eqref{eqn: l param main} for $l$. Again we defer the proof for states satisfying \eqref{eqn: MIC region main} and $\vecr \notin \Delta_{\left( 123\right)}$ to Appendix \ref{ssec: 4 outcomes any}, remarking only that the eavesdropper's optimal decomposition is into two states, rather than three. Figure \ref{fig: dilations} shows the shape $\Delta_{\left(123 \right)}$ and a state $\vecr$ that lives inside it. Note that these results on qubit MICs hold for source-device-independent randomness generation, where the state is untrusted, as we could reconstruct $\vecr$ straight from the measurement data.

\section{Scissors and SIC measurements}\label{sec: scissors}
\begin{figure}
    \centering
    \begin{subfigure}{0.45\textwidth}
    \includegraphics[width=\textwidth]{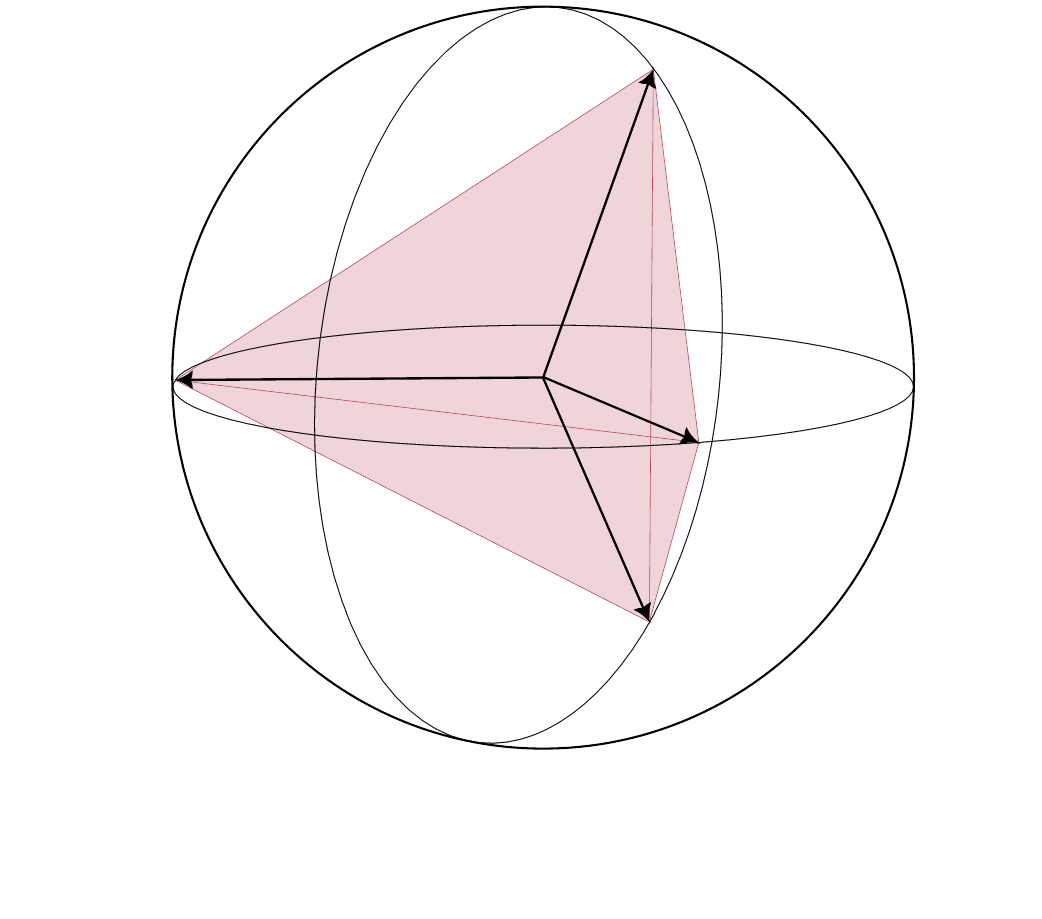}
    \caption{A SIC POVM.}
    \label{fig: SIC}
\end{subfigure}
\hfill
\begin{subfigure}{0.45\textwidth}
    \includegraphics[width=\textwidth]{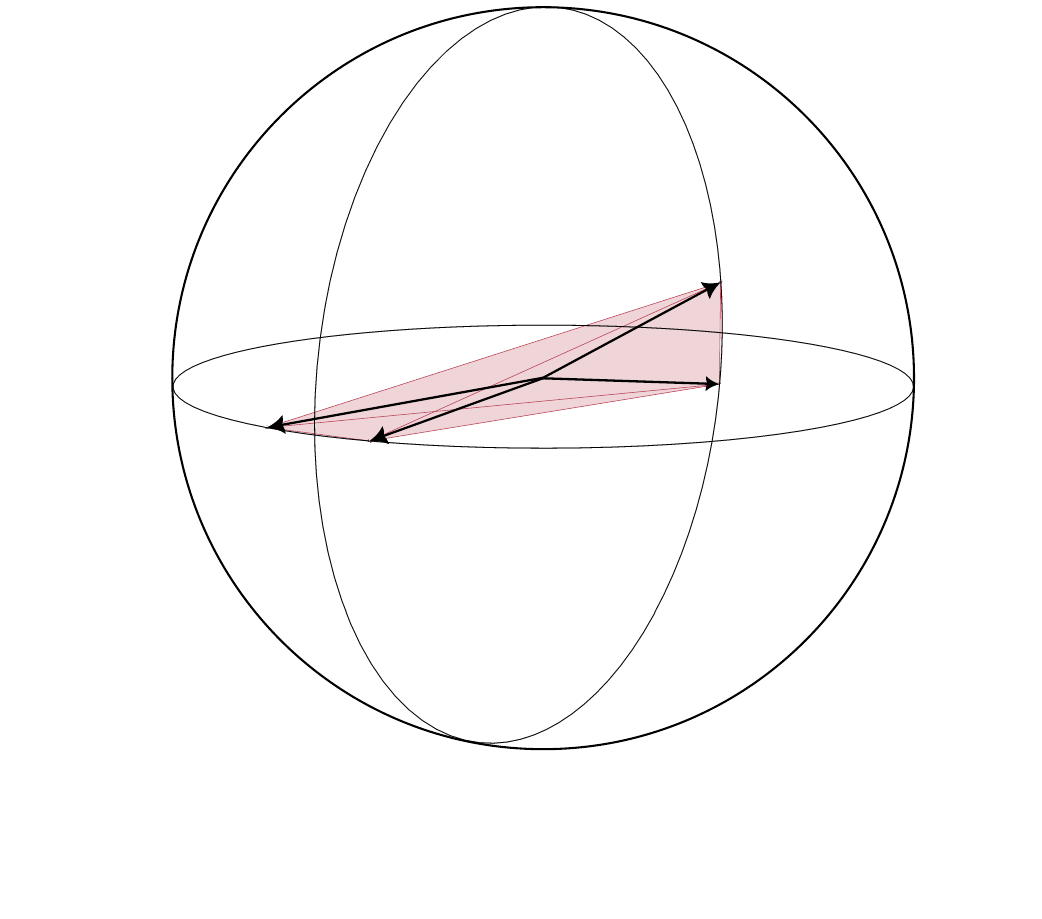}
    \caption{A general scissors POVM.}
    \label{fig: scissors}
\end{subfigure}
    \caption{Two scissors measurements in dimension two.}
    \label{fig: scissors_and_SIC}
\end{figure}
Let's return to the parameterization of unbiased qubit MICs given in \eqref{eqn: qubit UB MIC}. We show in Appendix \ref{ssec: scissors and SIC} that the squared area of any of the faces of the polytope $\mathcal{P}=\conv \mleft( \{\vecr_j\} \mright)$ is 
\begin{equation}
    A^2 = \sin^2 \delta \left( 4 \cos^2 \delta + \sin^2 \delta \cos^2 \alpha \right)\,.
\end{equation}
Since these POVMs are informationally complete, we can identify any state $\vecr$ by the probability distribution $\{p_j\}$ produced when it is measured by any such POVM. Instead of representing $\{p_j\}$ in standard Euclidean coordinates, we could represent it in baryonic coordinates on the Bloch sphere, choosing the states $\{\vecr_j\}$ as the basis (the resulting three-dimensional point is the average post-measurement state of $\vecr$). Applying this process to all states $\vecr$, we return an ellipsoid, whose properties depend on the POVM $\m$. In Appendix \ref{ssec: geometry}, we show that the volume of this ellipsoid is 
\begin{equation}
    \textnormal{Vol} = \frac{\pi}{3} \cos^2 \delta \sin^4 \delta \cos^2 \alpha\,.
\end{equation}
Interestingly, we find that the area $A$, the volume $\textnormal{Vol}$ and the perpendicular length $l$ from the origin to the faces of $\mathcal{P}$ are maximised exclusively by the SIC POVM, for which $\delta= \arccos \frac{1}{\sqrt{3}}$ and $\alpha=0$. This complements the known result \cite{debrota2021} that the Euclidean volume of the probability space of a SIC POVM in any dimension is strictly larger than that of any other MIC, be it unbiased or otherwise. 
Since we've shown in \eqref{eqn: max qubit four} that the maximal guessing probability of an unbiased qubit MIC depends only on the length $l$, the SIC is distinguished within this class as having the \emph{least} intrinsic randomness. 

We also consider a broader class of measurements, which we call \emph{scissors} measurements. Used in \cite{Woodhead_2020} to generate device-independent randomness, they are found by setting $\alpha=0$ in \eqref{eqn: qubit UB MIC}, so they're parameterized only by the angle $\delta$. Figure \ref{fig: scissors_and_SIC} shows the SIC POVM alongside a more general scissors measurement for comparison.
From \eqref{eqn: l param main}, the length $l$ is given by
\begin{equation}
l = \frac{\sin \delta \cos \delta}{ \sqrt{1 + 3 \cos^2 \delta }}  \,,  
\end{equation}
which tends to zero in the limit $\delta \rightarrow 0$.
If we consider a noisy state $\rho$ represented by $\vecr = p \vec{n}_{\left(123 \right)}$, Theorem \ref{thm: inside shape} tells us that the guessing probability tends towards
\begin{equation}
     \pg \mleft( \vecr, \, \m \mright) \rightarrow  \frac{1}{4} \left( 1 + \sqrt{1-p^2}  \right)\,,
 \end{equation}
such that, from \cite[Theorem B.4]{anco2024securerandomnessquantumstate}, we can come arbitrarily close to achieving maximal randomness from this state.

This one-parameter family of qubit measurements includes a SIC POVM and has a non-extremal measurement as its limit. We might ask if there exist families of scissors-like measurements in higher dimensions. The answer is yes, at least for dimensions three and four. In dimension three, we define the group of Weyl-Heisenberg operators
\begin{equation}
    D_{jk} = \eta^{\frac{jk}{2}} \sum_{m=0}^{2} \eta^{jm} \ketbra{k + m}{m}\,, \qquad \eta= e^{i \frac{2 \pi}{3}}\,,
\end{equation}
where addition is taken modulo $3$. Using \cite{D_Ariano_2004}, we show in Appendix \ref{ssec: qudit scissors} that we can define an unbiased MIC $\m=\{M_{jk}\}_{jk}$ with elements 
\begin{equation}
    M_{jk} = \frac{1}{3} \ketbra{\psi_{jk}}{\psi_{jk}}\,, \qquad \ket{\psi_{jk}} = D_{jk} \ket{\psi}  \qquad \textnormal{for all} \;\; j, \, k \in \{0, \, 1, \, 2\}
\end{equation}
using the fiducial state 
\begin{equation}
\ket{\psi} = \frac{1}{2\sqrt{3}} \left(  \left( 2 \cos \delta + \sqrt{2} \sin \delta\right) \left(\ket{0} + \ket{1} \right) + \left( 2 \cos \delta -2 \sqrt{2} \sin \delta\right) \ket{2}    \right)\,,   
\end{equation}
where $0 < \delta \leq \arccos \frac{1}{\sqrt{3}}$.
When $\cos \delta = \sqrt{\frac{2}{3}}$, we recover a SIC POVM \cite{Renes_2004}, while if we act on the state $\ket{0}$, the resulting probability distribution tends to the uniform one in the limit $\delta \rightarrow 0$. 
We defer the form of the scissors-like measurement in dimension four to Appendix \ref{ssec: qudit scissors}, only noting here that it is characterised by a single parameter $a$, which returns a SIC POVM from \cite{belovs2008, Appleby_2012} when $a = \sqrt{2 + \sqrt{5}}$ and gives arbitrarily close to uniform randomness when $a \rightarrow 1$. From Corollary \ref{corr: 1/d^2 rand}, these POVMs are the closest we can come to achieving maximal randomness without biasing the outcomes (e.g., by using a skewed SIC measurement); we leave the construction of scissors-like measurements in higher dimension for future work.

\section{Conclusions}\label{sec: conc}
We have investigated the intrinsic randomness of extremal rank-one quantum measurements in a device-dependent scenario. We focused particularly on the class of \emph{unbiased} extremal measurements, whose randomness we solved explicitly for any state in dimension two. In higher dimension $d$, we characterised, in terms of the quantum fidelity, the randomness generated by any unbiased extremal rank-one measurement of $n$ outcomes acting on any state, obtaining a state-independent lower bound of $\log \frac{n}{d}$ bits. We also proved that the maximal amount of intrinsic randomness, $ 2\log d$ bits, can be generated in any dimension $d$ in which there exists a SIC measurement. Our results for MICs apply equally to a source-device-dependent scenario, since these measurements reconstruct the untrusted quantum state.
Several open questions present themselves. One is whether the higher-dimensional measurements introduced in this work will, like their qubit analogues, be useful for device-independent randomness generation. Another question is whether SIC POVMs are distinguished in higher dimensions, like they are in dimension two, as having the least maximal intrinsic randomness among all the unbiased MICs. In future work, following \cite{curran2025maximalintrinsicrandomnessnoisy}, it would be interesting to explore how the randomness of the measurements considered here would change in the more realistic scenario where they were affected by noise.

\section*{Acknowledgments}
Thank you to Raja Yehia, Ranieri V. Nery and Antonio Acín for their feedback on this manuscript. I acknowledge funding
from the Government of Spain (Severo Ochoa CEX2019-
000910-S and TRANQI), Fundació Cellex, Fundació MirPuig, Generalitat de Catalunya (CERCA program), ERC
AdG CERQUTE and Ayuda PRE2022-101448 financiada por MCIN/AEI/ 10.13039/501100011033 y por el
FSE+.

\bibliographystyle{quantum}
\bibliography{refs2}

\appendix

\setcounter{theorem}{2}

\section{Optimisation problems}
\subsection{Vector optimisation}
Following \cite{Boyd2004}, we can write an optimisation problem in standard form as 
    \begin{equation}\label{eqn: opt stand}
    \begin{aligned}
         \underset{x}{\text{minimize}}  \, \qquad & f_0 \mleft( x \mright) && {}
         \\
          \text{subject to} \qquad & g_{i} \mleft( x \mright) = 0 \qquad &&  \textnormal{for all} \;\; i\,,
        \\
          \qquad & h_{j} \mleft( x \mright) \leq 0 \qquad &&  \textnormal{for all} \;\; j\,,
         \end{aligned}
\end{equation}
where the function $f_0$ returns a real value, and where the optimisation is over the intersection of the domains of the constraints $\{g_i\}$ and $\{h_j\}$. For our purposes, the variable $x$ is a real-valued vector of finite dimension $n$. The optimality, or KKT, conditions for the variable $x$ are 
\begin{align}
g_i \mleft( x \mright) &= 0\,, \qquad && \textnormal{for all} \;\; i\,,
    \\
    h_{j} \mleft( x \mright) & \leq 0\,, \qquad && \textnormal{for all} \;\; j\,, \label{eqn: KKT f ineq}
    \\
    \lambda_j & \geq 0 \,, \qquad && \textnormal{for all} \;\; j \,, \label{eqn: KKT lambda ineq}
    \\
    \lambda_j  h_j \mleft( x \mright) &= 0\,, \qquad && \textnormal{for all} \;\; j\,, \label{eqn: KKT comp slack}
    \\
    \nabla f_0 \mleft( x \mright) + \sum_{i} \nu_i \nabla g_i \mleft( x \mright) + \sum_{j} \lambda_j \nabla h_j \mleft( x \mright)  &= 0\,, \qquad && {} \label{eqn: KKT stat}
\end{align}
where we call $\{\nu_i\}$ and $\{\lambda_j\}$ the dual variables. We will call \eqref{eqn: KKT comp slack} the \emph{complementary slackness} condition and \eqref{eqn: KKT stat} the \emph{stationarity} condition. The dual variables are real, and the functions $\{g_i\}$ and $\{h_j\}$ return real values for all $i$ and $j$. The KKT conditions are necessary and sufficient when the primal problem is convex i.e., when all $\{g_i\}$ are affine and all $\{h_j\}$ are convex; otherwise, they are only necessary \cite{Boyd2004}.

\subsection{Semidefinite programs}
Following \cite{Skrzypczyk_2023, waltrous2018info},
we consider Hermitian operators $F$, which satisfy $F=F^{\dagger}$, and hermiticity-preserving maps $\Lambda \mleft(\cdot \mright)$, such that $\Lambda \mleft( F \mright)$ is Hermitian when $F$ is Hermitian. The adjoint map $\Lambda^{\dagger} \mleft( \cdot \mright)$ is the unique map satisfying
\begin{equation}
    \tr \mleft( \Lambda \mleft( X \mright) \, Y  \mright) = \tr \mleft( X \, \Lambda^{\dagger} \mleft( Y \mright)  \mright)\qquad \textnormal{for all} \;\; X, \, Y\,.   
\end{equation}
For a set of Hermitian operators $A$, $\{B_i\}$ and $\{C_j\}$ and hermiticity-preserving maps $\{\Phi_{i} \mleft( \cdot \mright) \}$ and $\{\Gamma_j \mleft( \cdot \mright)\}$, the primal and dual semidefinite programming problems (SDPs) are given by
\begin{align}\label{eqn: standard SDP}
  \begin{array}{rlcrl}
    \multicolumn{2}{c}{} & \qquad \quad  & 
      \multicolumn{2}{c}{} \vspace{0.1cm} \\
  \underset{  X   }{\text{maximize}}  &  \tr \mleft( A X \mright) & & \underset{ \{ Y_i \}\,, \,\{  Z_j     \}   }{\text{minimize}}    &  \underset{i}{\sum} \tr \mleft( Y_i B_i  \mright) - \underset{j}{\sum} \tr \mleft( Z_j C_j  \mright)  \vspace{0.1cm}
    \\
    \textrm{subject to} &  \Phi_i \mleft( X \mright) = B_i   \quad \text{for all} \;\; i   \vspace{0.1cm}  & & \textrm{subject to} &  \underset{j }{\sum  } \, \Gamma_j^{\dagger} \mleft( Z_j \mright) =   \underset{i }{\sum  } \Phi_i^{\dagger} \mleft( Y_i \mright) - A 
    \\
    & \Gamma_j \mleft( X \mright) \sgeq C_j \; \;\;\, \textnormal{for all} \;\; j \,    &&  &  Z_j \sgeq 0 \qquad \; \;\,  \textnormal{for all} \;\; j  \,.
  \end{array} 
\end{align}
The primal problem is strictly feasible if there exists a valid variable $X$ such that $\Gamma_j \mleft( X\mright)- C_j$ is \emph{strictly positive} for all $j$; similarly, the dual problem is strictly feasible if there exist valid variables $\{Y_i\}$ and $\{Z_j\}$ such that $Z_j$ is strictly positive for all $j$. From \cite[Conic Duality Theorem]{Nesterov1994}, if both the problems are strictly feasible then their solutions coincide, and we are free to use either the primal or dual formulations to solve our problem. The complementary slackness conditions, which are necessary for the optimal solution, are 
\begin{equation}\label{eqn: SDP comp slack}
    Z_j \left( \Gamma_j \mleft( X\mright) - C_j \right) = 0 \qquad \textnormal{for all} \;\; j \,.
\end{equation}
The 
intrinsic randomness of a (possibly mixed) state $\rho$ and an extremal POVM $\m=\{M_j\}_j$ with $n$ outcomes is given by
\begin{equation}\label{eqn: SDP rand}
\begin{aligned}
  \pg \mleft( \rho, \, \m \mright) = \;\;     \underset{\{ p_j, \, \rho_j \} }{\text{maximize}}  \, \quad & \sum_{j=1}^{n} p_j \tr \mleft( \rho_j M_j \mright)  && {}
        \\
         \text{subject to} \quad & \rho_j \sgeq 0 \quad \textnormal{for all} \;\; j \in \{1, \, \hdots, \, n\} &&   
         \\
          \quad & \sum_{j=1}^{n} p_j \rho_j = \rho \,, \quad &&  {}
    \end{aligned}
\end{equation}
where $\{p_j\}$ is a probability distribution. To cast this as a standard SDP of the form \eqref{eqn: standard SDP}, introduce
\begin{equation}
    \tilde{M} = \sum_{j=1}^{n} M_j \otimes \ketbra{j}{j}
\end{equation}
and a variable $\tilde{\sigma}$ such that
\begin{equation}
\begin{aligned}
  \pg \mleft( \rho, \, \m \mright) = \;\;     \underset{ \tilde{\sigma} }{\text{maximize}}  \, \quad & \tr \mleft( \tilde{\sigma} \; \tilde{M} \mright)  && {}
        \\
         \text{subject to} \quad & \tilde{\sigma} \sgeq 0 &&   
         \\
          \quad & \Phi \mleft( \tilde{\sigma}\mright) = \rho\,,  \quad &&  {}
    \end{aligned}
\end{equation}
where
\begin{equation}\label{eqn: phi eq}
  \Phi \mleft( \cdot \mright)  = \sum_{j=1}^{n} \left( \id \otimes \bra{j}\right) \cdot \left( \id \otimes \ket{j}\right)  \,.
\end{equation}
We can return to the original form \eqref{eqn: SDP rand} by choosing
\begin{equation}
    \tilde{\sigma} = \sum_{j=1}^{n} p_j \rho_j \otimes \ketbra{j}{j}
\end{equation}
for some valid decomposition $\{p_j, \, \rho_j\}$. Identifying the adjoint map $\Phi^{\dagger} \mleft( \cdot \mright) = \left( \cdot \right) \otimes \id$, we can write the dual in standard form as
\begin{equation}
\begin{aligned}
  \pg \mleft( \rho, \, \m \mright) = \;\;     \underset{ X, \, Z }{\text{minimize}}  \, \quad & \tr \mleft( X \rho \mright)  && {}
        \\
         \text{subject to} \quad & Z = X \otimes \id - \tilde{M} &&   
         \\
          \quad & Z \sgeq 0\,,  \quad &&  {}
    \end{aligned}
\end{equation}
or, more simply, as
\begin{equation}\label{eqn: rand dual}
\begin{aligned}
  \pg \mleft( \rho, \, \m \mright) = \;\;     \underset{ X }{\text{minimize}}  \, \quad & \tr \mleft( X \rho \mright)  && {}
        \\
         \text{subject to} \quad &X \sgeq M_j &&  \quad \textnormal{for all} \;\; j \in \{1, \, \hdots, \, n\}\,.  
    \end{aligned}
\end{equation}
Similarly, let $\{\nu_j\}$ be a known ensemble of $n$ quantum states $\{\nu_j\}$, subnormalised such that the probability of receiving the state labelled by $j$ is $p_j= \tr \nu_j$, and define
\begin{equation}
    \nu = \sum_{j=1}^{n} \nu_j\,.
\end{equation}
The probability of correctly discriminating the states is 
\begin{equation}\label{eqn: SDP SD}
\begin{aligned}
  \pg \mleft( \{\nu_j\} \mright) = \;\;     \underset{\{ E_j \}_j }{\text{maximize}}  \, \quad & \sum_{j=1}^{n}  \tr \mleft( \nu_j E_j \mright)  && {}
        \\
         \text{subject to} \quad & E_j \sgeq 0 \quad \textnormal{for all} \;\; j \in \{1, \, \hdots, \, n\} &&   
         \\
          \quad & \sum_{j=1}^{n} E_j = \id \,.\quad &&  {}
    \end{aligned}
\end{equation}
To cast this as an SDP in standard form \eqref{eqn: standard SDP}, we introduce the operator
\begin{equation}
    \tilde{\nu} = \sum_{j=1}^{n} \nu_j \otimes \ketbra{j}{j}\,
\end{equation}
and a variable $\tilde{E}$ such that 
\begin{equation}
\begin{aligned}
  \pg \mleft( \{\nu_j\} \mright) = \;\;     \underset{ \tilde{E} }{\text{maximize}}  \, \quad & \tr \mleft( \tilde{\nu} \; \tilde{E} \mright)  && {}
        \\
         \text{subject to} \quad & \tilde{E} \sgeq 0 &&   
         \\
          \quad & \Phi \mleft( \tilde{E}\mright) = \id \,,  \quad &&  {}
    \end{aligned}
\end{equation}
with $\Phi \mleft( \cdot \mright)$ given by \eqref{eqn: phi eq}. We can recover \eqref{eqn: SDP SD} by choosing 
\begin{equation}
    \tilde{E} = \sum_{j=1}^{n} E_j \otimes \ketbra{j}{j}
\end{equation}
for some valid POVM $\{E_j\}_j$. Using $\Phi^{\dagger} \mleft( \cdot \mright) = \left(\cdot \right) \otimes \id$ as before, the dual problem is
\begin{equation}
\begin{aligned}
  \pg \mleft( \{\nu_j\} \mright) = \;\;     \underset{ Y, \, Z }{\text{minimize}}  \, \quad & \tr  Y    && {}
        \\
         \text{subject to} \quad & Z = Y \otimes \id - \tilde{\nu} &&   
         \\
          \quad & Z \sgeq 0\,,  \quad &&  {}
    \end{aligned}
\end{equation}
or, more simply, 
\begin{equation}\label{eqn: SD dual}
\begin{aligned}
  \pg \mleft( \{\nu_j\} \mright) = \;\;     \underset{ Y }{\text{minimize}}  \, \quad & \tr  Y  && {}
        \\
         \text{subject to} \quad &Y \sgeq \nu_j &&  \quad \textnormal{for all} \;\; j \in \{1, \, \hdots, \, n\}\,.  
    \end{aligned}
\end{equation}

\section{Maximally random measurements}\label{app: max rand meas}

Let's start with any SIC POVM $\m= \{M_j\}_j$ in dimension $d$, where $M_j= \frac{1}{d} \ketbra{\psi_j}{\psi_j}$. For convenience, we use the notation $\Pi_j = \ketbra{\psi_j}{\psi_j}$ such that 
\begin{equation}
\tr \mleft( \Pi_i \Pi_j \mright) = \frac{d \delta_{i, \, j} +1}{  d+1} \qquad \textnormal{for all} \;\; i, \, j \in \{1, \, \hdots, \, d^2\}\,.
\end{equation}
The set $\{\Pi_j\}$ is a linearly independent basis for the operator space of $d \times d$ Hermitian matrices. We now define a parameter $0 < \gamma < 1$ and, singling out the element $\Pi_1$, we construct a new set of positive semidefinite linearly independent operators $\{A_j\}$ given by 
\begin{equation}
    A_j = \begin{cases}
        \gamma \Pi_1\,,  \qquad & j =1\,,
        \\[5pt]
        \frac{1- \gamma}{d-1} \Pi_j\,, \qquad & j \in \{2, \, \hdots, \, d^2 \} \,.
    \end{cases}
\end{equation}
These operators do not in general form a POVM as they do not sum to the identity (except when $\gamma= \frac{1}{d}$, which recovers the original SIC), but we can create a POVM $\mathcal{N}=\{N_j\}_j$ (see \cite[Section IV]{Caves_2002}) by defining 
\begin{equation}
    \Omega= \sum_{j=1}^{d^2} A_j
\end{equation}
and setting
\begin{equation}\label{eqn: POVM Nj}
    N_j = \Omega^{-\frac{1}{2}} A_j \Omega^{- \frac{1}{2}}\,.
\end{equation}
We can show by contradiction that the operator $\Omega$ is positive definite \cite{Caves_2002}. If it were not, there would exist at least one state $\ket{\psi}$ such that $\langle \psi | \Omega | \psi\rangle=0$. Since the operators $\{A_j\}$ are positive semidefinite, this would imply that $\langle \psi | A_j | \psi\rangle=0$ for all $j$, so $\ket{\psi}$ must live outside the support of $\{A_j\}$. This contradicts the assumption that the set $\{A_j\}$ spans the operator space, however, so $\Omega$ must have full rank, and we are free to invert it as in \eqref{eqn: POVM Nj}. Introducing the notation $\Pi_{\neq 1}= \id- \Pi_{1}$, we can write 
\begin{equation}
    \Omega = \sum_{j=1}^{d^2} A_j
 = \gamma \Pi_1 + \frac{1- \gamma}{d-1} \left( d \id - \Pi_1 \right) = \Pi_1 + \frac{d \left( 1- \gamma\right)}{d-1} \Pi_{\neq 1}\,.
 \end{equation}
 We can then solve $\Omega^{-\frac{1}{2}}$ explicitly as
 \begin{equation}
     \Omega^{- \frac{1}{2}} = \Pi_{1} + \sqrt{ \frac{d-1}{d \left( 1 - \gamma  \right)} } \Pi_{\neq 1} = \left(1 - \Gamma  \right)\Pi_{1} + \Gamma \id   \,,
 \end{equation}
 where 
 \begin{equation}
     \Gamma = \sqrt{ \frac{d-1}{d \left( 1 - \gamma  \right)} }\,.
 \end{equation}
The elements of the POVM $\mathcal{N}$ are then
\begin{equation}\label{eqn: max rand N}
    N_{j} = \begin{cases}
        \gamma \Pi_{1}\,, \qquad & j =1\,,
        \\[5pt]
      \frac{1- \gamma}{d-1} \left(  \frac{\left(1- \Gamma \right)^2}{d+1} \Pi_1 + \Gamma \left( 1 - \Gamma \right) \left( \Pi_1 \Pi_j + \Pi_j \Pi_1 \right) + \Gamma^{2} \Pi_j \right)\,, \qquad & j \in \{2, \, \hdots, \, d^2\}\,.
    \end{cases}
\end{equation}
If we performed $\mathcal{N}=\{N_j\}_j$ on the state $\rho = \Pi_1$, we would obtain the following probability distribution,
\begin{equation}
    p_j = \tr \mleft( \rho \,
    N_j\mright) =  
    \begin{cases}
     \gamma\,, \qquad & j=1\,,
     \\[5pt]
     \frac{1- \gamma}{d^2-1}\,, \qquad & j \in \{2, \, \hdots, \, d^2  \}\,.
    \end{cases}
\end{equation}
If we choose $\gamma=\frac{1}{d^2}$, the probability distribution is uniform, so we can achieve the bound \eqref{eqn: rand bound pure state}. Furthermore, imagine we have a state $\rho_\ve$ that's affected by an amount $0 < \varepsilon \leq 1$ of isotropic noise, i.e.,
\begin{equation}
    \rho_\ve  = \left( 1 - \varepsilon \right) \Pi_1 +\varepsilon \frac{\id}{d} \,.
\end{equation}
We have 
\begin{equation}
    \sqrt{\rho_\ve} = \sqrt{ \frac{d- \varepsilon \left( d-1 \right)}{d} } \Pi_1 + \sqrt{\frac{\varepsilon}{d}} \Pi_{\neq 1}\,,
\end{equation}
such that
\begin{equation}
    \tr \mleft(  \sqrt{\rho_\ve} \, N_j \mright) = 
    \begin{cases}
    \gamma \sqrt{\frac{d - \varepsilon \left( d-1\right)}{d} }\,, \qquad & j=1\,,
    \\[5pt]
     \frac{1}{d^2-1} \left( \tr \sqrt{\rho_\ve} - \gamma \sqrt{\frac{d- \varepsilon \left( d-1\right)}{d} } \right)\,, \qquad & j \in \{2, \, \hdots, \, d^2\}\,.
    \end{cases}
\end{equation}
If we choose 
\begin{equation}
    \gamma = \frac{1}{d^2} \left( 1+ \left( d-1\right) \sqrt{ \frac{\varepsilon}{d - \varepsilon \left( d-1 \right)} } \right) = \sqrt{\frac{d}{d - \varepsilon \left( d-1 \right)} }  \,\frac{\tr \sqrt{\rho_\ve}}{d^2}\,, 
\end{equation}
we have 
\begin{equation}
    \tr \mleft( \sqrt{\rho_\ve} \, N_j\mright) = \frac{\tr \sqrt{\rho_\ve}}{d^2} \qquad \textnormal{for all} \;\; j \in \{1, \, \hdots, \, d^2\}\,.
\end{equation}
From \cite[Theorem B.4]{anco2024securerandomnessquantumstate}, then, we can saturate the bound \eqref{eqn: rand bound state} and achieve maximal intrinsic randomness from $\rho_\ve$. When $\rho_\ve$ is maximally mixed, $\varepsilon=1$, the optimal measurement is the original SIC POVM.

We can find other interesting biased POVMs by choosing different values of $\gamma$. When we approach the limit $\gamma \rightarrow 1$, the probability distribution using the state $\rho=\Pi_1$ tends towards 
\begin{equation}
    p_j = \tr \mleft( \rho \,
    N_j\mright) \rightarrow 
    \begin{cases}
     1\,, \qquad & j = 1\,,
     \\[5pt]
     0\,, \qquad & j \in \{2, \, \hdots, \, d^2\}\,,
    \end{cases}
\end{equation}
so, although we can never achieve complete certainty in any of the outcomes of an informationally complete measurement (see e.g. \cite[Corollary 3]{debrota2021}), we can come arbitrarily close. Similarly, taking the limit $\gamma \rightarrow 0$ with the state $\rho=\Pi_1$, the probability distribution tends towards  
\begin{equation}
    p_j = \tr \mleft( \rho \,
    N_j\mright) \rightarrow 
    \begin{cases}
     0\,, \qquad & j=1\,,
     \\[5pt]
     \frac{1}{d^2-1}\,, \qquad & j \in \{2, \, \hdots, d^2\}\,.
    \end{cases}
\end{equation}
In dimension two, the limit of the POVM as $\gamma \rightarrow 0$ is the trine measurement.

\section{Unbiased extremal measurements in any dimension}

\subsection{Randomness}\label{ssec: fidelity proof}
\setcounter{theorem}{0}
\begin{theorem}
Let $\m = \{M_j\}_j$ be an unbiased extremal rank-one POVM with $n$ outcomes in dimension $d$, such that $M_j= \frac{d}{n} \ketbra{\psi_j}{\psi_j}$, let $\rho$ be a quantum state and let $\p$ be the convex hull of the set $\{\ketbra{\psi_j}{\psi_j}\}$. Then
\begin{equation}\label{eqn: fidelity rand}
    \pg \mleft( \rho, \, \m \mright) = \frac{d}{n} \max_{\sigma \in \p} \; F \mleft( \rho, \, \sigma \mright)\,,
\end{equation}
where $F \mleft( \rho, \, \sigma\mright)$ is the quantum fidelity 
\begin{equation}
    F \mleft( \rho, \, \sigma\mright) = \left( \tr \sqrt{\sqrt{\rho} \, \sigma \sqrt{\rho}} \right)^2\,.
\end{equation}
\end{theorem}
\begin{proof}
    From \cite[Theorem 1 (iii)]{Coles_2012}, the optimal guessing probability for an $n$-outcome projective measurement $\Pi=\{\Pi_{j}\}_j$ and any state $\rho$ is 
\begin{equation}\label{eqn: og fid}
        \pg \mleft( \rho, \, \Pi\mright) = \max_{\sigma} \; F \mleft( \rho, \, \sum_{j=1}^{n} \Pi_{j}\, \sigma \, \Pi_j  \mright)\,,
    \end{equation}
    where the maximisation is over all quantum states $\sigma$. Any $n$-outcome \emph{non}-projective measurement $\m= \{M_j\}_j$ on a state $\rho_S$ can be implemented using a pure auxiliary state $\ket{1}_A$ and the global projective measurement
\begin{equation}
\Pi_{j, \, SA} = U^{\dagger} \left( \id_S \otimes \ketbra{j}{j}_A \right) U \qquad \textnormal{for all} \;\; j \in \{1, \, \hdots, \, n\}\,,    
\end{equation}
where $U$ is a unitary such that    
\begin{equation}
    U \ket{\psi, \, 1}_{SA}= \sum_{j=1}^{n} \sqrt{M_{j}} \ket{\psi, \, j} \quad \text{for all} \;\; \ket{\psi}\,.
\end{equation}
Since the POVMs we consider are extremal, the guessing probability is not affected by the choice of Naimark dilation, so we have
\begin{equation}
    \pg \mleft(\rho, \, \m \mright) = \pg \mleft( \rho_S \otimes \ketbra{1}{1}_A, \, \Pi_{SA} \mright)\,.
\end{equation}
Note the following property of the fidelity for any state $\mu_{SA}$,
\begin{align}
    F \mleft( \rho_S \otimes \ketbra{1}{1}_A, \, \mu_{SA} \mright) &= \left( \tr \sqrt{ \left( \sqrt{\rho}_S \otimes \ketbra{1}{1}_{A} \right) \, \mu_{SA} \, \left( \sqrt{\rho}_S \otimes \ketbra{1}{1}_{A} \right)  } \right)^2 
    \\
    &= \left( \tr \sqrt{ \sqrt{\rho}_S \, \tr_A \left( \id_{S} \otimes \ketbra{1}{1}_A \, \mu_{SA}\right) \sqrt{\rho}_S } \right)^2\,.
\end{align}
Our guessing probability is then reduced to 
\begin{equation}
\begin{aligned}
    \pg \left( \rho, \, \m\right) &= \max_{\sigma_{SA}} F \mleft( \rho_S \otimes \ketbra{1}{1}_A, \ \sum_{j=1}^{n} \Pi_{j, \, SA} \, \sigma_{SA} \, \Pi_{j, \, SA}\mright)
    \\[5pt]
    &= \max_{\sigma_{SA}} \left( \tr \sqrt{ \sqrt{\rho}_S \, \sum_{j=1}^{n} \tr_A \mleft( \id_{S} \otimes \ketbra{1}{1}_A \,  \Pi_{j, \, SA} \, \sigma_{SA} \, \Pi_{j, \, SA} \,\mright) \sqrt{\rho}_S } \right)^2
    \\[5pt]
    &= \max_{\sigma_{SA}} \left( \tr \sqrt{ \sqrt{\rho}_S \, \sum_{j=1}^{n}  \sqrt{M_j} \tr_A \mleft( \left( \id_{S} \otimes \ketbra{1}{j}\right)  U \sigma_{SA} U^{\dagger}  \left(  \id_{S} \otimes \ketbra{j}{1} \right) \mright) \sqrt{M_j} \, \sqrt{\rho}_S } \right)^2
    \\[5pt]
    &= \max_{\sigma_{SA}} \frac{d}{n}\left( \tr \sqrt{ \sqrt{\rho}_S \, \underbrace{\sum_{j=1}^{n}  \tr \mleft( \ketbra{\psi_j}{\psi_j}_S \otimes \ketbra{j}{j}_A \; U \sigma_{SA}U^{\dagger}\mright) \ketbra{\psi_j}{\psi_j}_S}_{\left( * \right)} \, \sqrt{\rho}_S } \right)^2\,, \label{eqn: lastline}
\end{aligned}
\end{equation}
where in the final equality we use $M_j= \frac{d}{n}\ketbra{\psi_j}{\psi_j}$. We can see that the operator $\left( * \right)$ is proportional to some state in $\p$ up to a positive constant. Since $\{\ketbra{\psi_j}{\psi_j}_S \otimes \ketbra{j}{j}_A\}$ is a set of orthogonal projectors, we have 
\begin{equation}
    \sum_{j=1}^{n} \ketbra{\psi_j}{\psi_j}_S \otimes \ketbra{j}{j}_A \leq \id_{SA}\,, 
\end{equation}
so the trace of $\left( * \right)$ is upper bounded by one,
\begin{equation}
  \tr \mleft( * \mright) =   \tr \mleft( \left(\sum_{j=1}^{n} \ketbra{\psi_j}{\psi_j}_S \otimes \ketbra{j}{j}_A\right) \; U \sigma_{SA}U^{\dagger}\mright) \leq \tr \mleft( U \sigma_{SA} U^{\dagger}\mright)=1\,. 
\end{equation}
We can also show that the set of normalised states $\left( * \right)$ fills the whole set $\p$, by considering, for any probability distribution $\{p_j\}$, the state
\begin{equation}
    U \sigma_{SA} U^{\dagger} = \sum_{j=1}^{n} p_j \ketbra{\psi_j}{\psi_j}_S \otimes \ketbra{j}{j}_A\,,
\end{equation}
which results in $\left( *\right)= \sum_{j=1}^{n} p_j \ketbra{\psi_j}{\psi_j}_S$. We are then free to replace $\left( *\right)$ and the optimisation over $\sigma_{SA}$ in \eqref{eqn: lastline} with an optimisation over states $\sigma \in \p$,
\begin{equation}
\begin{aligned}
\pg \mleft( \rho, \, \m\mright) &=  \frac{d}{n} \max_{\sigma_{SA}} \mleft( \tr \sqrt{ \sqrt{\rho}_S \, \underbrace{\sum_{j=1}^{n}  \tr \mleft( \ketbra{\psi_j}{\psi_j}_S \otimes \ketbra{j}{j}_A \; U \sigma_{SA}U^{\dagger}\mright) \ketbra{\psi_j}{\psi_j}_S}_{\left( * \right)} \, \sqrt{\rho}_S } \mright)^2
\\[5pt]
&= \frac{d}{n} \max_{\sigma \in \p} \left( \tr \sqrt{ \sqrt{\rho}\, \sigma \sqrt{\rho} } \right)^2 = \frac{d}{n} \max_{\sigma \in \p}  \, F \mleft( \rho, \, \sigma \mright)\,.
\end{aligned}
\end{equation}
\end{proof}
We will now prove Corollary \ref{corr: notd^2 rand} from the main text. 
\setcounter{corollary}{2}
\begin{corollary}
    Let $\m = \{M_j\}_j$ be an unbiased extremal rank-one POVM with $d^2$ outcomes in dimension $d$, such that $M_j= \frac{1}{d} \ketbra{\psi_j}{\psi_j}$, and let $\mathcal{P}$ be the convex hull of the set $\{\ketbra{\psi_j}{\psi_j}\}$. When $\rho \notin \mathcal{P}$, the optimal decomposition $\{p_j , \, \rho_j\}$ for an eavesdropper contains at most $d^2-1$ nonzero elements.
\end{corollary}
\begin{proof}
We prove the corollary by contradiction, starting with the assumption that all $d^2$ elements $p_j \rho_j$ are non-zero. Summing the condition $X-M_j \sgeq 0$ over all $j \in \{1, \, \hdots, \, d^2\}$, we find $d^2 X \sgeq \id$, so $X$ must be full rank, and therefore invertible. The complementary slackness condition \eqref{eqn: SDP comp slack} tells us that, if every $p_j \rho_j$ is non-zero, each $X- M_j$ must have non-full rank, and therefore determinant zero.
From the matrix determinant lemma, since $X$ is invertible and every $M_j$ is rank one, this implies that
\begin{equation}
    0 = \det \mleft( X - M_j \mright) = \left( 1 - \tr \left( X^{-1} M_j\right) \right) \det \mleft(X\mright) \qquad \textnormal{for all} \;\; j \in \{1, \, \hdots, \, d^2\}\,,
\end{equation}
so 
\begin{equation}
     \tr \mleft( X^{-1} M_j\mright) = 1 \qquad \text{for all} \;\; j \in \{1, \, \hdots, \, d^2\}\,.
\end{equation}
The operator $X^{-1}$ is positive semidefinite, so it is proportional to a quantum state up to a positive constant. Since the POVM is an MIC, the overlaps $\tr \mleft( X^{-1} M_j\mright)$ uniquely specify the matrix $X^{-1}$. The above distribution is satisfied by $X^{-1}= d \id$, which gives $X= \frac{\id}{d}$. We then conclude that $X= \frac{\id}{d}$ is the unique optimal dual variable when there are $d^2$ nontrivial states. However, $X= \frac{\id}{d}$ gives a guessing probability of $\frac{1}{d}$ and we know from Theorem \ref{thm: fidelity rand} that $\pg \mleft( \rho, \, \m\mright) < \frac{1}{d}$ for states outside $\p$, so the optimal decomposition for Eve must have $d^2 -1$ or fewer nonzero states.
\end{proof}
\subsection{State discrimination}\label{app: SD}
Let's say an eavesdropper receives an ensemble of $n$ subnormalised states $\{\nu_j\}$ which sum to the state $\nu$, i.e, $\nu = \sum_{j=1}^{n}p_j \nu_j$, with $p_j = \tr \nu_j$ for some probability distribution $\{p_j\}$. We work in the subspace in which $\nu$ is full rank, and denote by $d$ the dimension of this space, $\rank \mleft( \nu \mright)=d$. If all the states $\nu_j$ are rank one and linearly independent on the operator space \emph{and} they satisfy the following equality,
\begin{equation}
    \tr \mleft( \nu^{-1} \, \nu_j \mright) = \frac{d}{n} \qquad \text{for all} \;\; j \in \{1, \, \hdots, \, n\}\,,
\end{equation}
then the ensemble $\{\nu_j\}$ can be considered to arise from the following unbiased extremal rank-one POVM $\m$ with $n$ outcomes acting on a purification of the state $\nu$,
\begin{equation}
    \m = \{M_j\}_j\,, \qquad M_j = \nu^{-\frac{1}{2}} \, \nu_j^{*} \, \nu^{-\frac{1}{2}}\,,
\end{equation}
where $\left(*\right)$ denotes conjugation in the eigenbasis of $\nu$.
Since the $n$ states $\{\nu_j\}$ are linearly independent on an operator space of dimension $d$, we must have $n \leq d^2$. Further, since the rank of a sum of matrices cannot be greater than the sum of the individual ranks, we have
\begin{equation}
    d = \rank \nu = \rank \mleft( \sum_{j=1}^{n} \nu_j \mright) \leq \sum_{j=1}^{n} \rank \nu_j =n\,,
\end{equation}
so we conclude that $d \leq n \leq d^2$. With this equivalence between the state discrimination and randomness pictures in mind, we translate the results of Section \ref{sec: ub any dim} from the main text to the language of state discrimination.

\begin{theorem}\label{thm: fidelity SD}
Let $\{\nu_j\}$ be a ensemble of $n$ rank-one states $\nu_j= p_j \ketbra{\phi_j}{\phi_j}$ that are linearly independent on the operator space and satisfy 
\begin{equation}
    \tr \mleft( \nu^{-1} \, \nu_j \mright) = \frac{d}{n} \qquad \textnormal{for all} \;\; j \in \{1, \, \hdots, \, n\}\,,
\end{equation}
where $\sum_{j=1}^{n} \nu_j = \nu$ and $d=\rank \nu$. Then
\begin{equation}\label{eqn: fidelity SD}
    \pg \mleft( \{\nu_j\} \mright) =  \max_{ \{q_j\} } \; \left( \tr \sqrt{  \sum_{j=1}^{n} q_j \, \nu_j    } \right)^2 \,,
\end{equation}
where $\{q_j\}$ is a probability distribution. 
\end{theorem}
\begin{proof}
    We have shown that our state discrimination scenario is equivalent to a randomness scenario in which Alice holds the state $\nu$ and performs the unbiased extremal measurement
    \begin{equation}
    \m = \{M_j\}_j\,, \qquad M_j = \nu^{-\frac{1}{2}} \, \nu_j^{*} \, \nu^{-\frac{1}{2}}\,, \qquad j \in \{1, \, \hdots, \, n\}
\end{equation}
with $n$ outcomes in dimension $d$. Then, taking Theorem \ref{thm: fidelity rand} and plugging in 
\begin{equation}
    \sigma = \frac{n}{d}\sum_{j=1}^{n} q_j M_j 
\end{equation}
for the state $\sigma$ inside $\p$, where $\{q_j\}$ is any probability distribution, we obtain \eqref{eqn: fidelity SD}.
\end{proof}
We will now make some observations on this state discrimination scenario.
\begin{observation}\label{thm: inP SD}
Let $\{\nu_j\}$ be a ensemble of $n$ rank-one states $\nu_j= p_j \ketbra{\phi_j}{\phi_j}$ that are linearly independent on the operator space and satisfy 
\begin{equation}
    \tr \mleft( \nu^{-1} \, \nu_j \mright) = \frac{d}{n} \qquad \textnormal{for all} \;\; j \in \{1, \, \hdots, \, n\}\,,
\end{equation}
where $\sum_{j=1}^{n} \nu_j = \nu$ and $d=\rank \nu$. If it is possible to write the operator $\nu^2$ as
\begin{equation}\label{eqn: nu squared}
    \nu^2 = \sum_{j=1}^{n} c_j \nu_j\,, \qquad c_j \geq 0 \qquad \textnormal{for all} \;\; j \in \{1, \, \hdots, \, n\}\,,
\end{equation}
then the optimal probability for state discrimination is
\begin{equation}
    \pg \mleft( \{\nu_j\}\mright) = \frac{d}{n}\,.
\end{equation}
\end{observation}
\begin{proof}
    If condition \eqref{eqn: nu squared} holds then it is possible to define a valid POVM $\mathcal{K}=\{K_j\}_j$ with elements
    \begin{equation}
        K_j = c_j \nu^{-1} \nu_j \nu^{-1}\,,
    \end{equation}
    satisfying
    \begin{equation}
      \tr \mleft( K_j \, \nu \mright) = c_j \tr \mleft( \nu^{-1} \, \nu_j \mright) = \frac{d}{n} c_j \qquad j \in \{1, \, \hdots, \, n\}
    \end{equation}
    such that $\sum_{j=1}^{n} c_j = \frac{n}{d}$. This POVM, when it exists, is a maximum confidence measurement for the set of states $\{\nu_j\}$ \cite{Croke_2006}. Using $\mathcal{K}$, the optimal guessing probability can be lower bounded by 
    \begin{equation}
\begin{aligned}
  \pg \left( \{\nu_j\}\right) & \geq    \sum_{j=1}^{n} \tr \mleft( K_j \, \nu_j\mright) = \sum_{j=1}^{n} c_j \tr \mleft( \nu^{-1} \nu_j \nu^{-1} \nu_j \mright)
     \\
     &= \sum_{j=1}^{n} c_j \, p_j^2 \langle \phi_j | \nu^{-1} | \phi_j \rangle^2 = \sum_{j=1}^{n} c_j \mleft( \tr \mleft( \nu^{-1} \nu_j  \mright) \mright)^2 = \frac{d^2}{n^2} \sum_{j=1}^{n} c_j = \frac{d}{n}\,.
    \end{aligned}
    \end{equation}
    We could then define the dual variable $Y = \frac{d}{n}\nu$, which is strictly positive. Since each $\nu_j$ is rank one, the terms $Y-\nu_j$ can have at most one non-positive eigenvalue. The corresponding eigenvector is $\nu^{-1}\ket{\phi_j}$, because
    \begin{equation}
        \left( Y - \nu_j\right) \nu^{-1} \ket{\phi_j} = \left( \frac{d}{n} \nu - p_j\ketbra{\phi_j}{\phi_j}\right) \nu^{-1} \ket{\phi_j} = \left( \frac{d}{n} - \tr \mleft( \nu^{-1} \nu_j\mright)\right) \ket{\phi_j} =0
    \end{equation}
    for all $j \in \{1, \, \hdots, \, n\}$,
    so $Y$ is a valid variable. It upper bounds the guessing probability as
    \begin{equation}
     \pg \mleft( \{\nu_j\} \mright) \leq   \tr Y = \frac{d}{n}\,,
    \end{equation}
    which completes the proof.
\end{proof}
Note that when $n=d^2$, since the states $\{\nu_j\}$ are linearly independent, they necessarily form a basis for all operators in dimension $d$. This means we can always find a set of real coefficients $\{c_j\}$ such that $\nu^2= \sum_{j=1}^{d^2} c_j \nu_j$, but they might not all be non-negative. 
\begin{corollary}\label{cor: d^2 SD}
Let $\{\nu_j\}$ be a ensemble of $n$ rank-one states $\nu_j= p_j \ketbra{\phi_j}{\phi_j}$ that are linearly independent on the operator space and satisfy 
\begin{equation}
    \tr \mleft( \nu^{-1} \, \nu_j \mright) = \frac{d}{n} \qquad \textnormal{for all} \;\; j \in \{1, \, \hdots, \, d^2\}\,,
\end{equation}
where $\sum_{j=1}^{d^2} \nu_j = \nu$ and $d=\rank \nu$. Then
\begin{equation}\label{eqn: ineq 1/d^2}
    \pg \left( \{\nu_j\}\right) > \frac{1}{d^2}\,.
\end{equation}
\end{corollary}
\begin{proof}
The probability of discriminating $n$ states is always lower bounded by $\frac{1}{n}$, so we need only consider the ensembles for which $n=d^2$. In this case, this bound can only be saturated if $\tr \nu_j = p_j=\frac{1}{d^2}$ for all $j \in \{1, \, \hdots, \, d^2\}$. 
Since $\{\nu_j\}$ form a basis for the operator space when $n=d^2$, we can uniquely define $\nu^2= \sum_{j=1}^{d^2}c_j \nu_j$ for some set of real coefficients $\{c_j\}$. These coefficients will satisfy
\begin{equation}
    1= \tr \left(\nu^2 \, \nu^{-1}  \right) = \sum_{j=1}^{d^2} c_j \tr \mleft( \nu_j \nu^{-1} \mright) = \frac{1}{d} \sum_{j=1}^{d^2} c_j\,,
\end{equation}
so we find that $\sum_{j=1}^{d^2} c_j = d$. But this implies that
\begin{equation}
    \tr  \nu^2  = \sum_{j=1}^{d^2} c_j \tr \nu_j = \frac{1}{d^2} \sum_{j=1}^{d^2} c_j = \frac{1}{d}\,,
\end{equation}
which is only possible if $\nu = \frac{\id}{d}$. Then, since
\begin{equation}
    \sum_{j=1}^{d^2} \nu_j = \nu = \frac{\id}{d}\,,
\end{equation}
the condition
\begin{equation}
    \sum_{j=1}^{d^2} c_j \nu_{j} = \nu^2 = \frac{\id}{d^2}
\end{equation}
is satisfied uniquely by the decomposition
\begin{equation}
    c_j = \frac{1}{d} \qquad \textnormal{for all} \;\; j \in \{1, \, \hdots, \, d^2\}\,.
\end{equation}
Since all the coefficients $c_j$ are positive, by Theorem \ref{thm: fidelity SD} we must have $\pg \mleft(\{\nu_j\} \mright)= \frac{1}{d}$, which means
\begin{equation}
    \pg \mleft( \{\nu_j\}\mright) > \frac{1}{d^2}\,.
\end{equation}
\end{proof}

\begin{observation}\label{thm: outP SD}
Let $\{\nu_j\}$ be a ensemble of $n$ rank-one states $\nu_j= p_j \ketbra{\phi_j}{\phi_j}$ that are linearly independent on the operator space and satisfy 
\begin{equation}
    \tr \mleft( \nu^{-1} \, \nu_j \mright) = \frac{d}{n} \qquad \textnormal{for all} \;\;  j \in \{1, \, \hdots, \, n\}\,,
\end{equation}
where $\sum_{j=1}^{n} \nu_j = \nu$ and $d=\rank \nu$. If it is \emph{not} possible to write the operator $\nu^2$ as 
\begin{equation}\label{eqn: nu squared 2}
    \nu^2 = \sum_{j=1}^{n} c_j \nu_j\,, \qquad c_j \geq 0 \qquad \textnormal{for all} \;\; j \in \{1, \, \hdots, \, n\}\,,
\end{equation}
then the optimal probability for state discrimination is upper bounded by
\begin{equation}\label{eqn: K POVM 2}
    \pg \mleft( \{\nu_j\}\mright) <  \frac{d}{n}\,.
\end{equation}
\end{observation}
\begin{proof}
Let's define $\K= \{K_j\}_j$ as the optimal POVM for state discrimination for the set $\{\nu_j\}$.
    The guessing probability is then upper bounded by the Cauchy-Schwartz inequality,
    \begin{equation}
    \begin{aligned}
    \pg \mleft( \{\nu_j\} \mright) & = \sum_{j=1}^{n} \tr \mleft( \nu_{j} K_j \mright) = \sum_{j=1}^{n} \tr \mleft( \nu^{- \frac{1}{2}} \nu_j \nu^{-\frac{1}{2}} \; \sqrt{\nu} K_j \sqrt{\nu} \mright) 
    \\
    &\leq \sum_{j=1}^{n} \tr \mleft( \nu^{-1} \nu_j \mright) \tr \mleft( K_j \nu \mright) = \frac{d}{n}\,.
    \end{aligned}
    \end{equation}
    Since each $\nu_j$ is rank one, the inequality on the second line is saturated if and only if
    \begin{equation}
     \sqrt{\nu} K_j \sqrt{\nu} \propto   \nu^{-\frac{1}{2}} \nu_j \nu^{- \frac{1}{2}} \qquad \textnormal{for all} \;\; j \in \{1, \, \hdots, \, n\}\,,
    \end{equation}
    or, in other words, if the the elements $\{K_j\}$ have the form
    \begin{equation}
        K_j = c_j \nu^{-1} \nu_j \nu^{-1}\,, \qquad c_j \geq 0 \;\; \textnormal{for all} \;\; j \in \{1, \, \hdots, \, n\}\,,
    \end{equation}
    which implies that condition \eqref{eqn: nu squared 2} holds.
\end{proof} 

The next corollary holds when $n=d^2$.
\begin{corollary}\label{cor: not d^2 SD}
 Let $\{\nu_j\}$ be a ensemble of $d^2$ rank-one states $\nu_j= p_j \ketbra{\phi_j}{\phi_j}$ that are linearly independent on the operator space and satisfy 
\begin{equation}
    \tr \mleft( \nu^{-1} \, \nu_j \mright) = \frac{d}{n} \qquad \textnormal{for all} \;\; j \in \{1, \, \hdots, \, d^2\}\,,
\end{equation}
where $\sum_{j=1}^{d^2} \nu_j = \nu$ and $d=\rank \nu$. If it is \emph{not} possible to write the operator $\nu^2$ as 
\begin{equation}\label{eqn: nu squared 3}
    \nu^2 = \sum_{j=1}^{d^2} c_j \nu_j\,, \qquad c_j \geq 0 \qquad \textnormal{for all} \;\; j \in \{1, \, \hdots, \, d^2\}\,,
\end{equation}
then the optimal measurement for state discrimination has at most $d^2-1$ nonzero elements.  
\end{corollary}

\begin{proof}
We prove the corollary by contradiction, starting with the assumption that all $d^2$ elements $K_j$ in the optimal POVM for state discrimination are non-zero. Summing the condition $Y-\nu_j \sgeq 0$ over all $j \in \{1, \, \hdots, \, d^2\}$, we find $d^2 Y \sgeq \rho$, so $Y$ must be full rank, and therefore invertible. The complementary slackness condition \eqref{eqn: SDP comp slack} tells us that, if every $\nu_j$ is non-zero then every $Y- \nu_j$ must have non-full rank, and therefore determinant zero.
From the matrix determinant lemma, since $Y$ is invertible and every $\nu_j$ is rank one, this implies that
\begin{equation}
    0 = \det \mleft( Y - \nu_j \mright) = \left( 1 - \tr \mleft( Y^{-1} \nu_j \mright) \right) \det \mleft(Y \mright) \qquad \text{for all} \;\; j \in \{1, \, \hdots, \, d^2\}\,,
\end{equation}
so 
\begin{equation}
     \tr \mleft( Y^{-1} \nu_j\mright) = 1 \qquad \text{for all} \;\; j \in \{1, \, \hdots, \, d^2\}\,.
\end{equation}
The operator $Y^{-1}$ is positive semidefinite, so it is proportional to a quantum state up to a positive constant. Since the $d^2$ states $\{\nu_j\}$ are linearly independent, the overlaps $\tr \Big( Y^{-1} \nu_j \Big)$ uniquely specify the matrix $Y^{-1}$. The above distribution is satisfied by $Y^{-1}= d \nu^{-1}$, which gives $Y= \frac{\nu}{d}$. We then conclude that $Y= \frac{\nu}{d}$ is the unique optimal dual variable when there are $d^2$ nontrivial POVM elements. However, $Y= \frac{\nu}{d}$ gives a guessing probability of $\frac{1}{d}$ and we know from Observation \ref{thm: outP SD} that $\pg \mleft( \{\nu_j\} \mright) < \frac{1}{d}$ for states satisfying the condition \eqref{eqn: nu squared 3}, so the optimal state discrimination strategy must be to discard at least one of the $d^2$ outcomes.
\end{proof}

\section{Proofs from Section \ref{sec: ub qubit}}

\subsection{Unbiased MICs in dimension two}

\subsubsection{Parameterization with two angles}\label{sssec: two angles}
We will show that any unbiased MIC POVM $\m=\{M_j\}_j$ in dimension two can be parameterized by two angles. Writing the POVM elements as $M_j= \frac{1}{2}\ketbra{\psi_j}{\psi_j}$, we can represent the normalised projectors $\{\ketbra{\psi_j}{\psi_j}\}$ in Bloch vector notation as
\begin{equation}
  \ketbra{\psi_j}{\psi_j}= \frac{1}{2} \left(  \id + \vecr_j \cdot \vec{\sigma}\right)\,, \qquad \abs{\vecr_j}^2=1 \qquad \textnormal{for all} \;\; j \in \{1, \, \hdots, \, 4\}\,,
\end{equation}
where $\vec{\sigma}= \left( X, \, Y, \, Z\right)$ is a vector of Pauli matrices. 
We'll label the elements of each measurement vector by $\vecr_j = \left( x_j, \; y_j, \; z_j\right)$. The POVM is valid and informationally complete if and only if $\sum_{j=1}^{4} \vecr_j=0$ and the operators $\{\ketbra{\psi_j}{\psi_j}\}$ are all linearly independent \cite{D_Ariano_2005} (note that this does not imply that the vectors $\{\ket{\psi_j}\}$ are themselves linearly independent). We append the element 1 to the start of each vector to obtain $\{1, \vecr_j\}$. For the second condition to hold, from \cite[Section 3]{D_Ariano_2005}, we need that the four vectors $\{1, \, \vecr_j\}$ are linearly independent, so, defining the tetrahedron formed by the vectors $\{\vecr_j\}$ as $\mathcal{P}$, i.e., $\mathcal{P}= \textnormal{conv} \left( \{\vecr_j\}\right)$, we need
\begin{equation}\label{eqn: det qubits}
  \abs{ \det 
    \begin{pmatrix}
    1 & x_1 & y_1 & z_1
    \\
    1 & x_2 & y_2 & z_2
    \\
    1 & x_3 & y_3 & z_3
    \\
    1 & x_4 & y_4 & z_4
    \end{pmatrix}} =  6 \textnormal{Vol} \left( \mathcal{P} \right) \neq 0\,.
\end{equation}
In other words, the POVM elements are linearly independent as long as the polytope $\mathcal{P}$ has non-zero volume, or as long as not all the vectors lie in the same plane. The tetrahedron has its centre of mass at the origin. The authors of \cite{D_Ariano_2005} also note that normalisation forces the following three conditions, 
\begin{equation}\label{eqn: sum qubits}
\begin{aligned}
    \vecr_1 + \vecr_2 = - \left( \vecr_3 + \vecr_4  \right)\,,
    \\
    \vecr_1 + \vecr_3 = - \left( \vecr_2 + \vecr_4  \right)\,,
    \\
    \vecr_1 + \vecr_4 = - \left( \vecr_2 + \vecr_3  \right)\,.
    \end{aligned}
\end{equation}
Squaring both sides of each expression, for unbiased MICs, we have 
\begin{equation}\label{eqn: dot prod}
    \vecr_1 \cdot \vecr_2 = \vecr_3 \cdot \vecr_4 \,, \qquad \vecr_1 \cdot \vecr_3 = \vecr_2 \cdot \vecr_4 \,, \qquad \vecr_1 \cdot \vecr_4 = \vecr_2 \cdot \vecr_3\,.
\end{equation}
From this, we find 
\begin{equation}\label{eqn: triangle sides}
\begin{aligned}
  \abs{\vecr_1 - \vecr_2}^2 = 2 \left( 1 - \vecr_1 \cdot \vecr_2\right)  &= 2 \left( 1 - \vecr_3 \cdot \vecr_4\right) = \abs{\vecr_3 - \vecr_4}^2\,,
  \\
 \abs{\vecr_1 - \vecr_3}^2 = 2 \left( 1 - \vecr_1 \cdot \vecr_3\right)  &= 2 \left( 1 - \vecr_2 \cdot \vecr_4\right) = \abs{\vecr_2 - \vecr_4}^2\,,
 \\
 \abs{\vecr_1 - \vecr_4}^2 = 2 \left( 1 - \vecr_1 \cdot \vecr_4\right)  &= 2 \left( 1 - \vecr_2 \cdot \vecr_3\right) = \abs{\vecr_2 - \vecr_3}^2\,,
\end{aligned}
\end{equation}
so the tetrahedron has four congruent triangular faces. Such a tetrahedron is called a disphenoid. Since no two vectors are equal or antiparallel, we can align our Bloch sphere such that the first two vectors are given by
\begin{equation}
    \begin{aligned}
        \vecr_1 &= \left( - \cos \delta , \, -\sin \delta, \, 0\right)
        \\
        \vecr_2 &= \left( - \cos \delta , \, \sin \delta, \, 0\right)\,
    \end{aligned}
\end{equation}
for some angle $0 < \delta < \frac{\pi}{2}$.
By \eqref{eqn: sum qubits}, the remaining vectors must have the form 
\begin{equation}
\begin{aligned}
     \vecr_3 &= \left( x_3 , \, y_3, \, z_3\right)
        \\
        \vecr_4 &= \left( 2 \cos \delta -x_3, \, -y_3, \, -z_3\right)\,.
        \end{aligned}
\end{equation}
Applying
\begin{equation}
    \abs{\vecr_3}^2 = \abs{\vecr_4}^2 =1\,,
\end{equation}
we find 
\begin{equation}
    x_3 = \cos \delta\,, \qquad y_3^2 + z_3^2 = \sin^2 \delta\,.
\end{equation}
Without loss of generality, we can then write 
\begin{equation}
    y_3 = - \sin \delta \sin \alpha\,, \qquad z_3 = \pm \sin \delta \cos \alpha
\end{equation}
for some angle $- \frac{\pi}{2} < \alpha < \frac{\pi}{2}$. We take the liberty of choosing the positive sign for $z_3$ such that $\vecr_3$ lives in the upper hemisphere of the Bloch sphere. The final vectors are then 
\begin{equation}
\begin{aligned}
     \vecr_3 &= \left( \cos \delta , \, - \sin \delta \sin \alpha, \, \sin \delta \cos \alpha\right)
        \\
        \vecr_4 &= \left( \cos \delta , \,  \sin \delta \sin \alpha, \, -\sin \delta \cos \alpha\right)\,.
        \end{aligned}
\end{equation}
From \eqref{eqn: det qubits}, the volume of the polytope is
\begin{equation}
    \textnormal{Vol} \mleft( \mathcal{P} \mright) = \frac{4}{3} \sin^2 \delta \cos \delta \cos \alpha\,,
\end{equation}
which is nonzero in the range $0 < \delta < \frac{\pi}{2}$ and $- \frac{\pi}{2} < \alpha < \frac{\pi}{2}$, so all POVMs defined in this way are informationally complete. The corresponding measurement vectors are
\begin{equation}\label{eqn: meas vec MIC}
    \begin{aligned}
        \ket{\psi_1} &= \frac{1}{\sqrt{2}} \left( \ket{0} - e^{ i \delta} \ket{1} \right)\,,
        \\
        \ket{\psi_2} &= \frac{1}{\sqrt{2}} \left( \ket{0} - e^{-i \delta} \right)\,,
        \\
        \ket{\psi_3} &= \frac{1}{\sqrt{2 \left( 1 + \sin \delta \cos \alpha\right)}} \left( \left( 1 + \sin \delta \cos \alpha\right) \ket{0} + \left( \cos \delta - i \sin \delta \sin \alpha \right) \ket{1} \right)\,,
        \\
        \ket{\psi_3} &= \frac{1}{\sqrt{2 \left( 1 - \sin \delta \cos \alpha\right)}} \left( \left( 1 - \sin \delta \cos \alpha\right) \ket{0} + \left( \cos \delta + i \sin \delta \sin \alpha \right) \ket{1} \right)\,.
    \end{aligned}
\end{equation}
We are now ready to prove Theorem \ref{thm: group cov} from the main text.

\setcounter{theorem}{1}
\subsubsection{Group covariance}\label{sssec: group cov}
\begin{theorem}
  All qubit unbiased extremal POVMs are group covariant.  
\end{theorem}
\begin{proof}
Following \cite[Chapter 12.1]{Bengtsson2017}, a collection of $d^2$ unitary matrices $\{U_j\}$ is a \emph{unitary operator basis} of group type if one of its elements, say $U_1$, is the identity matrix, if 
\begin{equation}
    \tr U_j =0 \qquad \textnormal{for all} \;\; j \in \{2, \, \hdots, \, d^2\}\,,
\end{equation}
and if
\begin{equation}
    U_{j} U_{k} = \lambda \mleft( j, \, k\mright) U_{jk}\,, \qquad \abs{\lambda \mleft( j, \, k\mright)}=1 \qquad \textnormal{for all} \;\; j , \, k \;\; \in \{1, \, \hdots, \, d^2\}\,.
\end{equation}
One can then construct a group-covariant POVM $\m=\{M_j\}_j$ with elements
\begin{equation}
    M_j = \frac{1}{d}U_j \nu U_j^{\dagger}\,,
\end{equation}
where $\nu$ is some density matrix \cite{D_Ariano_2004}.
One such group in dimension $d=2$ is given by $\{U_j\}$, where
\begin{equation}
    \begin{aligned}
        U_1 &= \id\,,
        \\
        U_2 &= X\,,
        \\
        U_3 &= \frac{1}{\sqrt{2}} \left( \cos \frac{\alpha}{2} + \sin \frac{\alpha}{2}\right) Y - \frac{1}{\sqrt{2}} \left( \cos \frac{\alpha}{2} - \sin \frac{\alpha}{2}   \right) Z\,,
        \\
         U_4 &= \frac{1}{\sqrt{2}} \left( \cos \frac{\alpha}{2} - \sin \frac{\alpha}{2}\right) Y + \frac{1}{\sqrt{2}} \left( \cos \frac{\alpha}{2} + \sin \frac{ \alpha}{2}\right) Z\,,
    \end{aligned}
\end{equation}
Choosing the fiducial vector
\begin{equation}
    \ket{\phi} = \frac{1}{\sqrt{2}} \left( \ket{0} - e^{i \delta} \ket{1} \right) = \ket{\psi_1}\,,
\end{equation}
where $\nu= \ketbra{\phi}{\phi}$, we recover the measurement vectors \eqref{eqn: meas vec MIC}.
\end{proof}

\subsubsection{Geometry}\label{ssec: geometry}
In this section, we will not use the parameterization of the qubit MIC $\m$ given in \ref{sssec: two angles}. Instead,
we introduce the following notation, 
\begin{equation}\label{eqn: a b c}
   a= \vecr_1 \cdot \vecr_2  \,, \qquad b= \vecr_1 \cdot \vecr_3 \,, \qquad c= \vecr_2 \cdot \vecr_3\,.
\end{equation}
From \eqref{eqn: dot prod}, note the useful properties
\begin{equation}\label{eqn: sum prop}
    a + b + c = -1
\end{equation}
and 
\begin{equation}\label{eqn: other prop}
    ab+ ac + bc = - a \left( 1 + a \right) + bc = -b \left( 1+ b\right) + ac = - c \left( 1+c \right) + ab\,.
\end{equation}
Summing the final three expressions together, we find
\begin{equation}
    2 \left( ab + ac + bc \right) = - a \left( 1 + a \right)  -b \left( 1+ b\right)- c \left( 1+c \right) = 1 - a^2 - b^2 -c^2\,.
\end{equation}
Using \eqref{eqn: triangle sides}, we can write the semiperimeter $s$ of the triangle formed by the vectors $\vecr_1$, $\vecr_2$ and $\vecr_3$ as
\begin{equation}
    s = \frac{1}{2} \left( \abs{\vecr_1- \vecr_2} + \abs{\vecr_1- \vecr_3} + \abs{\vecr_2- \vecr_3}\right) = \frac{1}{\sqrt{2}} \left( \sqrt{1-a} + \sqrt{1-b} + \sqrt{1-c} \right)
\end{equation}
so, using Heron's formula, the squared area of the triangle is 
\begin{align}
A^2 &= s \left( s -\abs{\vecr_1 - \vecr_2}\right) \left( s - \abs{\vecr_1 - \vecr_3} \right) \left(s - \abs{\vecr_2 - \vecr_3} \right)
\\
    &= \frac{1}{4} \left( 3 - 2 \left(a + b + c \right) + 2 \left( ab + ac + bc \right) - a^2 - b^2 -c^2 \right)
    \\
    &= \frac{1}{2} \left( 3 - a^2 - b^2 -c^2 \right) = 1 + ab + ac + bc\,. \label{eqn: Asquared}
    \end{align}
We will make extensive use of \eqref{eqn: sum prop}, \eqref{eqn: other prop} and the two formulations of $A^2$ in \eqref{eqn: Asquared} throughout this section. We could alternatively have found the area by taking the cross product
\begin{equation}
    A = \frac{1}{2} \abs{\left( \vecr_1 - \vecr_3 \right) \times \left( \vecr_2 - \vecr_3\right) }\,,
\end{equation}
where the choice of vectors is arbitrary. We denote the faces of the polytope $\p$ by $\left( 123\right)$, $\left( 124\right)$, $\left( 134\right)$ and $\left( 234\right)$. We choose an orientation of the vertices such that the following vectors normal to each face point out of the tetrahedron,
\begin{equation}\label{eqn: normal vecs}
\begin{aligned}
    &\vec{n}_{\left(123\right)} =  \frac{   ( \vecr_1 - \vecr_3 ) \times ( \vecr_2 - \vecr_3  ) }{2 A   } \,, \qquad
    &&\vec{n}_{\left(124\right)} =  -\frac{   ( \vecr_1 - \vecr_4 ) \times ( \vecr_2 - \vecr_4  ) }{ 2 A  } \,, \vspace{0.2 cm}
    \\
    &\vec{n}_{\left(134\right)} = \frac{   ( \vecr_1 - \vecr_4 ) \times ( \vecr_3 - \vecr_4  ) }{ 2 A  } \,, \qquad
    &&\vec{n}_{\left(234\right)} = -\frac{   ( \vecr_2 - \vecr_4 ) \times ( \vecr_3 - \vecr_4  ) }{ 2 A  } \,.
    \end{aligned}
\end{equation}
We will use $\{\vec{n}\}= \{ \vec{n}_{\left(123\right)}, \, \vec{n}_{\left(124\right)}, \, \vec{n}_{\left(134\right)}, \, \vec{n}_{\left(234\right)}   \}$ as a shorthand for this set of unit vectors.
Defining 
\begin{equation}
    l= \frac{\vecr_1 \cdot \left( \vecr_2 \times \vecr_3 \right)}{2 A}\,, 
\end{equation}
where again the vectors are chosen arbitrarily, we have
\begin{equation}
\begin{aligned}
l 
&= \vec{n}_{(123)} \cdot \vecr_i 
&& \qquad i \!\! &\! \in \;& \{1, \, 2, \, 3\} \\
&= \vec{n}_{(124)} \cdot \vecr_j 
&& \qquad j \!\! & \!\in \;& \{1, \, 2, \, 4\} \\
&= \vec{n}_{(134)} \cdot \vecr_k 
&& \qquad k \!\! & \!\in \;& \{1, \, 3, \, 4\} \\
&= \vec{n}_{(234)} \cdot \vecr_l 
&& \qquad l \!\! & \!\in \;& \{2, \, 3, \,4\}.
\end{aligned}
\end{equation}
It will be helpful to have an expression for $l$ that is free of cross products. Note that $l$ is the shortest distance from the origin to the plane $\left( 123 \right)$. The point $l \vec{n}_{\left(123 \right)}$, where the normal vector $\vec{n}_{\left( 123\right)}$ intersects the face $\left( 123\right)$, must be of the form 
\begin{equation}\label{eqn: ln in bary}
    l \vec{n}_{\left( 123 \right)}= p_1 \vecr_1 + p_2 \vecr_2 + p_3 \vecr_3\,
\end{equation}
for some probability distribution $\{p_j\}$. Defining the column vector
\begin{equation}
    \vec{p} = \left( p_1, \, p_2, \, p_3\right)^{\top}\,,
\end{equation}
we can consider the optimisation problem
\begin{equation}
    \begin{aligned}
     l^2=     \underset{ \vec{p}}{\text{minimize}}  \, \quad & \vec{p}^{ \top}    X \, \vec{p}  && {}
        \\
         \text{subject to} \quad & - \vec{p}  \leq 0\,, \quad &&  
        \\
           \quad &\left( 1, \, 1, \, 1\right) \, \vec{p}     = 1\,, \quad && 
    \end{aligned}
\end{equation}
where the inequality constraint holds for each component of $\vec{p}$ and where
\begin{equation}
    X = \begin{pmatrix}
        1 & a & b
        \\
        a & 1 & c
        \\
        b & c & 1
    \end{pmatrix}\,.
\end{equation}
The inequality constraints are convex and the equality constraint is linear, so the KKT conditions are necessary and sufficient. The complementary slackness condition \eqref{eqn: KKT comp slack} is
\begin{align}
\lambda_j \, p_j = 0\,, \qquad  \textnormal{for all} \;\; j \in \{1, \, 2, \, 3\}  \,.
\end{align}
We know that the state in question must be in the interior of the face $\left( 123\right)$, so all $p_j$ must be nonzero. This forces all the slack variables $\lambda_j$ to be zero, such that the stationarity condition \eqref{eqn: KKT stat} becomes
\begin{equation}
    X  \vec{p} + \nu \left(1, \, 1, \, 1 \right)^{\top} =0 \,.
\end{equation}
Inverting $X$ and using \eqref{eqn: sum prop}, we find
\begin{equation}
    X^{-1} = \frac{1}{2t} \begin{pmatrix}
     1 -c^2 & bc -a & ac - b
     \\
     bc - a & 1 - b^2 & ab -c
     \\
     ac - b & ab -c & 1 - a^2
    \end{pmatrix}\,,
\end{equation}
where
\begin{equation}\label{eqn: t}
    t = \left(1+a \right) \left( 1+b\right) \left( 1+c\right)\,.
\end{equation}
The quantity $t$ will be useful later. Setting
\begin{equation}
    \nu= - \frac{t}{2 A^2}
\end{equation}
such that $\vec{p}$ is normalised, we find
\begin{equation}
    \vec{p} = \frac{t}{2 A^2} X^{-1}  \left(1 , \, 1, \, 1 \right)^{\top} = \frac{1}{2 A^2} \left( 1- c^2, \, 1-b^2, \, 1-a^2  \right)^{\top}
\end{equation}
such that
\begin{equation}\label{eqn: l squared}
    l^2 =\frac{t}{2 A^2} \,   \left(1 , \, 1, \, 1 \right) \, X^{-1} \, X \,  \vec{p} = \frac{t}{2 A^2}
\end{equation}
and \eqref{eqn: ln in bary} becomes
\begin{equation}\label{eqn: l n bary}
    l \vec{n}_{\left( 123 \right)} = \frac{1}{2 A^2} \left(  \left( 1-c^2 \right)\vecr_1 + \left( 1-b^2\right) \vecr_2 + \left( 1-a^2\right) \vecr_3  \right)\,.
\end{equation}
Using \eqref{eqn: Asquared}, we can rewrite \eqref{eqn: l squared} as
\begin{equation}\label{eqn: l squared ito a}
    l^2 = \frac{1+a}{2 A^2} \left( 1+ b + c + bc \right) = \frac{1+a}{2} \left( 1 - \frac{1-a^2}{A^2} \right)\,,
\end{equation}
where $a$ could equivalently be swapped with $b$ or $c$. Finally, it will be useful to have an expression for $\vec{n}_{\left( 123\right)} \cdot \vec{n}_{\left( 124\right)}$ using only the variables $l$ and $a$. 
We will use yet another expression for the squared area,
\begin{equation}\label{eqn: weird area}
    A^2 = 1 + a \left(b + c\right) + bc = \left( 1-a^2\right)- \left( a - bc\right)\,,
\end{equation}
such that
\begin{equation}\label{eqn: weird area 2}
    A^2 - 2 \left( 1-a^2 \right) = - A^2 - 2 \left( a - bc\right)\,.
\end{equation}
Then, using \eqref{eqn: weird area 2} and the cross product forms \eqref{eqn: normal vecs}, we can write 
\begin{equation}\label{eqn: nprod 1}
\begin{aligned}
\vec{n}_{\left( 123\right)} \cdot \vec{n}_{\left( 124 \right)} &= \frac{A^2 -2 \left( 1- a^2\right)}{ A^2}  = -1 - \frac{2 \left( a - bc\right)}{A^2}\,. 
\end{aligned}
\end{equation}
Now, with \eqref{eqn: t} and \eqref{eqn: l squared}, we use 
\begin{equation}\label{eqn: other l}
  2   l^2 A^2 =  t =\left( 1+a \right) \left( 1+b\right) \left( 1+c\right) = -\left( 1+a \right) \left( a -bc \right)
\end{equation}
to write 
\begin{equation}\label{eqn: A^2 upside down}
    \frac{a-bc}{A^2} = - \frac{2 l^2}{1+a}
\end{equation}
such that
\begin{equation}\label{eqn: nprod 2}
\vec{n}_{\left( 123\right)} \cdot \vec{n}_{\left( 124 \right)} = -1 + \frac{4 l^2}{1+a}\,.   
\end{equation}
Similar expressions hold for the other overlaps if we swap $a$ for $b$ or $c$.

\subsubsection{Probability space}
We can gain an insight into a given informationally complete POVM by characterising the entire set of probability distributions it can produce. We begin by establishing some general properties of unbiased MICs in any finite dimension $d$, referring the reader to \cite{debrota2021} for more details. Since our POVM $\m=\{M_j\}_j$, with elements $M_j=\frac{1}{d}\ketbra{\psi_j}{\psi_j}$, is informationally complete by definition, we can represent any state $\rho$ uniquely as a normalised linear combination of the projectors $\{ \ketbra{\psi_j}{\psi_j} \}$, i.e.,
\begin{equation}\label{eqn: q rep}
    \rho = \sum_{j=1}^{d^2} q_j \ketbra{\psi_j}{\psi_j}\,, \qquad \sum_{j=1}^{d^2} q_j =1\,,
\end{equation}
where $\{q_j\}$ is a quasiprobability distribution, as its terms sum to one but may not all be positive. 
To determine the coefficients $\{q_j\}$, we first define the Gram matrix $G$, whose entries are
\begin{equation}\label{eqn: gram}
   \left[G \right]_{ij} =  \abs{ \langle \psi_i | \psi_j\rangle}^2\,, \qquad i, \, j \in \{1, \, \hdots, \, d^2\}\,.
\end{equation}
(Note that some works define the Gram matrix entries by $\left[G \right]_{ij}= \tr \mleft( M_i \, M_j\mright)= \frac{1}{d^2}\abs{ \langle \psi_i | \psi_j\rangle}^2$.) We can then find the probabilities by
\begin{equation}
 d p_i=   \tr \mleft( \rho \ketbra{\psi_i}{\psi_i}\mright)  = \sum_{j=1}^{d^2} q_j \left[ G \right]_{ij}\,.
\end{equation}
The Gram matrix $G$ is positive definite, and thus invertible, because the operators $\{\ketbra{\psi_j}{\psi_j}\}$ are linearly independent. Taking the inverse $G^{-1}$, we can define the \emph{dual basis} $\{F_j\}$ of the set $\{\ketbra{\psi_j}{\psi_j}\}$, which has elements 
\begin{equation}\label{eqn: F dual}
   F_j = \sum_{k=1}^{d^2} \left[ G^{-1} \right]_{jk} \ketbra{\psi_k}{\psi_k}\,, \qquad j \in \{1, \, \hdots, \, d^2\}
\end{equation}
such that 
\begin{equation}
    \tr \mleft( F_j \ketbra{\psi_i}{\psi_i}\mright) = \sum_{k=1}^{d^2} \left[ G^{-1}\right]_{jk} \left[ G\right]_{ki}= \delta_{i, \, j} \qquad \textnormal{for all} \;\; i, \, j \in \{1, \, \hdots, \, d^2\}\,,
\end{equation}
and we find the inverse relation
\begin{equation}\label{eqn: q ito p}
 q_j=   \tr \mleft( \rho F_j \mright)= d \sum_{k=1}^{d^2} \left[ G^{-1} \right]_{jk} p_k\,.
\end{equation}
Subbing in the maximally mixed state $\rho= \frac{\id}{d}$, for which $q_j= p_j =\frac{1}{d^2}$ for all $j$, we find
\begin{equation}
    \tr F_j = \sum_{k=1}^{d^2} \left[G^{-1} \right]_{jk} = \frac{1}{d} \qquad \textnormal{for all} \;\; j \in \{1, \, \hdots, \, d^2\}\,.
\end{equation}
We can also define the `classical' average post-measurement state of a given $\rho$ as
\begin{equation}\label{eqn: post meas}
    \mathcal{O} \mleft( \rho \mright) = \sum_{j=1}^{d^2} \tr \mleft( \rho M_j \mright) \ketbra{\psi_j}{\psi_j} = \sum_{j=1}^{d^2} p_j \ketbra{\psi_j}{\psi_j}\,.
\end{equation}
Since $\m$ is informationally complete, the map $\mathcal{O} \mleft( \cdot \mright)$ is invertible. Note, finally, that we can extract the probabilities $\{p_j\}$ from $\mathcal{O}\mleft( \rho\mright)$ by 
\begin{equation}\label{eqn: prob O}
    p_j = \tr \mleft( \mathcal{O} \mleft(\rho \mright) F_j  \mright)\,.
\end{equation}
Now let's return to our unbiased MIC POVMs in dimension two, with measurement vectors $\{\vecr_j\}$ that satisfy \eqref{eqn: a b c} and all of the relations given in \ref{ssec: geometry}. Consider a state $\rho$ in Bloch notation, 
\begin{equation}
    \rho = \frac{1}{2} \left( \id + \vecr \cdot \vec{\sigma}\right)\,, \qquad \abs{\vecr}^2 \leq 1\,.
\end{equation}
In terms of Bloch vectors, then, from \eqref{eqn: q rep} and \eqref{eqn: gram} we find 
\begin{equation}
    \vecr = \sum_{j=1}^{4} q_j \vecr_j\,, \qquad \left[ G \right]_{ij} = \frac{1}{2} \left( 1 + \vecr_i \cdot \vecr_j \right) \qquad \textnormal{for all} \;\; i, \, j \in \{1, \, \hdots, \, 4\}\,,
\end{equation}
and from \eqref{eqn: F dual} we have
\begin{equation}\label{eqn: f eqs}
    F_j = \frac{1}{4} \left( \id + \vec{f}_j \cdot \vec{\sigma} \right)\,, \qquad \vec{f}_j = 2 \sum_{k=1}^{4} \left[ G^{-1}\right]_{jk} \vecr_{k}\,, \qquad q_j = \frac{1}{4} \left( 1 + \vec{f}_j \cdot \vecr \right)\,.
\end{equation}
Let $\vec{R}$ be the Bloch vector of the average state $\mathcal{O} \mleft( \rho \mright)$, such that \eqref{eqn: post meas} and \eqref{eqn: prob O} become
\begin{equation}\label{eqn: prob ito R}
    \vec{R} = \sum_{j=1}^{4} p_j \vecr_j\,, \qquad p_j = \frac{1}{4} \left(1 + \vec{f_j} \cdot \vec{R} \right)\,.
\end{equation}
Notice that the vector $\vec{R}$ represents the probability distribution $\{p_j\}$ in baryonic coordinates, using the measurement vectors $\{\vecr_j\}$ as a basis.
Concretely, our Gram matrix is
\begin{equation}
    G = \frac{1}{2} \left( \mathbb{J} +  H \right) \,, \qquad 
    H = \begin{pmatrix}
        1 & a & b & c
        \\
        a & 1 & c & b
        \\
        b & c & 1 & a
        \\
        c & b & a & 1
    \end{pmatrix}\,,
\end{equation}
where $\mathbb{J}$ is the matrix of ones. Its inverse is
\begin{equation}
    G^{-1} = \frac{1}{8} \left( \mathbb{J} + \frac{4}{t} M \right)\,,
\end{equation}
where
\begin{equation}
  M =
    \frac{1}{2}\begin{pmatrix}
        A^2 & A^2 - 2 \left( 1- a^2\right) &  A^2 - 2 \left( 1- b^2\right) &  A^2 - 2 \left( 1- c^2\right)
        \\
        A^2 - 2 \left( 1- a^2\right) & A^2 & A^2 - 2 \left( 1- b^2\right) & A^2 - 2 \left( 1- c^2\right)
        \\
        A^2 - 2 \left( 1- b^2\right) & A^2 - 2 \left( 1- c^2\right) & A^2 & A^2 - 2 \left( 1- a^2\right)
        \\
        A^2 - 2 \left( 1- c^2\right) & A^2 - 2 \left( 1- b^2\right) & A^2 - 2 \left( 1- a^2\right) & A^2
    \end{pmatrix}\,,   
\end{equation}
and
\begin{equation}
    t = \left( 1+ a \right) \left( 1+b \right) \left( 1+c \right)
\end{equation}
as before. To prove that this is the correct inverse, we expand
\begin{equation}
    H M = \begin{pmatrix}
        \kappa & \alpha & \beta & \gamma
        \\
        \alpha & \kappa & \gamma & \beta
        \\
        \beta & \gamma & \kappa & \alpha
        \\
        \gamma & \beta & \alpha & \kappa
    \end{pmatrix}\,,
\end{equation}
with
\begin{equation}
    \begin{aligned}
        \kappa &= -a \left( 1-a^2\right) - b \left( 1 - b^2\right) -c \left( 1- c^2\right)\,,
        \\
        \alpha &= -1 + a^2 - b- c + bc \left( b+c \right) \,,
        \\
        \beta &= -1 + b^2 - a- c + ac \left( a+c \right) \,,
        \\
        \gamma &= -1 + c^2 - a- b + ab \left( a+b \right)\,.   
    \end{aligned}
\end{equation}
We can break $\kappa$ down as
\begin{equation}
    \begin{aligned}
        \kappa &= -a \left( 1-a^2\right) - b \left( 1 - b^2\right) -c \left( 1- c^2\right) 
        \\
        &= \left( 1 + b + c \right) \left( 1-a^2\right) + \left( 1 + a + c\right) \left( 1 - b^2\right) + \left(1 + a + b \right)  \left( 1- c^2\right)
        \\
        &=   3 + 2 \left( a+b+c\right) - a^2 - b^2 -c^2 - ab \left( a+b\right) - ac \left( a+c\right) - bc \left( b+c\right) 
        \\
        &=   1 - a^2 - b^2 -c^2  + ab \left( 1+c\right) + ac \left( 1+b\right) + bc \left( 1+a\right) 
        \\
        &= 3 \left( ab + ac + bc\right) + 3abc = 3t\,
    \end{aligned}
\end{equation}
and write $\alpha$ as
\begin{equation}
    \begin{aligned}
        \alpha &=   -1 + a^2 - b- c + bc \left( b+c \right)   
        \\
        &= -1 - a \left( 1 + b +c \right) - b - c - bc \left( 1 +a \right)
        \\
        &= -ab - ac - bc - abc=-t\,.
    \end{aligned}
\end{equation}
By symmetry, we can conclude that $\alpha = \beta = \gamma = -t$. Since the elements in each column of $H$ and $M$ sum to zero, we have
\begin{equation}\label{eqn: inverse correct}
  \left[  G G^{-1} \right]_{ij}= \left[ \frac{1}{16}\left(  \mathbb{J} + H  \right) \left(  \mathbb{J} + \frac{4}{t} M \right) \right]_{ij} = \left[ \frac{1}{4} \left( \mathbb{J} + \frac{1}{t} HM \right)\right]_{ij} = \frac{1}{4} \left( 1 + \frac{1}{t} \left[ HM\right]_{ij} \right) = \delta_{i, \, j}
\end{equation}
for all $i$, $j \in \{1, \, \hdots, \, 4\}$. Using \eqref{eqn: nprod 1} and \eqref{eqn: nprod 2}, we can rewrite the matrix $M$ as
\begin{equation}
    M = \frac{A^2}{2} \begin{pmatrix}
      1 & -1 + \frac{4 l^2}{1+a} & -1 + \frac{4 l^2}{1+b}  & -1 + \frac{4 l^2}{1+c} \vspace{0.1 cm}
      \\
     -1 + \frac{4 l^2}{1+a} & 1 & -1 + \frac{4 l^2}{1+c} & -1 + \frac{4 l^2}{1+b}  \vspace{0.1 cm}
     \\
     -1 + \frac{4 l^2}{1+b} & -1 + \frac{4 l^2}{1+c} & 1 & -1 + \frac{4 l^2}{1+a} \vspace{0.1 cm}
     \\
     -1 + \frac{4 l^2}{1+c} & -1 + \frac{4 l^2}{1+b} & -1 + \frac{4 l^2}{1+a} & 1
    \end{pmatrix}\,.
\end{equation}
It will be useful to have an expression for $M^3$ too. It can be written as
\begin{equation}
    M^3 =  \frac{A^{6}}{2} 
    \begin{pmatrix}
     m & n & o & p
     \\
     n & m & p & o
     \\
     o & p & m & n
     \\
     p & o & n & m
    \end{pmatrix}\,,
\end{equation}
with elements
\begin{equation}
    \begin{aligned}
        m &= 1+ 12 l^4 \bigg( w - \frac{4}{t}   \bigg)
  + \frac{96 l^6}{t} \,,
        \\
        n &= -1 + \frac{12 l^2}{1+a}   +12 l^4 \bigg(  \frac{4 a}{ t} -w   \bigg)
         + \frac{16 l^6}{\left( 1+a\right)  } \bigg( \frac{1}{ \left( 1+a\right)^2} + \frac{3}{ \left( 1+b\right)^2} + \frac{3}{ \left( 1+c\right)^2} \bigg) \,,
        \\
       o &= -1 + \frac{12 l^2}{1+b}+12 l^4 \bigg(  \frac{4 b}{ t} -w   \bigg)  
         + \frac{16 l^6}{\left( 1+b\right)  } \bigg( \frac{1}{ \left( 1+b\right)^2} + \frac{ 3}{ \left( 1+c\right)^2} + \frac{ 3}{ \left( 1+a\right)^2} \bigg) \,,
        \\
        p &= -1 + \frac{12 l^2}{1+c}  +12 l^4 \bigg(  \frac{4 c} {t} -w   \bigg) 
     + \frac{16 l^6}{\left( 1+c\right) } \bigg( \frac{1}{ \left( 1+c\right)^2} + \frac{3}{ \left( 1+b\right)^2} + \frac{3}{ \left( 1+a\right)^2}\bigg)\,,
    \end{aligned}
\end{equation}
where
\begin{equation}
    w = \frac{1}{ \left( 1+a\right)^2}  + \frac{1}{ \left( 1+b\right)^2} + \frac{1}{ \left( 1+c\right)^2}\,,
\end{equation}
such that
\begin{equation}
    w t^2 = \left( 1+a\right)^2 \left( 1+b\right)^2 + \left( 1+a\right)^2 \left( 1+c\right)^2 + \left( 1+b\right)^2 \left( 1+c\right)^2\,.
\end{equation}
We have the relation
\begin{equation}\label{eqn: m n o p}
\begin{aligned}
    m + n - o - p &= \frac{2}{t^3} \left( t + 2 l^2 \left( \left( 1+b\right) \left( 1+c\right) - \left( 1+a\right) \left( 1+b\right) - \left( 1+a\right) \left( 1+c\right) \right) \right)^3
    \\
    &= \frac{2}{t^3} \left( t - 2 l^2 \left( \left(1 + ab + ac + bc \right) + 2 \left( a - bc \right)\right)  \right)^3
    \\
    &= \frac{2}{t^3} \left( 2 A^2 l^2 - 2 l^2 \left( A^2 - 2 \left( a - bc \right)\right)  \right)^3
    \\
    &= \frac{2}{t^3} \left( \frac{8l^4 A^2}{1+a} \right)^3 = \frac{128 l^6}{\left( 1+a\right)^3}\,,
    \end{aligned}
\end{equation}
where in the third line we use \eqref{eqn: l squared} and in the last line we use \eqref{eqn: other l}. Similarly, we have 
\begin{equation}
    \begin{aligned}
     m - n + o - p &=   \frac{128 l^6}{ \left( 1+b \right)^3 }\,,
     \\
     m - n - o + p &=   \frac{128 l^6}{ \left( 1+c \right)^3 }\,.
    \end{aligned}
\end{equation}
We can now focus on characterizing the space of allowed probability distributions. Considering only pure states $\vecr$, which form the boundary, the allowed distributions $\{q_j\}$ are defined by 
\begin{equation}
    \abs{\vecr}^2 =  \sum_{i, \, j =1}^{4} q_i q_j \, \vecr_i \cdot \vecr_j  = \sum_{i, \, j =1}^{4} q_i \left[ H \right]_{ij} q_j = 1\,.
\end{equation}
Using \eqref{eqn: q ito p}, in terms of $\{p_j\}$ we find
\begin{equation}
\abs{\vecr}^2 =    4 \sum_{k, \, l=1}^{4} p_k \left( \sum_{i,\, j=1}^{4} \left[ G^{-1}\right]_{ki} \left[ H\right]_{ij} \left[ G^{-1} \right]_{jl} \right)p_l = \frac{1}{ t^2} \sum_{k, \, l=1}^{4} p_k \left[M H M \right]_{kl} p_l= 1\,.
\end{equation}
From \eqref{eqn: inverse correct}, we find 
\begin{equation}
    H M = t \left( 4 \id - \mathbb{J}\right)\,,
\end{equation}
such that
\begin{equation}
   \left[M H M \right]_{kl} = 4t \left[ M\right]_{kl}\
\end{equation}
and
\begin{equation}
    \abs{\vecr}^2 =  \frac{4}{t} \sum_{k, \,  l=1}^{4} p_k \left[ M \right]_{kl} p_l= 1\,.
\end{equation}
This gives us a quadratic expression for the boundary of the set of possible probability distributions $\{p_j\}$ as Cartesian coordinates. Now we want to characterise the space of the average post-measurement states $\vec{R}$. Using \eqref{eqn: prob ito R}, and the fact that $\sum_{k=1}^{4} \left[ M\right]_{kl}=\sum_{l=1}^{4} \left[ M\right]_{kl} =0$, we find
\begin{equation}\label{eqn: R ellipsoid eq}
   \frac{1}{4t} \sum_{k, \, l=1}^{4} \left( \vec{f}_k \cdot \vec{R} \right) \, \left[ M \right]_{kl} \, \left( \vec{f}_l \cdot \vec{R} \right) =  \vec{R} \cdot \left( \frac{1}{4t} \sum_{k, \, l=1}^{4} \left[ M \right]_{kl} \vec{f}_{k} \otimes \vec{f}_l   \right) \cdot \vec{R}= 1\,,
\end{equation}
where $\otimes$ denotes the outer product. We now have a quadratic equation for the states $\vec{R}$. Let's define the $3 \times 3$ matrix
\begin{equation}
    N =\frac{1}{4t} \sum_{k, \, l=1}^{4} \left[ M \right]_{kl} \vec{f}_{k} \otimes \vec{f}_l\,.
\end{equation}
Using \eqref{eqn: f eqs}, we have
\begin{equation}
    \vec{f}_k = \frac{1}{4} \sum_{l=1}^{4} \left[ \mathbb{J} + \frac{4}{t} M \right]_{k l} \vecr_l = \frac{1}{t} \sum_{l=1}^{4} \left[M \right]_{kl} \vecr_l\,,
\end{equation} 
so we can rewrite $N$ as 
\begin{equation}
    N = \frac{1}{4t^3} \sum_{k, \, l=1}^{4} \left[M^3 \right]_{k l} \vecr_{k} \otimes \vecr_l\,.
\end{equation}
\begin{figure}
\centering
    \begin{subfigure}{0.45\textwidth}
    \includegraphics[width=\textwidth]{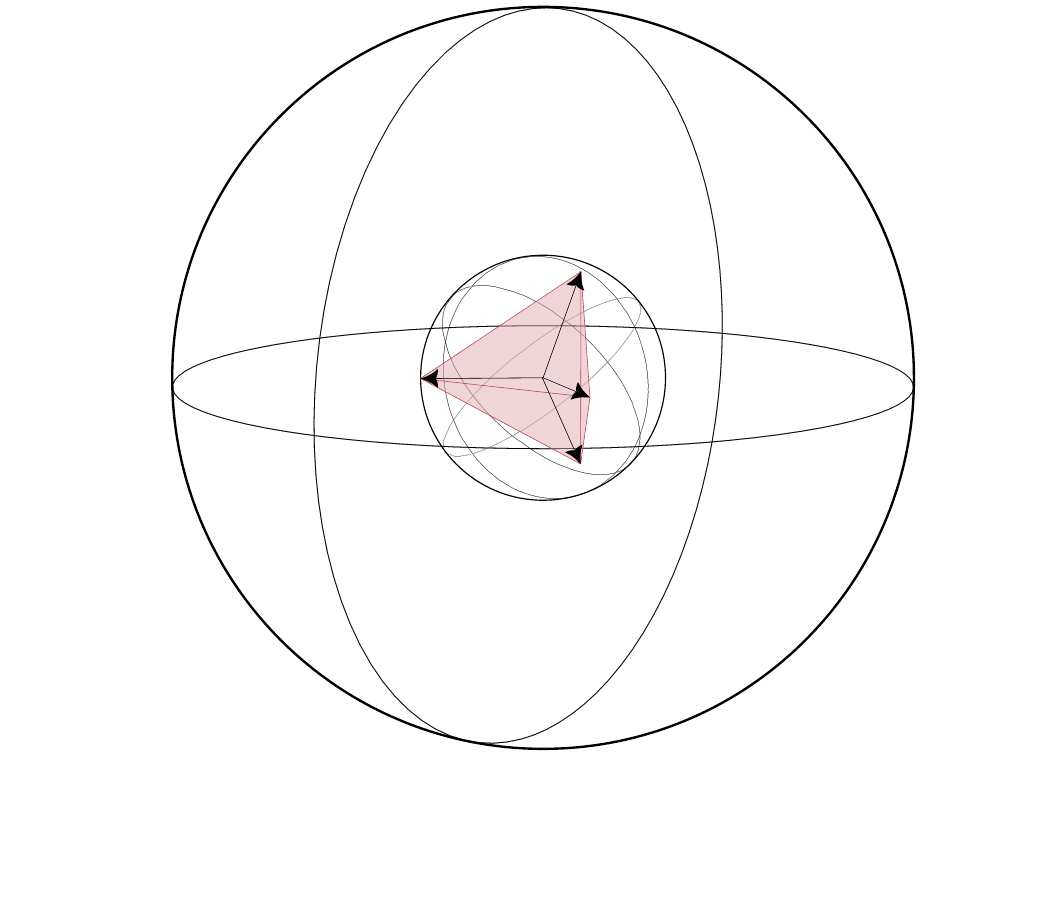}
    \caption{SIC POVM.}
    \label{fig: SIC after}
\end{subfigure}
\hfill
\begin{subfigure}{0.45\textwidth}
    \includegraphics[width=\textwidth]{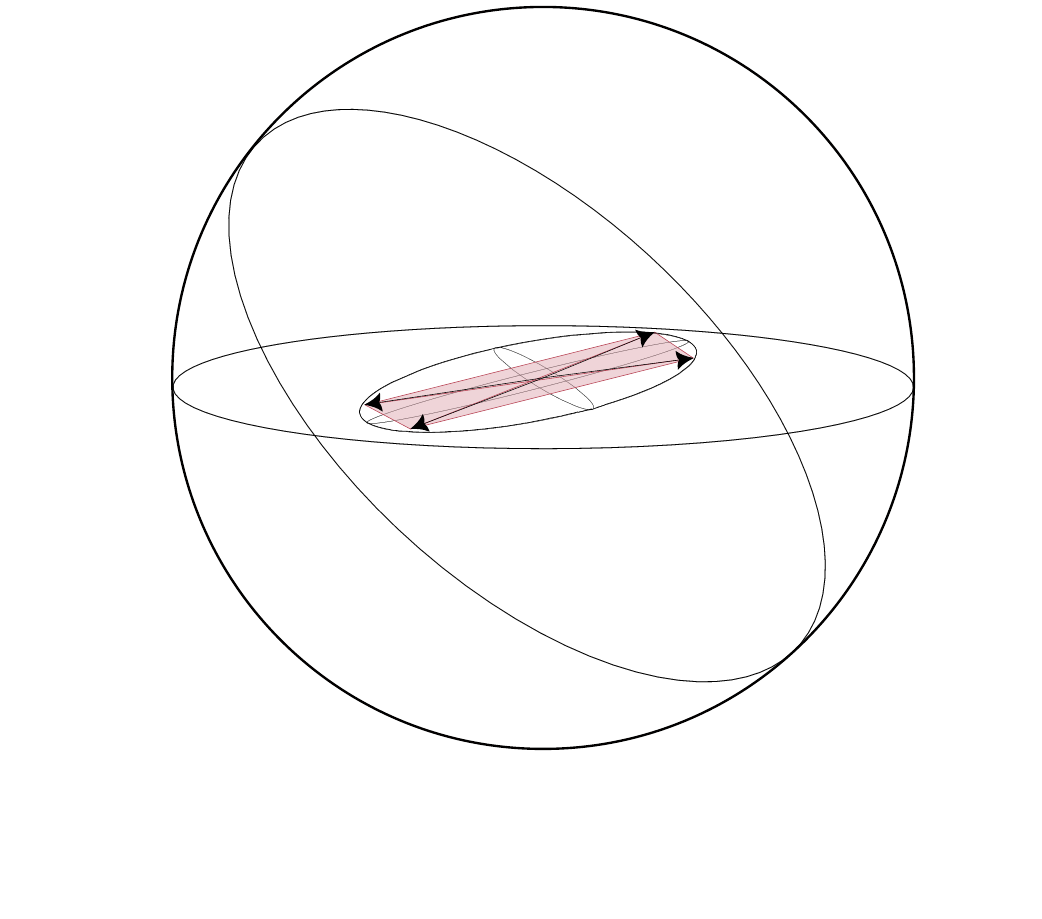}
    \caption{A general qubit MIC.}
    \label{fig: gen after}
\end{subfigure}
    \caption{The space of average post-measurement states $\vec{R}$ for unbiased MICs in dimension two.}
    \label{fig: POVMs after}
\end{figure}
We can see that $\left( \vecr_1 + \vecr_2\right)$ is an unnormalised eigenvector of $N$, as
\begin{equation}
    \begin{aligned}
        N \left( \vecr_1 + \vecr_2 \right) &= \frac{1}{4 t^3} \sum_{k, \, l=1}^{4} \left[ M^3 \right]_{kl} \left( \vecr_l \cdot \vecr_1 + \vecr_l \cdot \vecr_2 \right) \vecr_k
        \\
        &= \frac{A^6 \left( 1+a\right)}{8 t^3}    \left(m+n - o - p\right)\left( \vecr_1 + \vecr_2 - \vecr_3 - \vecr_4 \right) 
        \\
        &= \frac{A^6 \left( 1+a\right)}{4 t^3}    \left(m+n - o - p\right)\left( \vecr_1 + \vecr_2 \right)
        \\
        &= \frac{32 A^6 l^6}{\left(1+a \right)^2 t^3} \left( \vecr_1 + \vecr_2 \right) = \frac{4}{\left(1+a \right)^2} \left( \vecr_1 + \vecr_2 \right)\,,
    \end{aligned}
\end{equation}
where in the last line we use \eqref{eqn: l squared} and \eqref{eqn: m n o p}. Similarly, $\left(\vecr_1 + \vecr_3 \right)$ and $\left( \vecr_2 + \vecr_3\right)$ are also eigenvectors of $N$, with eigenvalues $\frac{4}{\left( 1+b\right)2}$ and $\frac{4}{\left( 1+c \right)^2}$. We can conclude, then, that $N$ is positive semidefinite, so \eqref{eqn: R ellipsoid eq} defines an ellipsoid for the boundary of the space of vectors $\vec{R}$. The ellipsoid is centred at the origin, and its principal axes lie along the vectors $\left( \vecr_1 + \vecr_2\right)$, $\left(\vecr_1 + \vecr_2 \right)$ and $\left(\vecr_2 + \vecr_3 \right)$. Denoting the eigenvalues as
\begin{equation}
    \lambda_a = \frac{4}{\left( 1+ a\right)^2}\,, \qquad \lambda_b = \frac{4}{\left( 1+b\right)^2}\,, \qquad \lambda_c = \frac{4}{\left(1+c \right)^2}\,,
\end{equation}
the length of each semi-axis is given by $\frac{1}{\sqrt{\lambda_j}}$, for $j \in \{a, \, b, \, c\}$. The volume of the ellipsoid is then 
\begin{equation}\label{eqn: vol after}
    \textnormal{Vol} = \frac{4 \pi}{3} \frac{1}{\sqrt{\lambda_a \lambda_b \lambda_c }} = \frac{\pi}{6} \left( 1+a\right) \left( 1+b\right) \left( 1+c\right) = \frac{\pi}{6} t\,.
\end{equation}
Finally, the states that live inside the polytope $\mathcal{P}= \conv \mleft( \{\vecr_j\}\mright)$ are mapped by $\mathcal{O} \mleft( \cdot \mright)$ to another, smaller polytope, $\mathcal{S}=\conv \mleft( \{\vec{R}_j\}\mright)$, where 
\begin{equation}
    \begin{aligned}
        \vec{R}_1 &= \frac{1}{4} \left( \vecr_1 + a \vecr_2 + b \vecr_3 + c \vecr_4 \right)\,,
        \\
        \vec{R}_2 &= \frac{1}{4} \left( a \vecr_1 + \vecr_2 + c \vecr_3 + b \vecr_4 \right)\,,
        \\
        \vec{R}_3 &= \frac{1}{4} \left( b \vecr_1 + c \vecr_2 + \vecr_3 + a \vecr_4\right)\,,
        \\
      \vec{R}_4 &=  \frac{1}{4} \left( c \vecr_1 + b \vecr_2 + a \vecr_3 + \vecr_4\right)\,.
    \end{aligned}
\end{equation}
The ellipsoid of states $\vec{R}$, which encloses the tetrahedron $\mathcal{S}$, is shown in Figure \ref{fig: POVMs after} for two different unbiased MICs. Note that, as discussed in \cite{Appleby2011-APPPOQ}, the space of states $\{\vec{R}_j\}$ for the SIC POVM is the in-sphere of the tetrahedron $\mathcal{P}= \conv \mleft( \{\vecr_j\} \mright)$.

\subsection{Maximal randomness}
Given an unbiased extremal measurement $\m$ in dimension two with $n$ outcomes, represented by the measurement vectors $\{\vecr_j\}$, the intrinsic randomness generated by a pure state $\vecr$ is given simply by the probability of obtaining the most likely outcome,
\begin{equation}
    \pg \mleft( \vecr, \, \m \mright) = \max_{j} \frac{1}{n} \left( 1 + \vecr \cdot \vecr_j\right)\,.
\end{equation}
We can then find the \emph{maximal} intrinsic randomness of $\m$ by minimising the above over all states $\vecr$, 
\begin{equation}
    \pgs \mleft( \m\mright) = \min_{\vecr} \; \max_{j} \frac{1}{n} \left( 1 + \vecr \cdot \vecr_j\right)\,.
\end{equation}
One approach to tackle this min-max problem is, for each $j \in \{1, \, \hdots, \, n\}$, to solve the minimisation in the region of the Bloch sphere for which
\begin{equation}
    \vecr \cdot \vecr_j \geq \vecr \cdot \vecr_i \qquad \textnormal{for all} \;\; i  \in \{1, \, \hdots, \, n\}\setminus j\,.
\end{equation}
Then for a fixed $j$, following the notation of \eqref{eqn: opt stand}, the optimisation problem  is
\begin{equation}
    \begin{aligned}
         \underset{\vecr}{\text{minimize}}  \, \qquad & f_0 \mleft( \vecr \mright) && {}
         \\
          \text{subject to} \qquad & g \mleft( \vecr \mright) = 0 \qquad &&  
        \\
          \qquad & h_{k} \mleft( \vecr \mright) \leq 0 \qquad &&  \textnormal{for all} \;\;  k \in \{1, \, \hdots, \, n\}\setminus j\,,
         \end{aligned}
\end{equation}
where
\begin{equation}
    f_0 \mleft( \vecr\mright) = \frac{1}{2} \left( 1 + \vecr \cdot \vecr_j\right)\,, \qquad g \mleft(\vecr \mright) = 1 - \abs{\vecr}^2\,, \qquad h_k \mleft(\vecr \mright) = \vecr \cdot \left( \vecr_k - \vecr_j\right)\,.
\end{equation}
Since the equality condition $g \mleft( \vecr\mright)$ is quadratic, and therefore not affine, the KKT conditions are necessary but they mightn't be sufficient. Introducing the dual variables $\nu$ and $\{\lambda_k\}$ for $k \in \{1, \, \hdots, \, n\}\setminus j$, such that 
\begin{equation}
    \lambda_k \geq 0 \qquad \textnormal{for all} \;\; k \in \{1, \, \hdots, \, n\}\setminus j\,,
\end{equation}
the complementary slackness and stationarity conditions \eqref{eqn: KKT comp slack} and \eqref{eqn: KKT stat} become
\begin{equation}\label{eqn: rand comp}
    \lambda_{k} h_{k} \mleft( \vecr \mright) = 0 \qquad \textnormal{for all} \;\; k \in \{1, \, \hdots, \, n\}\setminus j
\end{equation}
and
\begin{equation}\label{eqn: rand stat}
    \frac{1}{2} \vecr_j -2 \nu \vecr + \underset{k \neq j}{\sum_{k=1}^{n}} \lambda_{k} \left( \vecr_k - \vecr_j\right)=0\,.
\end{equation}
Multiplying every term by $\vecr$ and applying \eqref{eqn: rand comp}, we find $\nu=\frac{1}{4} \vecr \cdot \vecr_j$, so \eqref{eqn: rand stat} can be simplified to 
\begin{equation}\label{eqn: rand stat 2}
    \frac{1}{2} \vecr_j - \frac{1}{2} \left( \vecr \cdot \vecr_j \right) \vecr + \underset{k \neq j}{\sum_{k=1}^{n}} \lambda_{k} \left( \vecr_k - \vecr_j\right)=0\,.
\end{equation}

\subsubsection{Two outcomes}\label{ssec: two out max}
The case where $n=2$ is trivial, as we know that all states that are unbiased with respect to the measurement basis produce maximal randomness, so $\pgs \mleft(  \m \mright)= \frac{1}{2}$. 

\subsubsection{Three outcomes}\label{ssec: trine max}
The trine measurement is the only unbiased extremal rank-one measurement with three outcomes in dimension two. The two states that are unbiased to the plane of the measurement vectors produce maximal randomness, so $\pgs \mleft( \m \mright)= \frac{1}{3}$.

\subsubsection{Four outcomes}\label{ssec: qubit MIC max}
We will find the optimal state, and the maximal randomness, for MICs in dimension two by solving the KKT conditions \eqref{eqn: rand comp} and \eqref{eqn: rand stat} for four different cases of the dual variables $\{\lambda_j\}$.

\paragraph{All $\lambda_j$ are non-zero.} If all variables $\lambda_j$ are non-zero then, from \eqref{eqn: rand comp}, we must have $\vecr \cdot \vecr_1= \vecr \cdot \vecr_j$ for $j \in \{2, \,  3, \, 4\}$. Since $\sum_{j=1}^{4} \vecr_j =0$, this implies that $\vecr \cdot \vecr_j=0$ for all $j \in \{1, \, \hdots, \, 4\}$. This is impossible, however, as the vectors $\{\vecr_j\}$ cannot live on a common plane. 

\paragraph{Two $\lambda_j$ are non-zero.} Let's arbitrarily choose $\lambda_4=0$ such that, from \eqref{eqn: rand comp}, we have $\vecr \cdot \vecr_1 = \vecr \cdot \vecr_2= \vecr \cdot \vecr_3$. Then the condition \eqref{eqn: rand stat 2} becomes 
\begin{equation}
    \left( \frac{1}{2} - \lambda_2 - \lambda_3 \right) \vecr_1 + \lambda_2 \vecr_2 + \lambda_3 \vecr_3 - \frac{1}{2} \left( \vecr \cdot \vecr_1\right) \vecr =0\,.
\end{equation}
We can then write $\vecr$ as 
\begin{equation}\label{eqn: step 1}
    \frac{1}{2} \left( \vecr \cdot \vecr_1\right) \vecr = \left( \frac{1}{2} - \lambda_2 - \lambda_3\right) \vecr_1 + \lambda_2 \vecr_2 + \lambda_3 \vecr_3\,. 
\end{equation}
There are only two pure states $\vecr$ that satisfy $\vecr \cdot \vecr_1 = \vecr \cdot \vecr_2= \vecr \cdot \vecr_3$, as these conditions define a line that runs through the Bloch sphere, intersecting the origin and piercing the surface on both sides. The states are $\pm \vec{n}_{\left( 123\right)}$, from \eqref{eqn: normal vecs} or \eqref{eqn: l n bary}.  We can discard $- \vec{n}_{\left( 123\right)}$, since we need that $\vecr \cdot \vecr_1$ is positive, so we set $\vecr= \vec{n}_{\left( 123\right)}$ such that 
\begin{equation}
    \vecr \cdot \vecr_1 = \vecr \cdot \vecr_2 = \vecr \cdot \vecr_3 = l\,.
\end{equation}
Comparing \eqref{eqn: step 1} to \eqref{eqn: l n bary}, we define
\begin{equation}
    \lambda_2 = \frac{1-b^2}{4 A^2} >0 \,, \qquad \lambda_3 = \frac{1-a^2 }{4A^2} > 0\,.
\end{equation}
Noting that 
\begin{equation}
    \frac{1}{2} - \lambda_2 - \lambda_3 = \frac{2 A^2 - 2 + b^2 +c^2}{2 A^2} = \frac{1-c^2}{2A^2}\,,
\end{equation}
we conclude that $\vecr =\vec{n}_{\left( 123 \right)}$ satisfies the KKT conditions. This state gives a guessing probability of 
\begin{equation}
    \pg \left( \vecr, \, \m \right) = \frac{1+l}{4}\,.
\end{equation}
By choosing $\lambda_2$ or $\lambda_3$ instead to be zero, we find that $\vec{n}_{\left( 124\right)}$ and $\vec{n}_{\left( 134\right)}$ also satisfy the KKT conditions, and give the same guessing probability.

\paragraph{One $\lambda_j$ is non-zero.} We choose $\lambda_2 \neq 0$ such that $\vecr \cdot \vecr_1 = \vecr \cdot \vecr_2$. The condition \eqref{eqn: rand stat 2} becomes
\begin{equation}
    \left( \frac{1}{2} - \lambda_2\right) \vecr_1 + \lambda_2 \vecr_2 - \frac{1}{2} \left(\vecr \cdot \vecr_1 \right) \vecr =0 \,.
\end{equation} 
Multiplying separately by $\vecr_1$ and $\vecr_2$, we have
\begin{equation}
    \begin{aligned}
      \frac{1}{2} - \lambda_2 + \lambda_2 a - \frac{1}{2} \left( \vecr \cdot \vecr_1\right)^2 &=0\,,
      \\[5pt]
      \left( \frac{1}{2} - \lambda_2\right) a + \lambda_2 - \frac{1}{2}\left( \vecr \cdot \vecr_1\right)^2 &=0\,.
    \end{aligned}
\end{equation}
The unique solution is
\begin{equation}
    \lambda_2 = \frac{1}{4}\,, \qquad \left( \vecr \cdot \vecr_1\right)^2 = \frac{1+a}{2}\,.
\end{equation}
The condition \eqref{eqn: rand stat 2} then gives us the following expression for $\vecr$
\begin{equation}
    \vecr = \frac{\vecr_1 + \vecr_2}{\sqrt{ 2 \left(1 + a \right) }}\,,
\end{equation}
which satisfies all the necessary conditions. It gives a guessing probability of 
\begin{equation}
    \pg \mleft( \vecr, \, \m \mright) = \frac{1}{4} \left(  1 + \sqrt{ \frac{1 + a}{2} } \right)\,.
\end{equation}
Had we instead chosen $\lambda_3$ or $\lambda_4$ to be nonzero, we would have reached equivalent solutions with $\vecr_2$ exchanged with $\vecr_3$ or $\vecr_4$ and $a$ exchanged with $b$ or $c$.

\paragraph{All $\lambda_j$ are zero.} 
If all $\lambda_j$ are zero, the condition \eqref{eqn: rand stat} gives us $\vecr= \vecr_1$. However, this would give us a guessing probability of a half, which is the highest possible guessing probability for an unbiased qubit MIC.
\subsection*{}
We will now show that the state 
\begin{equation}
    \vecr =\frac{\vecr_1 + \vecr_2}{ \sqrt{ 2 \left( 1 + a \right) }}\,
\end{equation}
gives a higher guessing probability than the normal states $\vec{n}_{\left(123 \right) }$, $\vec{n}_{\left(124\right)}$ and $\vec{n}_{\left( 134\right)}$, by showing that 
\begin{equation}
  l <  \sqrt{\frac{1 + a}{2}}\,.  
\end{equation}
Using \eqref{eqn: t} and \eqref{eqn: l squared}, we have
\begin{equation}
\begin{aligned}
    \frac{1+a}{2} - l^2 &= \frac{\left( 1+a\right) \left( A^2 - \left( 1+b\right) \left(1+c\right) \right) }{2A^2} = \frac{\left( 1+a\right)  \left( 1-a^2\right) }{2A^2} > 0\,,
    \end{aligned}
\end{equation}
where in the second equality we use \eqref{eqn: Asquared}.
Since $l$ and $A$ are symmetric under any swapping of $a$, $b$ and $c$, we also have 
\begin{equation}
    l < \sqrt{ \frac{1+b}{2}}\,, \qquad l < \sqrt{ \frac{1+c}{2}}\,, 
\end{equation}
so the states $\vec{n}_{\left(123 \right)}$, $\vec{n}_{\left( 124\right)}$ and $\vec{n}_{\left(134 \right)}$ are the unique optimal states in the region where $\vecr \cdot \vecr_1 \geq \vecr \cdot \vecr_i$ for $i \in \{ 2, \, 3, \,  4\}$. Equivalent results holds if we choose any vector $\vecr_j$ other than $\vecr_1$ to have the largest overlap with $\vecr$. We conclude, then, that the four normal states $\vec{n}_{\left( 123\right)}$, $\vec{n}_{\left( 124\right)}$, $\vec{n}_{\left( 134\right)}$ and $\vec{n}_{\left( 234\right)}$ are the optimal states for randomness across the whole state space, and the optimal guessing probability is 
\begin{equation}
    \pgs \left( \m \right) = \frac{1}{4} \left( 1+l\right)\,.
\end{equation}

\subsection{Randomness for any state}
We will solve the optimisation problem \eqref{eqn: SDP rand} for any mixed qubit state $\vecr$ and any unbiased extremal qubit measurement.

\subsubsection{Two outcomes}\label{ssec: 2 outcomes any}
Let's fix the measurement basis as $\vecr_1 = \vec{m}, \, \vecr_2 = - \vec{m}$ for some normalised vector $\vec{m}$, and let $\vec{n}$ be a normal unit vector, such that $\vec{m} \cdot \vec{n}=0$. Without loss of generality, we can represent our qubit state as
\begin{equation}
    \vecr = m \vec{m} + p \vec{n}\,, \qquad m^2 + p^2 < 1\,.
\end{equation}
The eavesdropper could decompose this state as
\begin{equation}
    \vecr = p_1 \vec{q}_1 + p_2 \vec{q}_2\,,
\end{equation}
where 
\begin{equation}
    p_1 = \frac{1}{2} \left( 1 + \frac{m}{\sqrt{1-p^2}}\right)\,, \qquad p_2 = \frac{1}{2} \left( 1 - \frac{m}{\sqrt{1-p^2}}\right)\,
\end{equation}
and 
\begin{equation}
    \vec{q}_1 = \sqrt{1-p^2} \, \vec{m} + p \vec{n}\,, \qquad \vec{q}_2 = -\sqrt{1-p^2} \, \vec{m} + p \vec{n}\,.
\end{equation}
This gives the lower bound 
\begin{equation}
    \pg \mleft( \vecr, \, \m \mright) \geq \frac{1}{2} \left( 1 + \sqrt{1-p^2} \right)\,.
\end{equation}
Then consider the following dual variable for \eqref{eqn: rand dual},
\begin{equation}
    X = \frac{1}{4} \frac{ \left( \sqrt{1+p} + \sqrt{1-p} \right)^2}{\sqrt{1-p^2}} \left( \id - g \vec{n} \cdot \vec{\sigma} \right)\,, \qquad g = \frac{\sqrt{1+p} - \sqrt{1-p}}{\sqrt{1+p} + \sqrt{1-p}}\,.
\end{equation}
This is a valid variable, since
\begin{equation}
    \begin{aligned}
   &X - M_1 =    \!\! &&X - \frac{1}{2} \left( \id + \vec{m} \cdot \vec{\sigma} \right) = \frac{1}{2 \sqrt{1-p^2}} \left( \id - \vec{q}_1 \cdot \vec{\sigma} \right) \sgeq 0\,,
        \\
    &X - M_2 =  \!\!     &&X - \frac{1}{2} \left( \id - \vec{m} \cdot \vec{\sigma} \right) = \frac{1}{2 \sqrt{1-p^2}} \left( \id - \vec{q}_2 \cdot \vec{\sigma} \right) \sgeq 0\,.
    \end{aligned}
\end{equation}
We find the upper bound 
\begin{equation}
    \pg \mleft( \vecr, \, \m \mright) \leq \tr \mleft( X \rho\mright) = \frac{1}{4} \frac{ \left( \sqrt{1+p} + \sqrt{1-p} \right)^2}{\sqrt{1-p^2}} \left( 1 - g \vec{n} \cdot \vecr \right) = \frac{1}{2} \left( 1 + \sqrt{1-p^2} \right)\,.
\end{equation}
Since the upper and lower bounds match, we conclude that
\begin{equation}\label{eqn: rand qubits proj}
  \pg \mleft( \vecr, \, \m \mright) =  \frac{1}{2} \left( 1 + \sqrt{1-p^2} \right)\,.
\end{equation}
When $m=0$, \eqref{eqn: rand qubits proj} returns the maximal intrinsic randomness that can be generated from the state $\vecr$ using a projective measurement \cite{Meng_2023}.

\subsubsection{Three outcomes}\label{ssec: 3 outcomes any}
Let's consider a vector $\vec{n}$ that's normal to the plane of the trine measurement vectors, and assume that
\begin{equation}
    \vecr \cdot \vec{n} =p\,, \qquad 0 \leq p \leq 1\,.
\end{equation}
Note that the trine vectors $\{\vecr_j\}$ satisfy
\begin{equation}
    \vecr_1 \cdot \vecr_2 = \vecr_1 \cdot \vecr_3 = \vecr_2 \cdot \vecr_3 = - \frac{1}{2}\,. 
\end{equation}
We will construct a three-dimensional shape $\Delta$ in the Bloch ball by contracting and translating the vectors $\{\vecr_j\}$ so they lie a distance $p$ along the vector $\vec{n}$, and still touch the inside of the Bloch sphere. The new vectors will be
\begin{equation}
    \vec{s}_j = p \vec{n} + \sqrt{1-p^2} \vec{r}_j \qquad \textnormal{for all} \;\; j \in \{1, \, 2, \, 3\}\,.
\end{equation}
We can then define the shape $\Delta$ as the set
\begin{equation}
    \Delta = \conv \mleft( \{\vec{s}_j\}   \mright) \quad \text{for all } \quad 0 \leq p \leq 1\,. 
\end{equation}
We will now solve \eqref{eqn: SDP rand} for two different cases (cf. \cite{Avesani_2022}, and Figure 1 therein, for results that apply, in our randomness scenario, to states $\vecr$ that live in the plane spanned by the trine measurement vectors).

\paragraph{States inside $\Delta$.} If the state $\vecr$ lives in $\Delta$ then by definition it can be written as 
\begin{equation}
    \vecr = \sum_{j=1}^{3} p_j \vec{s}_j
\end{equation}
for some probability distribution $\{p_j\}$. This gives us a lower bound for the guessing probability, 
\begin{equation}
    \pg \mleft( \vecr, \, \m \mright) \leq \frac{1}{3} \sum_{j=1}^{3} p_j \left( 1 + \vecr_j \cdot \vec{s}_j \right) = \frac{1}{3} \left( 1 + \sqrt{1-p^2} \right)\,.
\end{equation}
Now consider the dual variable
\begin{equation}
    X = \frac{1}{3} \left(1 + \frac{1}{\sqrt{1-p^2}} \right) \left( \id - \frac{p}{1 + \sqrt{1-p^2}}  \,\vec{n} \cdot \vec{\sigma} \right)\,.
\end{equation}
This variable satisfies all the constraints of \eqref{eqn: rand dual}, because 
\begin{equation}
    X - M_j = X - \frac{1}{3} \left( \id + \vecr_j \cdot \vec{\sigma} \right) = \frac{1}{3 \sqrt{1-p^2}} \left( \id - \vec{s} _j \cdot \vec{\sigma}\right) \sgeq 0 \qquad \textnormal{for all} \;\; j \in \{1, \, 2, \, 3\}\,.
\end{equation}
We find the upper bound
\begin{equation}
    \pg \mleft( \vecr, \, \m \mright) \leq \tr \mleft( X \rho \mright) = \frac{1}{3} \left(1 + \frac{1}{\sqrt{1-p^2}} \right) \left( 1 - \frac{p}{1 + \sqrt{1-p^2}} \, \vec{n} \cdot \vecr \right) = \frac{1}{3} \left( 1 + \sqrt{1-p^2}\right)\,.
\end{equation}
Since the upper and lower bounds match, we conclude that 
\begin{equation}
    \pg \mleft( \vecr, \, \m \mright) = \frac{1}{3} \left( 1 + \sqrt{1-p^2}\right)
\end{equation}
for all $\vecr \in \Delta$.

\paragraph{States outside $\Delta$.}
It will be helpful to introduce the following three unit vectors, 
\begin{equation}
    \vec{\omega}_{12 } = \vecr_1+ \vecr_2\,, \qquad \vec{\omega}_{13} = \vecr_1 + \vecr_3\,,
    \qquad 
    \vec{\omega}_{23 } = \vecr_2 + \vecr_3\,,
\end{equation}
with the shorthand $\{\vec{\omega}\} = \{\vec{\omega}_{12}, \, \vec{\omega}_{13}, \, \vec{\omega}_{23} \}$.
Note that $\vec{\omega}_{12}$, $\vec{\omega}_{13}$ and $\vec{\omega}_{23}$ are perpendicular to $\left(\vec{s}_2 - \vec{s}_1\right)$, $\left(\vec{s}_1- \vec{s}_3\right)$ and $\left(\vec{s}_3 - \vec{s}_2\right)$ respectively, and that all $\vec{w} \in \{\vec{\omega}\}$ are perpendicular to $\vec{n}$. Since all $\vec{\omega} \in \{\vec{\omega}\}$ point out of the Bloch sphere, any state that satisfies $\vecr \cdot \vec{n} \geq 0$ and does not live inside $\Delta$ must fulfil exactly one of the following inequalities,
\begin{align}
    \vecr \cdot \vec{\omega}_{12} & > \vec{s}_1 \cdot  \vec{\omega}_{12 } = \vec{s}_2 \cdot  \vec{\omega}_{12}\,, \label{eqn: qq eq}
    \\
    \vecr \cdot \vec{\omega}_{13} & > \vec{s}_1 \cdot  \vec{\omega}_{13} = \vec{s}_3 \cdot  \vec{\omega}_{13}\,,
    \\
    \vecr \cdot \vec{\omega}_{23} & > \vec{s}_2 \cdot  \vec{\omega}_{23 } = \vec{s}_3 \cdot  \vec{\omega}_{23}\,.
\end{align}
Without loss of generality, we will work in the region where the condition \eqref{eqn: qq eq} holds, so we define
\begin{equation}\label{eqn: q trine}
    q = \vecr \cdot \vec{\omega}_{12 }> \vec{s}_1 \cdot  \vec{\omega}_{12 } = \frac{1}{2}\sqrt{1-p^2}\,.
\end{equation}
We are now ready to give the optimal decomposition for $\vecr$. We will contract and translate the line $\conv \mleft( \vec{s}_1, \, \vec{s}_2 \mright)$ so that the new line remains at a distance $p$ along $\vec{n}$, but is now also at a distance $q$ along $\vec{\omega}_{12}$, with both vertices still touching the Bloch sphere. The new vertices are 
\begin{equation}\label{eqn: t vertices}
    \vec{t}_j  = p \vec{n} + \left(q-\frac{f}{2} \right) \vec{w}_{12 } + f \vec{r_j} \qquad \textnormal{for all} \;\; j \in \{1,\, 2\}\,,
\end{equation}
where 
\begin{equation}
    f = 2 \sqrt{ \frac{ 1 - p^2 -q^2  }{ 3 } }\,.
\end{equation}
Note that in the limit where $q \rightarrow \frac{1}{2}\sqrt{1-p^2}$, each $\vec{t}_j$ reduces to $\vec{s}_j$. We can decompose $\vecr$ into $\{\vec{t}_j\}$ according to some probability distribution. This gives us the lower bound
\begin{equation}
    \pg \mleft( \vecr, \, \m \mright) \geq \frac{1}{3}\sum_{j=1}^{2} p_j \left( 1 + \vecr_j \cdot \vec{t}_j \right) = \frac{1}{3} \left( 1 + \frac{1}{2} \left( q + \sqrt{3 \left( 1 - p^2 - q^2 \right)}\right) \right)\,.
\end{equation}
Now consider the dual variable 
\begin{equation}
    X = \frac{1}{3} \left( 1 + \frac{1}{f} \right) \left( \id - \frac{1}{1+f} \vec{v} \cdot \vec{\sigma} \right)\,, \qquad \vec{v} = p \vec{n} + \left( q - \frac{f}{2}\right) \vec{\omega}_{12}\,.
\end{equation}
It satisfies
\begin{equation}
    X - M_j = X - \frac{1}{3} \left( \id + \vecr_j \cdot \vec{\sigma} \right) = \frac{1}{3f} \left( \id - \vec{t}_j \cdot \vec{\sigma} \right) \sgeq 0 \qquad \textnormal{for all} \;\; j \in \{1, \, 2\}\,.
\end{equation}
We also have
\begin{equation}
    X - M_3 = \frac{1}{3f} \left( \id - \left( \vec{v} + f \vecr_3 \right) \cdot \vec{\sigma} \right) = \frac{1}{3f} \left( \id - \left( \vec{v} - f \vec{\omega}_{12}\right) \right) \sgeq 0\,,
\end{equation}
where we use $\vecr_3 = - \vec{\omega}_{12}$, and where positive semidefiniteness holds because, using \eqref{eqn: q trine}, we have
\begin{equation}
    \abs{\vec{v} - f \vec{\omega}_{12}}^2 = 1 - 3f \left( q - \frac{f}{2}\right) = 1 - \frac{f \left( 2 q + \sqrt{1-p^2}\right)}{q + \frac{f}{2}} \left( 2q - \sqrt{1-p^2}\right) < 1\,.
\end{equation}
All the constraints of \eqref{eqn: rand dual} are then satisfied, so we can upper bound the guessing probability by
\begin{equation}
    \pg \mleft( \vecr, \, \m\mright) \leq \tr \mleft( X \rho \mright) =\frac{1}{3} \left( 1 + \frac{1}{f}\right) \left( 1 - \frac{1}{1+f} \vec{v} \cdot \vecr \right) = \frac{1}{3} \left( 1 + \frac{1}{2}\left( q + \sqrt{3\left( 1 - p^2 - q^2\right)  }\right)  \right)\,.
\end{equation}
Since the upper and lower bounds match, we conclude that 
\begin{equation}
    \pg \mleft( \vecr, \, \m \mright) = \frac{1}{3} \left( 1 + \frac{1}{2}\left( q + \sqrt{3\left( 1 - p^2 - q^2\right)  }\right)  \right)
\end{equation}
for states $\vecr$ outside $\Delta$ such that $\vecr \cdot \vec{\omega}_{12} > \vecr \cdot \vec{\omega}$ for all $\vec{\omega} \in \{\vec{\omega}\}$. Equivalent solutions can be found in the other two cases by swapping the vectors $\{\vecr_j\}$. Note that when the state $\vecr$ lies in the plane of the trine vectors, the decomposition \eqref{eqn: t vertices} reduces to that in \cite[Appendix A]{Avesani_2022}.

\subsubsection{Four outcomes}\label{ssec: 4 outcomes any}
We choose to work, without loss of generality, in the region of the Bloch sphere defined by
\begin{equation}\label{eqn: biggest n}
    \vecr \cdot \vec{n}_{\left(123 \right)} \geq \vecr \cdot \vec{n} \qquad  \textnormal{for all} \;\; \vec{n} \in \{ \vec{n}\}\,.
\end{equation}
Recall the definition of $\vec{n}_{\left( 123\right)}$ in \eqref{eqn: normal vecs} and \eqref{eqn: l n bary}. We are most interested in the part of this region that is outside the polytope $\mathcal{P}$. Note that, since the normal vectors intersect the faces of the polytope at a distance $l$ from the origin, any state that lives outside the polytope must satisfy $\vecr \cdot \vec{n} > l$ for at least one $\vec{n} \in \{ \vec{n}\}$.  For any state that lives outside $\mathcal{P}$ and satisfies \eqref{eqn: biggest n}, then, we have
\begin{equation}
    p > l\,, \qquad p = \vecr \cdot \vec{n}_{\left( 123 \right)}\,.
\end{equation}
In the region where $p \geq l$ and where \eqref{eqn: biggest n} holds, we can contract and translate the face $\left( 123\right)$ such that it lies in a plane normal to $\vec{n}_{\left(123 \right)}$ at a distance $p$ from the origin, with all three vertices still touching the surface of the Bloch sphere. For a given $p$, the transformed vectors are
\begin{equation}
    \vec{s}_j = \big( p -kl \big)\vec{n}_{\left( 123 \right)} + k \vecr_j\,, \qquad k = \sqrt{\frac{1-p^2}{1-l^2} }\, \qquad \textnormal{for all} \;\; j \in \left\{  1, 2, 3\right\}\,.
\end{equation}
We have $k \leq 1$ because $p \geq l$, so $p-kl \geq 0$.
We can check that the vectors are pure, 
\begin{equation}
    \abs{ s_j}^2 = \left( p-kl\right)^2 + 2k \left( p-kl\right) \vec{n}_{\left( 123 \right)} \cdot \vecr_j + k^2 = p^2 + \left( 1-l^2\right) k^2 =1 \qquad \textnormal{for all} \;\; j \in \{1, \, 2, \, 3\}\,,
\end{equation}
and we know that they live in a plane normal to $\vec{n}_{\left(123 \right)}$ because 
\begin{equation}
  \vec{n}_{\left( 123\right)} \cdot \vec{s}_j = p \qquad \text{for all} \;\; j \in \{1, \, 2, \, 3 \}\,.  
\end{equation} 
We can construct a shape, which we call $\Delta_{\left( 123 \right)}$, defined by 
\begin{equation}
    \Delta_{\left( 123\right)} = \conv \mleft( \{\vec{s}_j\}   \mright) \quad \text{for all } \quad l \leq p \leq 1\,. 
\end{equation}
We will solve the intrinsic randomness of any state $\vecr$ satisfying \eqref{eqn: biggest n} in two parts. 

\paragraph{States inside $\Delta_{\left(123 \right)}$.} If a state $\vecr$ lives inside $\Delta_{\left( 123\right)}$ then by definition it can written as a convex combination of the vectors $\{\vec{s}_j\}$, i.e.
\begin{equation}
    \vecr = \sum_{j=1}^{3} p_j \vec{s}_j
\end{equation}
for some probability distribution $\{p_j\}$. Then we can lower bound the guessing probability by
\begin{align}
 \pg \mleft( \vecr, \, \m \mright) \leq \frac{1}{4} \sum_{j=1}^{3} p_j \left( 1 + \vec{s}_j \cdot \vecr_j \right) =  \frac{1}{4} \left( 1 + pl + \sqrt{1-p^2} \sqrt{1-l^2}  \right)\,.
\end{align}
We can bound the guessing probability in the opposite direction using the dual variable
\begin{equation}
    X = \frac{1}{4}\left( 1 + \frac{1}{k} \right) \Big( \id - \frac{p-kl}{1+k} \, \vec{n}_{\left( 123 \right)} \cdot \vec{\sigma} \Big)\,.
\end{equation}
From \eqref{eqn: rand dual}, the constraints on $X$ are 
\begin{equation}
    X - M_j = \frac{1}{4k} \Big( \id -  \left( \big( p-kl\big)\vec{n}_{\left(123 \right)}+k\vecr_j \right) \cdot \vec{\sigma} \Big) \sgeq 0 \qquad \text{for all} \;\; j \in \{1, \, \hdots ,\, 4\}\,.
\end{equation}
For $j \in \{1, \, 2, \, 3\}$, we have
 \begin{equation}
      X - M_j = \frac{1}{4k} \Big( \id -  \vec{s}_j \cdot \vec{\sigma} \Big) \sgeq 0\,
 \end{equation}
 so we just need to show that 
\begin{equation}
    \abs{ \big( p-kl\big) \vec{n}_{\left( 123\right)} +k \vecr_4 }^2 \leq 1\,.
\end{equation}
Noting that $\vec{n}_{\left(123 \right)} \cdot \vecr_4 = -3l$, we find
\begin{equation}
 \abs{ \big( p-kl\big) \vec{n} +k \vecr_4 }^2 = \abs{ \big( p-kl\big) \vec{n} +k \vecr_1 }^2 - 8 kl \left( p-kl\right) = 1    - 8 kl \left( p-kl\right) \leq 1\,.
\end{equation}
The variable $X$ is then valid, so 
\begin{equation}
\begin{aligned}
    \pg \mleft( \vecr, \, \m \mright) & \leq \tr \mleft( X \rho \mright) = \frac{1}{4}\left( 1 + \frac{1}{k} \right) \left( 1 - \frac{p-kl}{1+k} \vec{n}_{\left( 123 \right)} \cdot \vecr \right)
    \\[3pt]
    &= \frac{1}{4} \left( 1 + pl + \sqrt{1-p^2} \sqrt{1-l^2} \right)\,. 
\end{aligned}
\end{equation}
Since the upper and lower bounds match, we have proven Theorem \ref{thm: inside shape} from the main text.

\paragraph{States outside $\Delta_{\left(123 \right)}$.} It will be helpful to introduce three more unit vectors, 
\begin{equation}
    \vec{\omega}_{12 | 3} = \frac{\left( \vecr_2 -  \vecr_1\right) \times \vec{n}_{\left( 123\right)}}{\abs{\vecr_2 - \vecr_1}}\,, \qquad \vec{\omega}_{13 | 2} = \frac{\left( \vecr_1 -  \vecr_3\right) \times \vec{n}_{\left( 123\right)}}{\abs{\vecr_1 - \vecr_3}}\,,
    \qquad 
    \vec{\omega}_{23 | 1} = \frac{\left( \vecr_3 -  \vecr_2\right) \times \vec{n}_{\left( 123\right)}}{\abs{\vecr_3 - \vecr_2}}\,,
\end{equation}
with the shorthand $\{\vec{\omega}\}= \{\vec{\omega}_{12 | 3}, \, \vec{\omega}_{13 | 2}, \, \vec{\omega}_{23 | 1}  \}$.
Note that $\vec{\omega}_{12 | 3}$, $\vec{\omega}_{13 | 2}$ and $\vec{\omega}_{23 | 1}$ are perpendicular to $\left(\vec{s}_2 - \vec{s}_1\right)$, $\left(\vec{s}_1- \vec{s}_3\right)$ and $\left(\vec{s}_3 - \vec{s}_2\right)$ respectively, and that all $\vec{w} \in \{\vec{\omega}\}$ are perpendicular to $\vec{n}_{\left( 123\right)}$. Since all of these unit vectors point out of the Bloch sphere, any state that satisfies \eqref{eqn: biggest n} and does not live inside $\mathcal{P}$ or $\Delta_{\left(123 \right)}$ must fulfil exactly one of the following inequalities,
\begin{align}
    \vecr \cdot \vec{\omega}_{12 | 3} & > \vec{s}_1 \cdot  \vec{\omega}_{12 | 3} = \vec{s}_2 \cdot  \vec{\omega}_{12 | 3}\,, \label{eqn: q eq}
    \\
    \vecr \cdot \vec{\omega}_{13 | 2} & > \vec{s}_1 \cdot  \vec{\omega}_{13 | 2} = \vec{s}_3 \cdot  \vec{\omega}_{13 | 2}\,,
    \\
    \vecr \cdot \vec{\omega}_{23 | 1} & > \vec{s}_2 \cdot  \vec{\omega}_{23 | 1} = \vec{s}_3 \cdot  \vec{\omega}_{23 | 1}\,.
\end{align}
Note that no state can satisfy more than one of these conditions, since the vertices $\{\vec{s}_j\}$ live on the surface of the Bloch sphere. Without loss of generality, we will work in the region where the condition \eqref{eqn: q eq} holds, so
\begin{equation}
    q > km\,, \qquad q = \vecr \cdot \vec{\omega}_{12 | 3}\,, \qquad m=\vecr_1 \cdot  \vec{\omega}_{12 | 3} = \vecr_2 \cdot  \vec{\omega}_{12 | 3}\,.
\end{equation}
Using \eqref{eqn: normal vecs} and the cross product identity 
\begin{equation}
    \vec{a} \times \left(\vec{b} \times \vec{c }\right) = \vec{b} \left(\vec{a} \cdot \vec{c} \, \right) - \vec{c} \left(\vec{a} \cdot \vec{b} \right)\,,
\end{equation}
we can write $\vec{\omega}_{12 | 3}$ as
\begin{equation}\label{eqn: omega bary}
\begin{aligned}
    \vec{\omega}_{12 | 3} &= \frac{1}{2 A \sqrt{2 \left( 1-a\right)}}\Big( \left( \vecr_2 - \vecr_1 \right) \times \left( \left( \vecr_1 - \vecr_3 \right) \times \left( \vecr_2 - \vecr_3 \right)    \right)   \Big)
    \\
    &= \frac{1}{A \sqrt{2 \left(1 - a \right)}} \Big( \left( 1 + b  \right) \vecr_1 + \left( 1 + c\right) \vecr_2 - \left( 1-a \right) \vecr_3 \Big)\,.
\end{aligned}    
\end{equation}
We can then define
\begin{equation}\label{eqn: m def}
    m = \vecr_1 \cdot \vec{\omega}_{12 | 3} = \vecr_2 \cdot \vec{\omega}_{12 | 3} = \frac{1-a^2}{A \sqrt{2  \left( 1 - a\right) }}\,,
\end{equation}
such that, from \eqref{eqn: l squared ito a}, we have
\begin{equation}\label{eqn: m squared}
  m^2 =  \frac{1-a^2}{A^2} \frac{1+a}{2}  =  \frac{1 + a}{2} -  l^2\,.
\end{equation}
Comparing \eqref{eqn: m def} to the formulation \eqref{eqn: l n bary} for $\vec{n}_{\left( 123 \right)}$, we find the useful formula 
\begin{equation}\label{eqn: r1 and r2}
    \frac{1}{2} \Big( \vecr_1 + \vecr_2\Big) = l \vec{n}_{\left( 123\right)} + m \vec{w}_{12 |3}\,,
\end{equation}
where we use \eqref{eqn: Asquared}. Let's briefly return to the condition \eqref{eqn: biggest n} that we imposed at the start of this section. Expanding $\vecr$ in terms of $\vec{n}_{\left( 123\right)}$ and $\vec{\omega}_{12 | 3}$, it implies, in particular, that
\begin{equation}\label{eqn: ns ineq}
 p =   \vec{n}_{\left( 123\right)} \cdot \vecr  \geq \vec{n}_{\left(124 \right)} \cdot \vecr = p \vec{n}_{\left( 123\right)} \cdot \vec{n}_{\left( 124\right)} + q \vec{\omega}_{12 |3} \cdot \vec{n}_{\left(124 \right)} \,.
\end{equation}
Using \eqref{eqn: r1 and r2} and \eqref{eqn: nprod 2}, we have 
\begin{equation}\label{eqn: next omega}
\vec{\omega}_{12 |3} \cdot \vec{n}_{\left(124 \right)} = \frac{1}{m} \left( l - l \vec{n}_{\left( 123\right)} \cdot \vec{n}_{\left( 124\right)}\right) = \frac{2l}{m} \left( 1 - \frac{2l^2}{1+a}\right)\,.  
\end{equation}
From \eqref{eqn: ns ineq}, using \eqref{eqn: next omega} and \eqref{eqn: nprod 2} again, we find
\begin{equation}
   \left( pm - lq\right) \left(1 + a - 2l^2 \right) \geq 0\,.
\end{equation}
Finally, using \eqref{eqn: m squared}, we conclude that
\begin{equation}\label{eqn: pm lq}
    pm \geq lq\,.
\end{equation}
We are now ready to define the optimal decomposition for $\vecr$. We will contract and translate the line $\conv \mleft( \vec{s}_1, \, \vec{s}_2 \mright)$ so that the new line remains at a distance $p$ along $\vec{n}_{\left( 123 \right)}$ but is now also at a distance $q$ along $\vec{\omega}_{\left( 123\right)}$, with both vertices still touching the Bloch sphere. The new vertices are 
\begin{equation}
    \vec{t}_j  = (p-fl) \vec{n}_{\left(123 \right)} + (q-fm) \vec{w}_{12 | 3} + f \vec{r_j} \qquad \textnormal{for all} \;\; j \in \{1, \, 2\} \,,
\end{equation}
where 
\begin{equation}
    f = \sqrt{ \frac{ 1 - p^2 -q^2  }{ 1 - l^2 -m^2 } }\,.
\end{equation}
We can check that this is correct by noting that $\abs{\vec{t}_j}^2=1$ and
\begin{equation}
    \vec{t}_j \cdot \vec{n}_{\left(123 \right)} = p\,, \qquad \vec{t}_j \cdot \vec{w}_{12 | 3}= q \qquad \text{for all} \;\; j \in \{1, \, 2\}\,.
\end{equation}
Using $q \geq km$, we find
\begin{equation}
    f^2 = \frac{1-p^2-q^2}{1-l^2-m^2} \leq \frac{1-p^2}{1-l^2} = k^2 < 1\,, 
\end{equation}
so $f \leq k < 1$, such that $p -fl$ and $q-fm$ are both non-negative. We can decompose $\vecr$ into $\{\vec{t}_j\}$, choosing the appropriate probability distribution $\{p_j\}$. This gives a guessing probability of 
\begin{equation}
\begin{aligned}
     \pg \mleft( \vecr, \, \m \mright)& \leq \frac{1}{4} \sum_{j=1}^{2} p_j \left( 1 + \vec{t}_j \cdot \vecr_j \right) = \frac{1}{4} \left(1 + pl + qm + \sqrt{1-p^2-q^2} \sqrt{1-l^2 -m^2}  \right)\,.
     \end{aligned}
\end{equation}
The corresponding dual variable is
\begin{equation}\label{eqn: X dual qub}
    X = \frac{1}{4}\left( 1 + \frac{1}{f} \right) \Big( \id - \frac{1}{1+f} \, \vec{v} \cdot \vec{\sigma} \Big)\,, \qquad \vec{v} = (p-fl) \vec{n}_{\left( 123 \right)} + (q-fm) \vec{w}_{12 | 3}\,.
\end{equation}
To prove that it's valid, we need to show that
\begin{equation}
    X - M_j = \frac{1}{4f} \bigg( \id - \big( \vec{v} + f \vecr_j  \big) \cdot \vec{\sigma} \bigg) \geq 0 \qquad \text{for all} \;\; j \in \{1, \, \hdots, \, 4\}\,.
\end{equation}
Equivalently, we can show that
\begin{equation}
    \abs{\vec{v} + f \vecr_j}^2   \leq 1 \qquad \text{for all} \;\; j \in \{1, \hdots, \, 4\}\,,
\end{equation}
where
\begin{equation}
    \abs{\vec{v} + f \vecr_j}^2 =  \big( p-fl \big)^2 + \big( q-fm\big)^2 +f^2 + 2f \left( \big( p-fl \big)\vec{n}_{\left( 123\right)} +  \big( q-fm \big) \vec{w}_{12 | 3} \right) \cdot \vecr_j \,.
\end{equation}
When $j \in \{1, \, 2\}$, we have 
\begin{equation}
    \abs{\vec{v} + f \vecr_j}^2 = \abs{\vec{t}_j}^2 =1\,.
\end{equation}
This means for the remaining two POVM elements, we only need to show that 
\begin{equation}
    \abs{\vec{v} + f \vecr_j }^2 = 1 - 2f \left( \big( p-fl \big)\vec{n}_{\left( 123\right)} +  \big( q-fm \big) \vec{w}_{12 | 3} \right) \cdot \left( \vecr_1 - \vecr_j \right) \leq 1\,,
\end{equation}
or
\begin{equation}
    \left( p - fl\right) \vec{n}_{\left( 123 \right)} \cdot \vecr_j + \left( q-fm\right) \vec{w}_{12 | 3} \cdot \vecr_j \leq \left( p - fl\right) l + \left( q-fm\right) m 
\end{equation}
for $ j \in \{3, \, 4\}$.
The case where $j=3$ is straightforward, because $\vec{n}_{\left( 123\right)} \cdot \vecr_3=l$ and $\vec{w}_{12 |3} \cdot \vecr_3$ is negative, since, from \eqref{eqn: omega bary},
\begin{equation}
  \vec{w}_{12 |3} \cdot \vecr_3 =  \frac{1}{2 A \sqrt{2 \left( 1-a\right)}} \left( b^2 + c^2 -2\right) < 0\,.
\end{equation}
When $j=4$, combining \eqref{eqn: r1 and r2} and \eqref{eqn: A^2 upside down}, we find
\begin{equation}
    \vec{\omega}_{12 | 3} = \frac{1}{m} \left( - \frac{1+a}{2} + 3 l ^2 \right) = \frac{1}{m} \left( 2 l^2- m^2\right)\,,
\end{equation}
so the condition is
\begin{equation}
    -3l\left( p - fl\right)  + \left( q-fm\right) \frac{2 l^2-m^2}{m} \leq \left( p - fl\right) l + \left( q-fm\right) m \,.
\end{equation}
Multiplying both sides by $m$ and rearranging, we have
\begin{equation}
    \left( q - fm \right) \left( l^2 - m^2 \right) \leq 2 ml \left( p - fl \right)\,.
\end{equation}
If $l \leq m $, the inequality is easily shown, since the left-hand side is negative. We focus on the case where $l \geq m $, and rearrange the inequality to get 
\begin{equation}\label{eqn: ineq final X}
    q \left( l^2 - m^2\right) \leq 2 mlp - fm \left(l^2 + m^2 \right)\,.
\end{equation}
Using \eqref{eqn: pm lq}, we can upper bound the left-hand side by 
\begin{equation}
 q \left( l^2 -m^2  \right) \leq \frac{mp}{l} \left( l^2 -m^2  \right) \leq m \left( pl - m^2 \right)\,, 
\end{equation}
while we can lower bound the right-hand side by 
\begin{equation}
   2 mlp - fm \left(l^2 + m^2 \right) \geq 2 mlp - m \left(l^2 + m^2 \right) =  mlp + ml \left( p -l\right) - m^3 \geq mlp - m^3\,,  
\end{equation}
which together prove the inequality \eqref{eqn: ineq final X}. The variable $X$ from \eqref{eqn: X dual qub} is then valid. It gives a guessing probability of 
\begin{equation}
\begin{aligned}
    \pg \mleft( \vecr, \, \m\mright) &\leq \tr \mleft( X \rho\mright) =\frac{1}{4} \left( 1 + \frac{1}{f} \right) \left(1 - \frac{1}{1+f} \vec{v} \cdot \vecr \right)
    \\[0.3pt]
    &= \frac{1}{4} \left( 1 + lp + mq + \sqrt{1-p^2 -q^2} \sqrt{1-l^2 - m^2} \right)\,.
    \end{aligned}
\end{equation}
Combining the upper and lower bounds, we have proven that  
\begin{equation}
\begin{aligned}
    \pg \mleft( \vecr, \, \m\mright) =  \frac{1}{4} \left( 1 + lp + mq + \sqrt{1-p^2 -q^2} \sqrt{1-l^2 - m^2} \right)\,
    \end{aligned}
\end{equation}
for any state satisfying \eqref{eqn: biggest n} such that $\vecr \notin \Delta_{\left( 123\right)}$.

\section{Scissors and SIC measurements}
\subsection{Qubits}\label{ssec: scissors and SIC}
Let's now return to our parameterization of unbiased qubit MIC POVMs from Appendix \ref{sssec: two angles}, with measurement vectors
\begin{equation}\label{eqn: unbiased all vecs}
    \begin{aligned}
        \vecr_1 &= \left( - \cos \delta , \, -\sin \delta, \, 0\right)
        \\
        \vecr_2 &= \left( - \cos \delta , \, \sin \delta, \, 0\right)\,,
        \\
        \vecr_3 &= \left( \cos \delta , \, - \sin \delta \sin \alpha, \, \sin \delta \cos \alpha\right)
        \\
        \vecr_4 &= \left( \cos \delta , \,  \sin \delta \sin \alpha, \, -\sin \delta \cos \alpha\right)\,.
    \end{aligned}
\end{equation}
The vectors overlaps from \eqref{eqn: a b c} are
\begin{equation}
    \begin{aligned}
        a &= \vecr_1 \cdot \vecr_2 = \cos^2 \delta - \sin^2 \delta\,,
        \\
        b &= \vecr_1 \cdot \vecr_3 = - \cos^2 \delta + \sin^2 \delta \sin \alpha\,,
        \\
        c &= \vecr_2 \cdot \vecr_3 =- \cos^2 \delta - \sin^2 \delta \sin \alpha\,.
    \end{aligned}
\end{equation}
The squared area of each of the faces of the tetrahedron, from \eqref{eqn: Asquared}, is then
\begin{equation}\label{eqn: A squared parameterized}
    A^2 = \sin^2 \delta \left( 4 \cos^2 \delta + \sin^2 \delta \cos^2 \alpha \right)\,.
\end{equation}
The perpendicular length $l$ from the origin to each of the tetrahedron faces, from \eqref{eqn: l squared}, is then given by
\begin{equation}\label{eqn: l squared parameterized}
    l^2 = \frac{\cos^2 \delta \sin^4 \delta \cos^2 \alpha}{A^2} =  \frac{\cos^2 \delta \sin^2 \delta \cos^2 \alpha}{4 \cos^2 \delta + \sin^2 \delta \cos^2 \alpha}\,.
\end{equation}
Finally, from \eqref{eqn: vol after} the volume of the space of average post-measurement states, when represented in baryonic coordinates using the vectors $\{\vecr_j\}$, is
\begin{equation}\label{eqn: vol parameterized}
    \textnormal{Vol} = \frac{\pi}{3} \cos^2 \delta \sin^4 \delta \cos^2 \alpha\,.
\end{equation}
 We can find the parameters $\delta$ and $\alpha$ that \emph{maximise} these three quantities by solving the derivatives
\begin{equation}
    \begin{aligned}
        \frac{\partial A^2}{\partial \delta} &= 4 \cos \delta \sin \delta \left( 2 - 4 \sin^2 \delta + \sin^2 \delta \cos^2 \alpha\right)\,, \qquad  &\frac{\partial A^2}{\partial \alpha} &=  - 2 \sin^4 \delta \cos \alpha \sin \alpha\,,
        \\[5pt]
        \frac{\partial l^2}{\partial \delta} &= \frac{2 \sin^4 \delta \cos^2 \alpha \left(4 \cos^4 \delta - \sin^4 \delta \cos^2 \alpha \right)}{A^4}\,, \qquad  &\frac{\partial l^2}{\partial \alpha}  & = - \frac{8 \sin^6 \delta \cos^4 \delta \sin \alpha \cos \alpha}{A^4}\,, 
        \\[5pt]
         \frac{\partial \textnormal{Vol}}{\partial \delta} &= 2 \sin^3 \delta \cos \delta \cos^2 \alpha\left( 2 \cos^2 \delta - \sin^2 \delta\right)\,, \qquad &\frac{\partial \textnormal{Vol}}{\partial \alpha} &= - 2 \sin^4 \delta \cos^2 \delta \sin \alpha \cos \alpha\,.
    \end{aligned}
\end{equation}
The unique critical point in each case occurs when $\alpha=0$ and $\delta= \arccos{\frac{1}{\sqrt{3}}}$, which defines the tetrahedron SIC POVM. (We can check that the Hessian of each quantity is negative definite at this point, so this point is the local maximum in each case.) The corresponding values for the SIC are
\begin{equation}
    A^2 = \frac{4}{3}\,, \qquad l^2 = \frac{1}{9}\,, \qquad \textnormal{Vol} = \frac{4 \pi}{81}\,.
\end{equation}

We can also consider a very different kind of measurement which \emph{minimises} these geometric quantities. Let's define the scissors measurement in dimension two, which was previously introduced in \cite[Appendix C.C]{Woodhead_2020}, by the POVM \eqref{eqn: unbiased all vecs} where $\alpha=0$. We are particularly interested in the case where $\delta$ is arbitrarily close to zero, as then the squared area \eqref{eqn: A squared parameterized}, the squared length \eqref{eqn: l squared parameterized} and the volume \eqref{eqn: vol parameterized} all tend to zero as well. Note that $\delta$ cannot be zero, as the POVM would not then be extremal, but it can be arbitrarily close, such that we can achieve arbitrarily close to perfect randomness. Does there exist such a one-parameter family of measurements in higher dimensions that include a SIC POVM and achieve arbitrarily high randomness close to a limiting POVM which is non-extremal? In the following, we present two such `scissors-like' families in dimensions three and four, leaving other dimensions for future study. 

\subsection{Qutrits and ququarts}\label{ssec: qudit scissors}
The Weyl-Heisenberg group in dimension three is given by the ladder operators 
\begin{equation}
    D_{jk} = \eta^{\frac{jk}{2}} \sum_{m=0}^{2} \eta^{jm} \ketbra{k + m}{m}\,, \qquad \eta= e^{i \frac{2 \pi}{3}}\,,
\end{equation}
where addition is modulo $3$. 
From \cite{D_Ariano_2004}, a POVM $\m = \{M_{jk}\}_{jk}$ generated from a fiducial state $\ket{\psi}$, with elements
\begin{equation}
    M_{jk} = \frac{1}{3} \ketbra{\psi_{jk}}{\psi_{jk}}\,, \qquad \ket{\psi_{jk}} = D_{jk} \ket{\psi} 
\end{equation}
will be informationally complete as long as long as the following condition is satisfied, 
\begin{equation}
    \langle  \psi | D_{jk} | \psi \rangle \neq 0 \qquad \textnormal{for all} \;\; j, \, k \in \{0, \, 1, \, 2\}\,.
\end{equation}
Let's take qutrit fiducial state 
\begin{equation}
\ket{\psi} = \frac{1}{2\sqrt{3}} \left(  \left( 2 \cos \delta + \sqrt{2} \sin \delta\right) \left(\ket{0} + \ket{1} \right) + \left( 2 \cos \delta -2 \sqrt{2} \sin \delta\right) \ket{2}    \right)\,.    
\end{equation}
For this $\ket{\psi}$, we have
\begin{equation}
\begin{aligned}
 \langle  \psi | D_{jk} | \psi \rangle = \frac{1}{2} \left( \left( 1 + \cos^2 \delta + 2 \sqrt{2} \cos \delta \sin \delta \right) \delta_{j, 0} + 3\,  \omega^{2 j} \, \delta_{k, 0} \, \sin^2 \delta \right)  
 \\
 - \omega^{2 j} \cos \mleft(  \frac{\pi j k}{3}\mright) \sin \delta \left( \sqrt{2} \cos \delta + \sin \delta \right)\,,
 \end{aligned}
\end{equation}
such that 
\begin{equation}
\langle  \psi | D_{jk} | \psi \rangle = 
\begin{cases}
1 \qquad & j, \,  k =0\,,
\\[5pt]
- \frac{1}{2} \left( 1 - 3 \cos^2 \delta \right) \qquad & j=0, \, \, k \in \{1, 2\}\,,
\\[5pt]
- \frac{1}{2} \, \omega^{2j} \sin \delta \left( 2 \sqrt{2} \cos \delta - \sin \delta   \right) \qquad &j \in \{1, 2\},\,\, k =0 \,,
\\[5pt]
- \omega^{2 j} \cos \mleft(  \frac{\pi j k}{3}\mright) \sin \delta \left( \sqrt{2} \cos \delta + \sin \delta \right) \qquad & j, \, k \in \{1, 2\}\,.
\end{cases}    
\end{equation}
All of these overlaps are nonzero in the range $0 < \delta < \arccos \frac{1}{\sqrt{3}}$, so we define our family of MIC POVMs in this range. When $\cos \delta = \sqrt{\frac{2}{3}}$, we recover a SIC POVM (see \cite[Section IV B]{Renes_2004}). 
The probability distribution if we measure on $\ket{0}$ is 
\begin{equation}
    p_{jk} = \frac{1}{3}\abs{ \langle 3 - k | \psi \rangle  }^2 = \frac{1}{36}\begin{cases}
        \left( 2 \cos \delta + \sqrt{2}\sin \delta \right)^2\,, \qquad & k=0\,,
        \\[5pt]
        \left( 2 \cos \delta - 2\sqrt{2}\sin \delta \right)^2\,, \qquad & k=1\,,
        \\[5pt]
        \left( 2 \cos \delta + \sqrt{2}\sin \delta \right)^2\,, \qquad & k=2\,,
    \end{cases}
\end{equation}
which tends to the uniform distribution in the limit $\delta \rightarrow 0$.

In dimension four, while we could also define an appropriate family of fiducial states for the Weyl-Heisenberg basis, it is easier to directly generalise the SIC POVM given in \cite[Chapter 7]{belovs2008} and \cite[Section 6]{Appleby_2012}. Consider a POVM with 16 elements $\m= \{\frac{1}{4} \ketbra{\psi_j}{\psi_j}\}_j$, where the states are given in row form by
\begin{equation}
    \begin{aligned}
    \ket{\psi_1} &= \frac{1}{\sqrt{3 + a^2}} \left( a, \, 1, \, 1, \, 1 \right)\,, \qquad &&\ket{\psi_2} = \frac{1}{\sqrt{3 + a^2}} \left(a, \, -1, \, 1, \, -1 \right)   \,,
    \\
    \ket{\psi_3} &= \frac{1}{\sqrt{3 + a^2}} \left( a, \, 1, \, -1, \, -1 \right)\,, \qquad &&\ket{\psi_4} = \frac{1}{\sqrt{3 + a^2}} \left(a, \, -1, \, -1, \, 1 \right)   \,,
    \\
    \ket{\psi_5} &= \frac{1}{\sqrt{3 + a^2}} \left( 1, \,  i a, \, -i, \, 1 \right)\,, \qquad &&\ket{\psi_6} = \frac{1}{\sqrt{3 + a^2}} \left(1, \, - i a, \, -i, \, -1 \right)   \,,
    \\
    \ket{\psi_7} &= \frac{1}{\sqrt{3 + a^2}} \left( 1, \, ia, \, i, \, -1 \right)\,, \qquad &&\ket{\psi_8} = \frac{1}{\sqrt{3 + a^2}} \left(1, \, - i a, \, i, \, 1 \right)   \,,
    \\
    \ket{\psi_9} &= \frac{1}{\sqrt{3 + a^2}} \left( 1, \,  1, \, -i a, \, i \right)\,, \qquad &&\ket{\psi_{10}} = \frac{1}{\sqrt{3 + a^2}} \left(1, \, - 1, \, -i a, \, - i \right)   \,,
    \\
    \ket{\psi_{11}} &= \frac{1}{\sqrt{3 + a^2}} \left( 1, \, 1, \, i a, \, - i \right)\,, \qquad &&\ket{\psi_{12}} = \frac{1}{\sqrt{3 + a^2}} \left(1, \, - 1 , \, i a, \, i \right)   \,,
     \\
    \ket{\psi_{13}} &= \frac{1}{\sqrt{3 + a^2}} \left( 1, \,  i , \, 1, \, -i a \right)\,, \qquad &&\ket{\psi_{14}} = \frac{1}{\sqrt{3 + a^2}} \left(1, \, - i , \, 1, \, ia \right)   \,,
    \\
    \ket{\psi_{15}} &= \frac{1}{\sqrt{3 + a^2}} \left( 1, \, i, \, -1, \, ia \right)\,, \qquad &&\ket{\psi_{16}} = \frac{1}{\sqrt{3 + a^2}} \left(1, \, - i , \, -1, \, -ia \right)   \,,
    \end{aligned}
\end{equation}
for some parameter $a > 1$. 
When $a = \sqrt{2 + \sqrt{5}}$, we return the SIC POVM of \cite{belovs2008, Appleby_2012}. It is straightforward to see that the vectors form a POVM. To show that this POVM is informationally complete, we can form a $16 \times 16$ matrix by flattening the projectors $\{\ketbra{\psi_j}{\psi_j}\}$ to fill the rows. Up to a constant, this matrix is 
\begin{equation} 
    \begin{pmatrix}
        a^2 & a & a & a &  & a & 1 & 1 & 1 && a & 1 & 1 & 1&&  a & 1 & 1 & 1 
        \\
        a^2 & -a & a & -a &&  -a & 1 & -1 & 1  && a & -1 & 1 & -1  && -a & 1 & -1 & 1
        \\
        a^2 & a & -a & -a && a & 1 & -1 & -1 && -a & -1 & 1 & 1 && -a & -1 & 1 & 1
        \\
        a^2 & -a & -a & a && -a & 1 & 1 & -1 && -a & 1 & 1 & -1& & a & -1 & -1 & 1
        \\
        \\
        1 & -i a & i & 1 && ia & a^2 & -a & ia && -i & -a & 1 & -i && 1 & -ia & i & 1
        \\
        1  & ia & i & -1 && -ia & a^2 & a & ia && -i & a & 1 & i && -1 & -ia & -i & 1
        \\
        1 & -ia & -i & -1 && ia & a^2 & a& -ia && i & a & 1 & -i && -1 & ia & i & 1
        \\
        1 & ia & -i & 1 && -ia & a^2 & -a & -ia && i & -a & 1 & i && 1 & ia & -i & 1
        \\
        \\
        1 & 1& ia & -i && 1 & 1 & ia & -i && -ia & -ia & a^2 & -a && i & i & -a & 1
        \\
        1 & -1 & ia & i && -1 & 1 & -ia & -i && -ia & ia & a^2 & a&& -i & i & a & 1
        \\
        1 & 1 & -ia & i && 1 & 1 & -ia & i && ia & ia & a^2 & -a && -i & -i & -a & 1
        \\
        1 & -1 & -ia & -i && -1 & 1 & ia & i && ia & -ia & a^2 & a && i & -i & a & 1
        \\
        \\
        1 & -i & 1 & ia && i & 1 & i & -a && 1 & -i & 1 & ia && -ia & -a & -ia & a^2
        \\
        1 & i & 1 & -ia && -i & 1 & -i & -a && 1 & i & 1 & -ia && ia & -a & ia & a^2
        \\
        1 & -i & -1 & -ia && i & 1 & -i & a && -1 & i & 1 & ia && ia & a & -ia & a^2
        \\
        1 & i & -1 & ia && -i & 1 & i & a && -1 & -i & 1 & -ia && -ia & a & ia & a^2
    \end{pmatrix}\,.
\end{equation}
Its determinant is
\begin{equation}
    \textnormal{det}  \propto \left( a^2 - 1\right)^3 \left( 1 + a^2 \right)^{6} \left( 3 + a^2\right)\,, 
\end{equation}
which is nonzero in our range, so we have defined a family of extremal POVMs. Interestingly, there does not seem to be an upper limit for $a$, and in the limit as $a \rightarrow \infty$ it tends to a non-extremal measurement. The limiting POVM, for which $a = 1$, is also non-extremal. If we measure the state $\ket{0}$ with a POVM in this family, we find 
\begin{equation}
    p_{j} = \tr \mleft( M_j \ketbra{0}{0} \mright) = \frac{1}{4 \left( 3 + a^2 \right)} 
    \begin{cases}
     a^2\,, \qquad  & j \in \{1, \, \hdots, \, 4\}\,,
     \\
     1\,, \qquad   &j \in \{5, \, \hdots, \, 16\}\,,
    \end{cases}
\end{equation}
which tends to the uniform distribution in the limit $a \rightarrow 1$.

\end{document}